\documentclass[a4paper,10pt]{article}

\usepackage[utf8]{inputenc}

\usepackage{authblk}

\usepackage{ifthen,ifpdf}

\ifpdf
  \usepackage{hyperref}
  \hypersetup{colorlinks=false,allcolors=blue}
  \usepackage{hypcap}
  \pdfpagewidth=\paperwidth
  \pdfpageheight=\paperheight
\fi

\usepackage[utf8]{inputenc}
\usepackage[T1]{fontenc}

\usepackage{amsmath, amsthm, amssymb}
\usepackage{mathtools}
\usepackage{enumerate}
\usepackage{afterpage}
\usepackage{tabularx}
\usepackage{dsfont}

\usepackage{graphicx}

\usepackage{algorithm}
\usepackage{algpseudocode}
\usepackage{caption}
\usepackage[labelformat=simple]{subcaption}

\usepackage{color}
\usepackage{units}
\usepackage{array, multirow}
\usepackage{booktabs}

\newcolumntype{M}[1]{>{\vspace{3pt}\raggedleft\arraybackslash}m{#1}}

\usepackage{siunitx}[=v2]
\usepackage{tikz}
\usetikzlibrary{calc}
\usetikzlibrary{matrix}
\usepackage[export]{adjustbox}
\usepackage{pgfplots}
\usepackage{pgfplotstable}
\usepgfplotslibrary{units}
\pgfplotsset{compat=1.9,samples=1000,every axis/.append style={
font=\large,
line width=1pt,
tick style={line width=0.8pt}}}

\usetikzlibrary{arrows}
\usetikzlibrary{hobby}
\usetikzlibrary{calc}
\usetikzlibrary{arrows}
\usetikzlibrary{shapes}
\usetikzlibrary{positioning}
\usetikzlibrary{intersections}
\usetikzlibrary{decorations.pathreplacing}
\usetikzlibrary{patterns}
\usetikzlibrary{decorations.pathmorphing}
\usetikzlibrary{decorations.markings}
\usetikzlibrary{3d}
\usetikzlibrary{matrix}
\usetikzlibrary{spy}

\usepgfplotslibrary{polar}
\usepgfplotslibrary{external}
\usepgfplotslibrary{fillbetween}
\usepgfplotslibrary{groupplots}
\usepgfplotslibrary{patchplots}

\usepackage{cite}

\usepackage{anysize} 
\marginsize{2.5cm}{2.5cm}{2cm}{2cm}

\usepackage{wrapfig}

\def\XXint#1#2#3{{\setbox0=\hbox{$#1{#2#3}{\int}$}
     \vcenter{\hbox{$#2#3$}}\kern-.5\wd0}}

\theoremstyle{plain}
\newtheorem{thm}{Theorem}[section]
\newtheorem{lem}[thm]{Lemma}
\newtheorem{prop}[thm]{Proposition}

\newtheorem{ass}{Assumption}

\theoremstyle{remark}

\numberwithin{equation}{section}

\usepackage{bm}
\usepackage[bbgreekl]{mathbbol}

\newcommand{\R}{\mathds{R}}

\newcommand{\C}{\mathds{C}}

\newcommand{\N}{\mathds{N}}

\newcommand{\Z}{\mathds{Z}}

\newcommand{\Sym}[1]{\textrm{Sym}(#1)}

\renewcommand{\div}{\textrm{div }}

\newcommand{\review}[1]{{\color{black}{#1}}}

\newcommand{\multiplication}{m}
\newcommand{\Multiplication}{\bm{\tau}}

\DeclareMathOperator{\eps}{\varepsilon}

\newcommand{\sty}[1]{\bm{#1}}

\newcommand{\fb}{\sty{ b}}

\newcommand{\fd}{\sty{ d}}

\newcommand{\fu}{\sty{ u}}
\newcommand{\fv}{\sty{ v}}

\newcommand{\fx}{\sty{ x}}
\newcommand{\fy}{\sty{ y}}

\newcommand{\fG}{\sty{ G}}

\newcommand{\fK}{\sty{ K}}

\newcommand{\fW}{\sty{ W}}

\newcommand{\feps}{\bm{\varepsilon}}
\newcommand{\fsigma}{\bm{\sigma}}
\newcommand{\fxi}{\bm{\xi}}

\newcommand{\ftau}{\bm{\tau}}
\newcommand{\feta}{\mbox{\boldmath $\eta$}}

\newcommand{\ffGamma}{\mathbb{\Gamma}}

\title{Universal Fourier Neural Operators for \\\review{periodic homogenization problems in linear elasticity}}
\author[1]{Binh Huy Nguyen}
\author[1,2,3,*]{Matti Schneider}

\affil[1]{University of Duisburg-Essen, Institute of Engineering Mathematics}
\affil[2]{Center for Nanointegration Duisburg-Essen (CENIDE)}
\affil[3]{Fraunhofer Institute for Industrial Mathematics ITWM, Kaiserslautern}
\affil[*]{correspondence to: \texttt{matti.schneider@uni-due.de}}

\date{\today}

\begin{document}
\maketitle 

\begin{abstract}
\noindent Solving cell problems in homogenization is hard, and available deep-learning frameworks fail to match the speed and generality of traditional computational frameworks. More to the point, it is generally unclear what to expect of machine-learning approaches, let alone single out which approaches are promising. In the work at hand, we advocate Fourier Neural Operators (FNOs) for micromechanics, empowering them by insights from computational micromechanics methods based on the fast Fourier transform (FFT).\\
We construct an FNO surrogate mimicking the basic scheme foundational for FFT-based methods and show that the resulting operator predicts solutions to cell problems with \emph{arbitrary} stiffness distribution only subject to a material-contrast constraint up to a desired accuracy. In particular, there are no restrictions on the material symmetry like isotropy, on the number of phases and on the geometry of the interfaces between materials. Also, the provided fidelity is sharp and uniform, providing explicit guarantees leveraging our physical empowerment of FNOs.\\
To show the desired universal approximation property, we construct an FNO explicitly that requires no training to begin with. Still, the obtained neural operator complies with the same memory requirements as the basic scheme and comes with runtimes proportional to classical FFT solvers. In particular, large-scale problems with more than 100 million voxels are readily handled.\\
The goal of this work is to underline the potential of FNOs for solving micromechanical problems, linking FFT-based methods to FNOs. This connection is expected to provide a fruitful exchange between both worlds.\\

\end{abstract}




\newpage

\section{Introduction}
\label{sec:intro}

\subsection{State of the art}
\label{sec:intro_state}

To design components and structures made of microstructured materials, in particular in case of long-term use or when operating in a nonlinear regime, computational approaches based on the mathematical theory of homogenization~\cite{BlancLeBris2023Book} form an integral part. Such techniques split the physical problem into a \review{macroscopic} problem which involves a \emph{homogenized} model that is computationally tractable and a microscopic model which resolves the microstructure and permits to identify the parameters of the homogenized model appropriately. This splitting approach hides a significant "{}amount"{} of difficulty, encoded by the complexity of the microstructure, in the resolution of the microscopic problem. With solid mechanics applications in mind, these challenges include:
\begin{enumerate}
	\item Discontinuous materials, leading to solution fields with (weak) discontinuities across interfaces,
	\item Complex arrangement of the microstructure phases, in particular for materials used in industry,
	\item Randomness of the microstructure.
\end{enumerate}
Computational multiscale methods have a long history, and we refer to the review article~\cite{MatousSummary} for an in-depth overview. We just summarize some key findings. Due to challenges one and two, traditional interface-conforming finite elements may not be the best choice when optimizing runtime vs. accuracy. Rather, regular-grid based methods~\cite{RegularGrid,RegularGrid2,RegularGrid3} lead the way. When restricted to regular grid, traditional finite element (FE) discretizations and solvers lose their edge~\cite{Schneider2022HGC}, and spectral~\cite{MoulinecSuquet1994,MoulinecSuquet1998,Schneider2023MSconvergence} or finite-difference discretizations~\cite{Willot2015,StaggeredGrid,Finel2025TET} combined with solvers based on the fast Fourier transform (FFT)~\cite{EyreMilton1999,MichelMoulinecSuquet2001,ZemanCG2010,BB2019} give rise to highly efficient options for dealing with micromechanics problems. To handle item 3, the representative volume element (RVE) method is used~\cite{HillRVE,Sab1992,DruganWillis}, which requires solving \emph{multiple} microscopic cell problems on domains of varying size~\cite{KanitRVE,Otto2021}.\\
This work is concerned with neural techniques for solving homogenization problems. The advantages of approaches based on neural networks were demonstrated for high-dimensional and parametric PDEs~\cite{Karniadakis2021PINNOverview,Cuomo2022PINN,Beck2023OverviewDeepLearningPDEs} where they may overcome the "{}curse of dimensionality"{}~\cite{Hutzenthaler2020OvercomingCurseOfDimensionality}. However, for low to moderate dimensions, traditional discretizations and solvers typically outperform neural-network approaches. Rather, deep learning may be used to improve the associated discretizations and solvers~\cite{Azulay2022MultigridDeepLearningHelmholtz,Han2024UGrid,Huang2020IntDeep}.\\
In the context of homogenization problems, the use of neural networks was reported. Due to the considerable size of the domains to be treated, convolutional neural networks are favored and applied to metal-matrix composites~\cite{Rao2020CNN}, porous materials~\cite{Peng2022Phnet} and random checkerboards~\cite{Zhu2024CNN}. However, the trained networks use a fixed resolution, require excessive training and are unlikely to compete with dedicated solvers~\cite{Segurado2020CPReview,LebensohnRolletCPReview,FFTReview2020}. Only under very restrictive circumstances, such techniques may shine. For example, Zheng et al.~\cite{Zheng2023TrussesKochmann} consider the inverse design problem for truss-based cellular metamaterials, and report a significant speed up compared to full-field simulations.\\
To address the shortcomings of traditional deep-learning approaches, the attention shifted from neural networks -- which approximate \emph{functions} -- to neural operators which should approximate \emph{operators}, i.e., mappings between function spaces. From an abstract point of view, such a strategy appears to be more natural when solving ordinary and partial differential equation, as typically the solution is sought in dependence of the loading and the boundary conditions. Such operator-learning frameworks include the Deep-O-Net~\cite{Lu2021DeepONet} and neural operators~\cite{Kovachki2023NO,Cao2024LaplaceNO}. Whereas the former represents operators between infinite-dimensional spaces via collocation, i.e., point evaluation of functions, the latter augments the weights and biases traditionally used in neural network architectures by an integral operator. This non-local term complements the action of the activation function and the weights/biases which is applied locally, i.e., at every continuum point. For reasons of efficiency and practical applicability of neural operator, instead of a general integral operator, a \emph{convolution operator} may be used, whose action is readily computed by the FFT. The so-called Fourier Neural Operators (FNOs)~\cite{Li2021FNO}
appear to balance simplicity with sufficient expressivity, as demonstrated by mathematical approximation results~\cite{Kovachki2021FNOApproximation}. Also, such operator-learning techniques may be infused with physical knowledge, giving rise to so-called physics-informed operator learning~\cite{Wang2021PIDeepONet,Goswami2022PIDeepONet,Li2024PINO}.\\
Recently, FNOs were applied to homogenization problems in conductivity and linear elasticity. Bhattacharya et al.~\cite{Bhattacharya2024LearningHomogenization} consider the task of learning the temperature field solving the cell problem as a function of the microscopic conductivity field. They provide both a theoretical framework and computational examples to demonstrate their findings. Their approximation result ensures that the temperature field may be learned either uniformly on a compact set of coefficients or on average for an ensemble of coefficients. Wang et al.~\cite{Wang2025PretrainingFNO} train an FNO to provide an initial guess for micromechanical cell problems with linear elastic materials to be solved by the finite-element method. 
Harandi et al.~\cite{Harandi2025SPiFOL} use the FFT to train a conventional neural network with a Lippmann-Schwinger based loss function. Contributing to physics-informed neural operators, Zhang-Guilleminot~\cite{zhang2024operator} devise a label-free FNO for homogenizing hyperelastic materials with the stored elastic energy as the loss function. \review{The FNO is also used as an encoder in a transformer-based deep learning model to predict the mechanical response of heterogeneous composites, as demonstrated by Wang et al. \cite{wang2026micrometer}.}

\subsection{Contributions}
\label{sec:intro_contributions}

The work at hand targets the principal difficulties encountered when using conventional deep-learning frameworks for micromechanics problems in elasticity:
\begin{enumerate}
	\item Severe restrictions on the local elastic stiffness tensors, e.g., isotropy.
	\item Grid dependence of the trained network or the ability to solve rather small grid sizes only.
	\item Excessive resources required for the training -- both in terms of time and energy usage -- in particular compared to the online use.
	\item Restrictions to specific microstructures materials, e.g., porous materials or matrix-inclusion composites with specific type of reinforcements.
\end{enumerate}
To address these issues, we consider one of the most efficient frameworks currently available for computational homogenization, FFT-based micromechanics~\cite{Segurado2020CPReview,LebensohnRolletCPReview,FFTReview2020}, and transfer insights from their construction to the design of an FNO.  The main message of the article is the following:
\begin{center}
	\fbox{What can be done with FFT solvers, can also be done with FNOs.}
\end{center}
As FFT-based solvers serve as the current gold standard for computational homogenization, the framework of FNOs has the potential to serve as a real improvement to current computational practice. As will become clear below, we take the simplest FFT solver, the basic scheme~\cite{MoulinecSuquet1994,MoulinecSuquet1998,PorousSLE2020}, and construct a corresponding FNO. However, also alternative linear and nonlinear solvers like polarization schemes~\cite{EyreMilton1999,MichelMoulinecSuquet2001,Donval2024Polarization}, conjugate gradient methods~\cite{ZemanCG2010,BrisardDormieux2010,NonlinearCG2020} or Quasi-Newton approaches~\cite{Chen2019Anderson,QuasiNewton2019,BB2019} could be re-cast as FNOs.\\
As our main result, we establish the following: We fix a material contrast $\kappa$ a priori and construct, for pre-defined error margin $\delta$,  an FNO $\mathcal{F}_\theta$ with the following properties:
\begin{enumerate}
	\item For every (unit) macroscopic strain $\bar{\feps}$ and \emph{every} stiffness distribution $\C$ with a contrast not greater than $\kappa$, the output $\mathcal{F}_\theta(\bar{\feps},\C)$ provides a displacement field which is $\delta$-close to the exact solution. In particular, there are no restrictions other than the contrast bound on the form of the elasticity tensor, the arrangement and the number of phases and the geometry of the interface.
	\item The FNO $\mathcal{F}_\theta$ may be applied to any voxel grid.
	\item There is no training involved in setting up the operator $\mathcal{F}_\theta$.
\end{enumerate}
The first item provides, in particular, a \emph{uniform} approximation results for the solutions of micromechanical problems for general coefficients. This result differs strongly from Bhattacharya et al.~\cite{Bhattacharya2024LearningHomogenization} who provide a result for a compact subset of coefficients only. The third item is not a direct feature, but rather a byproduct of our construction: As the homogenization-learning problem is complicated and non-convex, showing approximation results requires constructing a competitor more or less explicitly. It is expected that actual training would significantly enhance the capabilities of the framework.\\
Our main idea is to interpret the basic scheme of Moulinec-Suquet~\cite{MoulinecSuquet1994,MoulinecSuquet1998} -- which has the form
\begin{equation}\label{eq:intro_contributions_basic}
	\feps^{k+1} = \bar{\feps} - \ffGamma^0 : (\C - \C^0) : \feps^k
\end{equation}
for a reference material $\C^0$ and the Eshelby-Green operator $\ffGamma^0$ -- as a single layer of an FNO. As the basic scheme \eqref{eq:intro_contributions_basic} converges for $k \rightarrow \infty$, stacking such layers should give rise to approximations with increasing fidelity. For the convenience of the reader, the properties of the basic scheme \eqref{eq:intro_contributions_basic} are recalled in section \ref{sec:homogenization}.\\
It appears natural to interpret the action of the basic scheme \eqref{eq:intro_contributions_basic} as an FNO - the operation is affine linear in the strain $\feps^k$, and a single convolution operator is evaluated per iteration \eqref{eq:intro_contributions_basic}. However, there is a caveat. The goal of the sought FNO is to take the local field of stiffness tensors $\C$ as the input and provide the solution of the micromechanics problem as the output. Storing the stiffness $\C$ in a register gives rise to the update
\begin{equation}\label{eq:intro_contributions_basic_agumented}
	\left[
	\begin{array}{l}
		\feps^{k+1}\\
		\C^{k+1}\\
	\end{array}
	\right] = 
	\left[
	\begin{array}{l}
		\bar{\feps} - \ffGamma^0 : (\C^k - \C^0) : \feps^k\\
		\C^k
	\end{array}
	\right],
\end{equation}
which is nonlinear due to the double contraction 
\begin{equation}\label{eq:intro_contributions_learnME}
	(\C^k, \feps^k) \mapsto \C^k : \feps^k.
\end{equation}
\review{To elaborate this statement, traditionally, the stiffness tensor (field) is fixed, and the operation $\feps \mapsto \C : \feps$ is a linear operation. However, when learning homogenization, the stiffness tensor field needs to be treated as a \emph{variable} and no longer as a mere \emph{parameter}. Indeed, in linear elastic homogenization problems, one seeks the strain field $\feps$ as a function of the input stiffness $\C$. Therefore, during the iterative scheme \eqref{eq:intro_contributions_basic_agumented}, both the strain field $\feps$ and the input stiffness field $\C$ need to be treated independently. In our setup \eqref{eq:intro_contributions_basic_agumented} -- which we consider to be minimal -- the components of the stiffness $\C$ and of the strain $\feps$ are treated as the components of a suitably high-dimensional vector field. Then, the operation \eqref{eq:intro_contributions_learnME} involves multiplying components of a specific vector field, which is indeed a nonlinear operation.}\\
\review{Thus, to cope with linear elastic homogenization problems}, the sought FNO needs to \review{learn the nonlinearity \eqref{eq:intro_contributions_learnME}}. In section \ref{sec:FNO}, we provide conditions on the neural-network approximation to the function \eqref{eq:intro_contributions_learnME} which permit us to conclude the desired statement  on the constructed FNO:
\begin{thm}\label{thm:main}
	For every periodic unit cell $Y$, every pair of lower and upper bounds $0<\alpha_- < \alpha_+$ on the material stiffness, every strain bound $\eps_0 > 0$ and every error margin $\delta^{\texttt{target}}>0$, there is an FNO
	\begin{equation}\label{eq:intro_contributions_FNO}
		\mathcal{F}_\theta: \Sym{d} \times \mathcal{M}_Y(\alpha_-,\alpha_+) \rightarrow H^1_\#(Y)^d,
	\end{equation}
	s.t. the estimate 
	\begin{equation}\label{eq:intro_contributions_result}
		\left\| \mathcal{F}_\theta(\bar{\feps},\C) - \mathcal{S}(\bar{\feps},\C)\right\|_{H^1_\#(Y)^d} \leq \delta^{\texttt{target}}
	\end{equation}
	holds for all macroscopic strains $\bar{\feps} \in \Sym{d}$ obeying the bound
	\begin{equation}
		\| \bar{\feps} \| \leq \eps_0
	\end{equation}
	and all stiffness distributions 
	\begin{equation}\label{eq:intro_contributions_MAlphaPM}
		\C \in \mathcal{M}_Y(\alpha_-,\alpha_+),
	\end{equation}
	and where $\mathcal{S}$ refers to the exact solution operator, i.e., for all macroscopic strains $\bar{\feps} \in \Sym{d}$ and all stiffnesses \eqref{eq:intro_contributions_MAlphaPM}, the displacement field $\fu = \mathcal{S}(\bar{\feps},\C)\in H^1_\#(Y)^d$  solves the corrector equation
	\begin{equation}\label{eq:intro_contributions_corrector_eqn}
		\div \C:\left( \bar{\feps} + \nabla^s \fu \right) = \bm{0}.
	\end{equation}
\end{thm}
\review{Fourier neural operators were established~\cite{Li2021FNO,Kovachki2021FNOApproximation} as \emph{universal approximators} for general PDE solution operators. As periodic homogenization problems with linear elastic phases represent a subclass of such problems, the interested reader may wonder about the significance of the established theorem. The approximation statements provided for the general case concern proximity guarantees on compact subsets of input fields only, i.e., those which can be covered by finitely many balls of pre-defined radius. Then, in case an approximation property is established at the centers of these balls, a suitable continuity property of the operator to be approximated guarantees the approximation property on the entire compact set to be monitored. This strategy was carried out for (thermal) periodic homogenization problems in Bhattacharya et al.~\cite{Bhattacharya2024LearningHomogenization}, where establishing the continuity properties of the to-be-learned operator \eqref{eq:intro_contributions_FNO} turned out to be challenging. We refer the reader to the discussion after Proposition 1.1.~in Bhattacharya et al.~\cite[p.~1849]{Bhattacharya2024LearningHomogenization}\\
Our approach is completely different. The principal challenge is that the "{}space of stiffnesses"{} in its natural $L^\infty$-topology is rather large, even when imposing lower and upper bounds on the eigenvalues of the local stiffnesses as done above. The space $\mathcal{M}_Y(\alpha_-,\alpha_+)$ contains an open ball of an infinite-dimensional topological vector space, and is therfore not compact. It is even not separable, i.e., the set cannot be covered by a countable number of balls of fixed positive radius. Put differently, the space is \emph{huge}. Still, FNOs work on this space, and with guaranteed fidelity. To establish such a statement, we depart from the original compactness plus continuity strategy and rather look at how the operator \eqref{eq:intro_contributions_FNO} is formed, i.e., we disregard continuity entirely and look at the "{}construction plan"{} of the operator under scrutiny. As the operator is implicitly defined as the solution to an integral equation, the Lippmann-Schwinger equation \eqref{eq:intro_contributions_basic}, we rather mimic the iterative scheme to obtain the solution to this equation, i.e., we provide FNO proxies to all the components of the iterative scheme, eventually establishing the desired statement.}\\
\review{To realize this roadmap,} the main technical tools here are perturbation estimates which turn the constructed FNO into a contraction. The principal difficulty is that neural networks are typically only learned on compact domains. The operation \eqref{eq:intro_contributions_learnME} may be confronted with strains of arbitrary magnitude, for instance due to a singularity at a material interface with a corner or a kink. Thus, the error made by disregarding such large strain values needs to be quantified with care.\\
We discuss a specific approximation to the double contraction \eqref{eq:intro_contributions_learnME} in section \ref{sec:activations}, where we make use of the rectified linear unit (ReLU) as a reasonably simple and explicit activation function. In fact, conventional sigmoid activations may be used for learning the mapping \eqref{eq:intro_contributions_learnME} via conventional training, as well, but we chose to remain fully explicit. More precisely, our construction is based on the approximation of the multiplication of two real numbers by a ReLU-network introduced by Yarotsky~\cite{Yarotsky2017ReLU}.\\
We showcase our construction in section \ref{sec:computations} via dedicated computational examples. We demonstrate that the construction -- which was designed for theoretical purposes in the first place -- gives rise to an FNO which comes with an execution time that is on the same order of magnitude as FFT-based solvers, the current gold standard for computational homogenization.


\section{Homogenization in linear elasticity}
\label{sec:homogenization}

\subsection{What is given and what is sought}
\label{sec:homogenization_setup}

We consider a fixed periodic cell
\begin{equation}\label{eq:homogenization_setup_cell}
	Y = [0,L_1] \times [0,L_2] \times \cdots \times [0,L_d] 
\end{equation}
in $d$ spatial dimensions. We denote by $\Sym{d}$ the vector space of symmetric $d\times d$ tensors, and we let $\Sym{\Sym{d}}$ refer to the vector space of linear mappings on the vector space $\Sym{d}$ which are furthermore self-adjoint w.r.t. the Frobenius inner product
\begin{equation}\label{eq:homogenization_setup_Frobenius}
	\langle \feps_1, \feps_2 \rangle = \textrm{tr}(\feps_1 \feps_2), \quad \feps_1, \feps_2 \in \Sym{d},
\end{equation}
where $\textrm{tr}$ refers to the trace. The central object of study is the following. For positive constants $\alpha_\pm$, we denote by
\begin{equation}\label{eq:homogenization_setup_microstructures}
	\mathcal{M}(\alpha_-,\alpha_+) = \left\{ \C \in L^\infty(Y;\Sym{\Sym{d}}) \,\middle| \, \alpha_- \, \|\feps\|^2 \leq \feps : \C(\fx) : \feps \leq \alpha_+ \, \|\feps\|^2, \quad \text{a.e. }\fx \in Y, \quad \feps \in \Sym{d}\right\}
\end{equation}
the set of periodic fields of elasticity tensors with a material contrast
\begin{equation}\label{eq:homogenization_setup_material_contrast}
	\kappa = \frac{\alpha_+}{\alpha_-}.
\end{equation}
As the set \eqref{eq:homogenization_setup_microstructures} is empty unless the inequality $\alpha_- \leq \alpha_+$ is satisfied, the material contrast \eqref{eq:homogenization_setup_material_contrast} is always greater or equal to unity.\\
For every macroscopic strain $\bar{\feps} \in \Sym{d}$ and fixed stiffness $\C \in \mathcal{M}(\alpha_-,\alpha_+)$, we seek the periodic displacement field
\begin{equation}\label{eq:homogenization_setup_H1per}
	\fu \in H^1_\#(Y)^d \equiv \left\{ \fu \in H^1(Y) \text{ periodic}\,\middle| \, \int_Y \fu(\fx)\, d\fx = \bm{0} \right\}
\end{equation}
with vanishing mean, s.t. the balance equation
\begin{equation}\label{eq:homogenization_setup_equilibrium}
	\div \C:\left[ \bar{\feps} + \nabla^s \fu \right] = \bm{0}
\end{equation}
is satisfied in $H^{-1}_\#(Y)^d$, where $\nabla^s$ stands for the symmetrized gradient. By standard arguments, it may be shown the the solution \eqref{eq:homogenization_setup_H1per} is unique, and satisfies the a-priori estimate
\begin{equation}\label{eq:homogenization_setup_a_priori_L2}
	\|\nabla^s \fu\|_{L^2} \leq \sqrt{\frac{\alpha_+}{\alpha_-}} \|\bar{\feps}\|
\end{equation}
with the $L^2$-norm
\begin{equation}\label{eq:homogenization_setup_L2}
	\|\feps\|_{L^2} = \left( \frac{1}{\text{vol}(Y)} \int_Y \| \feps(\fx)\|^2 \, d\fx \right)^{\frac{1}{2}}, \quad \feps \in L^2(Y;\Sym{d}).
\end{equation}

\subsection{The Lippmann-Schwinger framework}
\label{sec:homogenization_basic}

Since its inception by Zeller-Dederichs~\cite{LSE} and Kröner~\cite{LSE2}, the Lippmann-Schwinger equation is at the heart of "{}modern"{} micromechanics~\cite{LSE3}. With the advent of digital microstructure images, the Lippmann-Schwinger formulation also found its way into computer codes, following the seminar articles by Moulinec-Suquet~\cite{MoulinecSuquet1994,MoulinecSuquet1998}. For the work at hand, we require some basic analytic properties of the Lippmann-Schwinger operator~\cite{Schneider2015Convergence,MSConvergence2023}, where we restrict to a reference material proportional to the identity~\cite{PorousSLE2020}.\\
For any prescribed strain $\bar{\feps} \in \Sym{d}$, stiffness field $\C \in \mathcal{M}(\alpha_-,\alpha_+)$ and positive constant $\alpha_0> 0$, we consider the associated Lippmann-Schwinger operator
\begin{equation}\label{eq:homogenization_basic_LSO1}
	\mathcal{L}(\cdot; \bar{\feps},\C,\alpha_0) \in L(L^2(Y;\Sym{d})),
\end{equation}
acting via
\begin{equation}\label{eq:homogenization_basic_LSO2}
	\mathcal{L}(\feps; \bar{\feps},\C,\alpha_0) = \bar{\feps} - \frac{1}{\alpha_0}\ffGamma : (\C - \C^0):\feps, \quad \feps \in L^2(Y;\Sym{d}),
\end{equation}
with the reference material $\C^0 \in \Sym{\Sym{d}}$ with action
\begin{equation}\label{eq:homogenization_basic_refMat}
	\C^0:\feps = \alpha_0 \, \feps, \quad \feps \in \Sym{d},
\end{equation}
and the Eshelby-Green operator $\ffGamma$ is defined via 
\begin{equation}\label{eq:homogenization_basic_EshelbyGreen}
	\ffGamma:\ftau (\fx) = \sum_{\fxi \in \Z^d} \hat{\ffGamma}(\tilde{\fxi}):\hat{\ftau}(\fxi) e^{\texttt{i}\, \tilde{\fxi} \cdot \fx}, \quad \fx \in Y, \quad \ftau \in L^2(Y;\Sym{d}),
\end{equation}
with the complex unit $\texttt{i}$, the Fourier coefficients
\begin{equation}\label{eq:homogenization_basic_EshelbyGreen_Fourier_coefficients}
	\hat{\ftau}(\fxi) = \frac{1}{\textrm{vol}(Y)} \int_Y \ftau(\fx) e^{-\texttt{i}\, \tilde{\fxi} \cdot \fx}\, d\fx,
\end{equation}
the scaled frequencies
\begin{equation}\label{eq:homogenization_basic_frequencies}
	\tilde{\xi}_j = \frac{2 \pi}{L_j} \, \xi_j, \quad j=1,2,\ldots,d,
\end{equation}
and the tensor-valued function
\begin{equation}\label{eq:homogenization_basic_EshelbyGreen_Fourier_coefficients_general}
	\hat{\ffGamma}: \R^d \rightarrow L(\Sym{d}),
\end{equation}
defined by $\hat{\ffGamma}(\bm{0}) = \bm{0}$ for vanishing vector $\fd$ and by
\begin{equation}\label{eq:homogenization_basic_EshelbyGreen_Fourier_coefficients_explicit}
	\hat{\ffGamma}(\fd):\ftau = \frac{\fd \otimes \left( \ftau \fd \right)}{\|\fd\|^2} + \frac{\left( \ftau \fd \right) \otimes \fd}{\|\fd\|^2} - \frac{ (\fd \cdot \ftau) \fd \, \fd \otimes \fd }{\|\fd\|^4}
\end{equation}
for non-zero $\fd \in \R^d$.\\
After these definitions, we list the main properties of the Lippmann-Schwinger operator \eqref{eq:homogenization_basic_LSO2} in the following (well-known) result. The first two items are evident from the connection of the Lippmann-Schwinger operator \eqref{eq:homogenization_basic_LSO2} to the fixed point operator associated to a gradient descent method for the average elastic energy, e.g., discussed in Schneider~\cite[§3.1]{Schneider2017FastGradient}. The third part about the $L^p$-estimates, where we use the following $L^p$-norm on strain fields
\begin{equation}\label{eq:homogenization_basic_Lp}
	\|\feps\|_{L^p} = \left( \frac{1}{\text{vol}(Y)} \int_Y \| \feps(\fx)\|^p \, d\fx \right)^{\frac{1}{p}}, \quad \feps \in L^p(Y;\Sym{d}), \quad 1 \leq p < \infty,
\end{equation}
i.e., we use the Frobenius norm \eqref{eq:homogenization_setup_Frobenius} inside the $L^p$-norm for convenience, is also known~\cite{Schneider2015Convergence}.
\begin{thm}\label{thm:homogenization_basic_contraction_estimates_EshelbyGreen}
The following statements hold for the Lippmann-Schwinger operator \eqref{eq:homogenization_basic_LSO2}:
\begin{enumerate}
	\item For any prescribed strain $\bar{\feps} \in \Sym{d}$, stiffness field $\C \in \mathcal{M}(\alpha_-,\alpha_+)$ and reference constant $\alpha_0$, the Lippmann-Schwinger operator \eqref{eq:homogenization_basic_LSO1} has precisely one fixed point $\feps^*$ which has the form
	\begin{equation}\label{eq:homogenization_basic_fixed_point}
		\feps^* = \bar{\feps} + \nabla^s \fu,
	\end{equation}
	where the displacement fluctuation $\fu \in H^1_\#(Y)^d$ is the unique solution to the equilibrium equation \eqref{eq:homogenization_setup_equilibrium}. In particular, the fixed point \eqref{eq:homogenization_basic_fixed_point} does not depend on the reference constant $\alpha_0$.
	\item 
	For the choice
	\begin{equation}\label{eq:homogenization_basic_refMat_basic}
		\alpha_0 = \frac{\alpha_+ + \alpha_-}{2},
	\end{equation}
	the Lippmann-Schwinger operator defines an $L^2$-contraction with constant
	\begin{equation}\label{eq:homogenization_basic_contraction_constant}
		\gamma = \frac{\alpha_+ - \alpha_-}{\alpha_+ + \alpha_-} \equiv 1 - \frac{2}{\kappa + 1},
	\end{equation}
	i.e., the estimate
	\begin{equation}\label{eq:homogenization_basic_contraction_estimate}
		\left\| \mathcal{L}(\feps_1; \bar{\feps},\C,\alpha_0) - \mathcal{L}(\feps_2; \bar{\feps},\C,\alpha_0) \right\|_{L^2} \leq \gamma\,\|\feps_1 - \feps_2\|_{L^2}
	\end{equation}
	holds for all $\feps_1, \feps_2 \in L^2(Y;\Sym{d})$.
	\item There are numbers $p \in (2,\infty)$ and $C_p > 0$, depending only on the pair $(\alpha_-,\alpha_+)$, s.t. the estimate
	\begin{equation}\label{eq:homogenization_basic_Lp_bound}
		\|\feps^*\|_{L^p} \leq C_p \, \|\bar{\feps}\|
	\end{equation}
	holds for the fixed point \eqref{eq:homogenization_basic_fixed_point}.
\end{enumerate}
\end{thm}
The key (non-trivial) arguments for Thm.~\ref{thm:homogenization_basic_contraction_estimates_EshelbyGreen} are related to certain properties of the Eshelby-Green operator \eqref{eq:homogenization_basic_EshelbyGreen}, which we collect in the following:
\begin{lem}\label{lem:homogenization_basic_Lp_estimates_EshelbyGreen}
	The Eshelby-Green operator \eqref{eq:homogenization_basic_EshelbyGreen} has following properties:
	\begin{enumerate}
		\item The tensor-valued Fourier multiplier gives rise to a bounded linear operator on $L^p(Y;\Sym{d})$ for exponents $1<p<\infty$, i.e., 
		\begin{equation}\label{eq:homogenization_basic_LpBound_EshelbyGreen}
			\ffGamma \in L(L^p(Y;\Sym{d})), \quad 1 < p < \infty,
		\end{equation}
		\item For the exponent $p=2$, the Eshelby-Green operator \eqref{eq:homogenization_basic_EshelbyGreen} defines an orthogonal projector on the $L^2$-space \eqref{eq:homogenization_setup_L2} onto the (closed) subspace of compatible strains. In particular, the norm equality
		\begin{equation}\label{eq:homogenization_basic_L2Bound_EshelbyGreen}
			\| \ffGamma \|_{L(L^2(Y;\Sym{d}))} = 1
		\end{equation}
		holds.
		\item For every constant $C > 1$, there is an exponent $\check{p}(C) > 2$, s.t., the inequality
		\begin{equation}\label{eq:homogenization_basic_ExplicitLpBound_EshelbyGreen}
			\| \ffGamma \|_{L(L^p(Y;\Sym{d}))} \leq C \quad \text{holds for all} \quad 2 \leq p \leq \check{p}(C).
		\end{equation}
	\end{enumerate}
\end{lem}
The arguments~\cite{Schneider2015Convergence} for Lemma \ref{lem:homogenization_basic_Lp_estimates_EshelbyGreen} are collected in Appendix \ref{apx:basic_EshelbyGreen}.

\section{Fourier neural operators via Lippmann-Schwinger solvers}
\label{sec:FNO}

 
\subsection{Fourier neural operators}
\label{sec:FNO_defn}
\begin{figure}
	\centering
	\includegraphics[width=\textwidth]{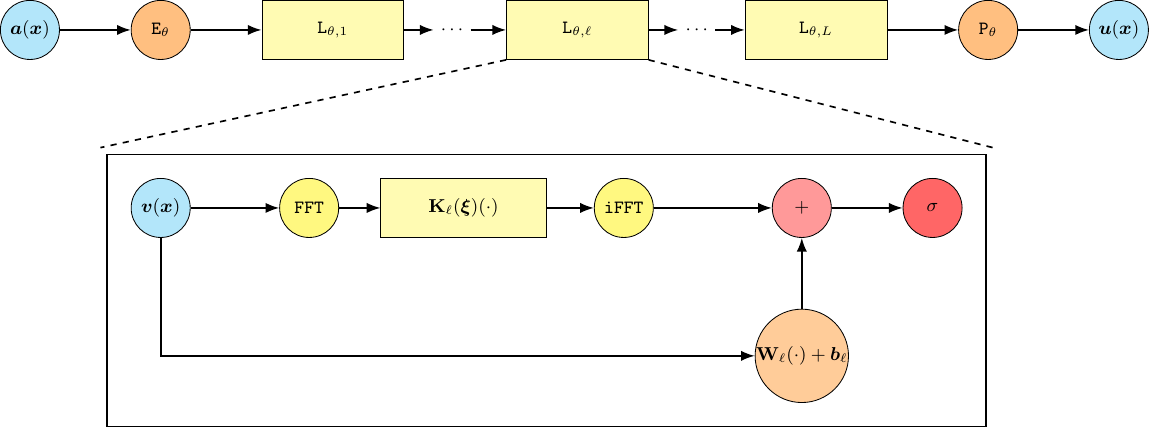}
	\caption{Schematic representation of the FNO topology \eqref{eq:FNO_defn_FNO_topology}.}
	\label{fig:fno_schematic}
\end{figure}
For a given domain $D \subseteq \R^d$, a Fourier Neural Operator (FNO)~\cite{Li2021FNO,Kovachki2021FNOApproximation,Bhattacharya2024LearningHomogenization} is a mapping
\begin{equation}\label{eq:FNO_defn_FNO_domain}
	\mathcal{F}_\theta: \mathcal{A}(D;\R^{d_a}) \rightarrow \mathcal{U}(D;\R^{d_u}), \quad d_a, d_u \in \N_{+},
\end{equation}
between Banach spaces of vector-valued functions on the domain $D$ which has the form of a composition
\begin{equation}\label{eq:FNO_defn_FNO_topology}
	\mathcal{F}_\theta = \texttt{P}_\theta \circ \texttt{L}_{\theta,L} \circ \cdots \texttt{L}_{\theta,2} \circ \texttt{L}_{\theta,1} \circ \texttt{E}_\theta, \quad L \in \N_+,
\end{equation}
of a lifting/extension operator
\begin{equation}\label{eq:FNO_defn_lifting_domain}
	\texttt{E}_\theta: \mathcal{A}(D;\R^{d_a}) \rightarrow \mathcal{U}(D;\R^{d_v})
\end{equation}
with $d_v \in \N_+$, $d_v \geq d_u$, neural-operator layers 
\begin{equation}\label{eq:FNO_defn_layer_domain}
	\texttt{L}_{\theta,\ell}: \mathcal{U}(D;\R^{d_v}) \rightarrow \mathcal{U}(D;\R^{d_v})
\end{equation}
of the form
\begin{equation}\label{eq:FNO_defn_layer_form}
	\texttt{L}_{\theta,\ell} (\fv) = \sigma(\fW_\ell \, \fv + \fb_\ell + \fK_\ell \star \fv), \quad \fv \in \mathcal{U}(D;\R^{d_v}),
\end{equation}
with an activation function $\sigma$, and a projection operator
\begin{equation}\label{eq:FNO_defn_projection_domain}
	\texttt{P}_\theta:\mathcal{U}(D;\R^{d_v}) \rightarrow \mathcal{U}(D;\R^{d_u}).
\end{equation}
Here, the weights $\fW_\ell \in \R^{d_v \times d_v}$ are matrices, the biases $\fb_\ell \in \mathcal{U}(D;\R^{d_v})$ are functions and the kernel operator
\begin{equation}\label{eq:FNO_defn_kernel_convolution}
	\left(\fK_\ell \star \fv\right) (\fx) = \int_D \fK_\ell(\fx - \fy)\,\fv(\fy)\, d\fy
\end{equation}
acts via convolution with a $\R^{d_v \times d_v}$-matrix valued function $\fK_\ell$, for $\ell = 1,2,\ldots, L$. The input- and output layers \eqref{eq:FNO_defn_lifting_domain} and \eqref{eq:FNO_defn_projection_domain} are typically taken as locally acting neural networks.\\
The name FNOs stems from the fact that the convolution \eqref{eq:FNO_defn_layer_form} with the kernel $\fK_\ell$ may be equivalently expressed via a multiplication in Fourier space with a $\R^{d_v \times d_v}$-matrix valued function Fourier multipliers $\hat{\fK}_\ell$. Fig.~\ref{fig:fno_schematic} depicts a schematic of the FNO construction.\\
A few remarks are at hand.
\begin{enumerate}
	\item We will always consider a rectangular domain $Y$, see eq.~\eqref{eq:homogenization_setup_cell}, and consider periodic boundary conditions when evaluating the convolution \eqref{eq:FNO_defn_kernel_convolution}.
	\item For the application at hand, we will consider a few simplifications to reduce the complexity of the FNO. For a start, we will assume that the activation function $\sigma$ is applied to \emph{some} of the components in the Fourier layer \eqref{eq:FNO_defn_layer_form} only. This assumption does not restrict generality, as the identity mapping may be approximated to arbitrary precision by a neural network. Secondly, we will use a neural network for the input layer \eqref{eq:FNO_defn_lifting_domain} instead of the local lifting operation originally proposed. As discussed in Kovachki et al.~\cite{Kovachki2021FNOApproximation}, such an operation may be represented by a deeper FNO at the expense of simplicity. Similarly, we will consider an output layer \eqref{eq:FNO_defn_projection_domain} of the form
\begin{equation}\label{eq:FNO_defn_projection_specific}
	\texttt{E}_\theta (\fv) = \fW_{L+1} \, \fv + \fb_{L+1} + \fK_{L+1} \star \fv, \quad \fv \in \mathcal{U}(D;\R^{d_v}),
\end{equation}
with a matrix $\fW_{L+1} \in \R^{d_u \times d_v}$, a function $\fb_{L+1} \in \mathcal{U}(D;\R^{d_u})$ and a $\R^{d_u \times d_v}$-matrix valued function $\fK_{L+1}$. Such a layer could also be approximated by a deeper FNO with only local and affine input as well as output layers.
	\item Expanding a general function $\fv \in \mathcal{U}(D;\R^{d_v})$ into a Fourier series
	\begin{equation}\label{eq:FNO_defn_Fourier_series_general}
		\fv (\fx) = \sum_{\fxi \in \Z^d} \hat{\fv}(\tilde{\fxi})\, e^{\texttt{i}\, \tilde{\fxi} \cdot \fx}, \quad \fx \in Y
	\end{equation}
	with Fourier coefficients
	\begin{equation}\label{eq:FNO_defn_Fourier_coefficients}
		\hat{\fv}(\fxi) = \frac{1}{\textrm{vol}(Y)} \int_Y \fv(\fx) e^{-\texttt{i}\, \tilde{\fxi} \cdot \fx}\, d\fx,
	\end{equation}
	and the scaled frequencies \eqref{eq:homogenization_basic_frequencies},
	the action of the convolution operator \eqref{eq:FNO_defn_kernel_convolution} in Fourier space reads
	\begin{equation}\label{eq:FNO_defn_convolution_Fourier}
		\left( \fK_\ell \star \fv \right) (\fx) = \sum_{\fxi \in \Z^d} \hat{\fK_\ell}(\fxi) \hat{\fv}(\fxi) e^{\texttt{i}\, \tilde{\fxi} \cdot \fx}, \quad \fx \in Y, 
	\end{equation}
	with suitable Fourier coefficients
	\begin{equation}\label{eq:FNO_defn_kernel_Fourier_coefficients}
		\hat{\fK_\ell}: \Z^d \rightarrow \C^{d_v \times d_v}
	\end{equation}
	which satisfy the condition
	\begin{equation}\label{eq:FNO_defn_keep_it_real}
		\hat{\fK_\ell}(-\fxi) = \overline{\hat{\fK_\ell}(\fxi)}, \quad \fxi \in \Z^d,
	\end{equation}
	ensuring that the function $\fK_\ell$ has real-valued components, and where we use the complex conjugation $z \mapsto \bar{z}$. 
	\item We follow Kovachki~\cite{Kovachki2021FNOApproximation} and do not assume a compact support of the Fourier coefficients \eqref{eq:FNO_defn_kernel_Fourier_coefficients}, in line with the original neural operator definition~\cite{Kovachki2023NO} in contrast to other treatments~\cite{Li2021FNO,Bhattacharya2024LearningHomogenization}. A restriction
	\begin{equation}\label{eq:FNO_defn_restricted_lattice}
		\hat{\fK_\ell}(\fxi) = \bm{0} \text{ for all frequencies $\fxi \in \Z^d$ with } |\xi_j | \leq \frac{N}{2}, \quad j=1,2,\ldots,d,
	\end{equation}
	is not able to handle all input stiffnesses $\C$ in the set $\mathcal{M}(\alpha_-,\alpha_+)$, for trivial reasons. Suppose that a particular inhomogeneous stiffness $\C \in \mathcal{M}(\alpha_-,\alpha_+)$ is fixed, and let $\feps = \bar{\feps} + \nabla^s \fu$ denote the unique solution to the Lippmann-Schwinger equation, i.e., a fixed point of the Lippmann-Schwinger operator \eqref{eq:homogenization_basic_LSO2}. For $K \in \N$, consider the rescaled stiffness
	\begin{equation}\label{eq:FNO_defn_rescaled_stiffness}
		\tilde{\C}(\fx) = \C(2^K \, \fx), \quad \fx \in Y.
	\end{equation}
	The stiffness $\tilde{\C}$ is also periodic and defines an element of the space $\mathcal{M}(\alpha_-,\alpha_+)$. The corresponding (unique) solution $\tilde{\feps}$ has the form
	\begin{equation}\label{eq:FNO_defn_rescaled_strain}
		\tilde{\feps}(\fx) = \feps(2^K \, \fx), \quad \fx \in Y.
	\end{equation}
	In particular, the Fourier coefficients of the field $\tilde{\feps}$ are non-zero only in the case
	\begin{equation}\label{eq:FNO_defn_rescaled_Fourier_coefficients}
		\widehat{\tilde{\feps}}(2^K\fxi) = \hat{\feps}(\fxi), \quad \fxi \in \Z^d.
	\end{equation}
	As a result, in case the power $2^K$ is chosen to exceed $N$, we have
	\begin{equation}\label{eq:FNO_defn_vanishing strains}
		\widehat{\tilde{\feps}}(\feta) = \bm{0}\text{ for all frequencies $\fxi \in \Z^d \backslash\{\bm{0}\}$ with } |\xi_j | \leq \frac{N}{2}, \quad j=1,2,\ldots,d.
	\end{equation}
	In particular, the solution $\widehat{\tilde{\feps}}$ cannot be represented on the restricted lattice \eqref{eq:FNO_defn_restricted_lattice}, no matter how large the parameter $N$ is chosen in advance! To sum up: Restricting the lattice \eqref{eq:FNO_defn_restricted_lattice} to be compactly supported in the FNO layer \eqref{eq:FNO_defn_layer_form} implies a resolution constraint on the input stiffnesses \eqref{eq:homogenization_setup_microstructures} to be treated. In case all stiffnesses \eqref{eq:homogenization_setup_microstructures} are to be treated, such a constraint must not be enforced!
\end{enumerate}

\subsection{A neural fixed-point method}
\label{sec:FNO_LS}

To construct an expressive FNO, we take the Lippmann-Schwinger operator \eqref{eq:homogenization_basic_LSO2}
\begin{equation}\label{eq:FNO_LS_LSO}
	\mathcal{L}(\feps; \bar{\feps},\C,\alpha_0) = \bar{\feps} - \ffGamma : \frac{1}{\alpha_0}(\C - \C^0):\feps, \quad \feps \in L^2(Y;\Sym{d}),
\end{equation}
as our blueprint. More precisely, we approximate the double-contraction operator
\begin{equation}\label{eq:FNO_LS_doubleContraction}
	\Multiplication: L(\Sym{d}) \times \Sym{d} \rightarrow \Sym{d}, \quad (\mathds{T},\feps)\mapsto \mathds{T}:\feps,
\end{equation}
by a neural network, i.e., work with a more or less abstract (finite-dimensional) operator
\begin{equation}\label{eq:FNO_LS_doubleContractionNN}
	\Multiplication_\theta: L(\Sym{d}) \times \Sym{d} \rightarrow \Sym{d}.
\end{equation}
Then, we set up the Fourier neural operator
\begin{equation}\label{eq:FNO_LS_FLSO}
	\mathcal{L}_\theta(\feps; \bar{\feps},\C,\alpha_0) = \bar{\feps} - \ffGamma : \Multiplication_\theta \left( \frac{1}{\alpha_0}(\C - \C^0), \feps \right), \quad \feps \in L^2(Y;\Sym{d}).
\end{equation}
To proceed, we consider a few assumptions satisfied by the approximation \eqref{eq:FNO_LS_doubleContractionNN}.
\begin{ass}\label{ass:mtheta}
	For any positive constant $M \geq 1$ and every error margin $\delta>0$, there is a positive constant $C>0$, independent of the parameters $M$ and $\delta$, and a neural network \eqref{eq:FNO_LS_doubleContractionNN}, s.t., 
	\begin{enumerate}[(i)]
	\item The inequality
		\begin{equation}\label{eq:ass_mtheta_smallInputEstimate}
			\left\| \Multiplication_\theta(\mathds{T},\feps) - \mathds{T}:\feps \right\| \leq \delta \, \|\feps\|
		\end{equation}
		holds for all tensors $\mathds{T} \in L(\Sym{d})$ and $\feps \in \Sym{d}$ satisfying the constraints
		\begin{equation}\label{eq:ass_mtheta_smallInputEstimate_prereq}
			\| \mathds{T} \| \leq 1 \quad \text{and} \quad \| \feps\| \leq M.
		\end{equation}
	\item The bound 
		\begin{equation}\label{eq:ass_mtheta_largeInputEstimate}
			\left\| \Multiplication_\theta(\mathds{T},\feps) - \mathds{T}:\feps \right\| \leq C \, \|\feps\|
		\end{equation}
		is valid for all tensors $\mathds{T} \in L(\Sym{d})$ and $\feps \in \Sym{d}$ satisfying the constraints
		\begin{equation}\label{eq:ass_mtheta_largeInputEstimate_prereq}
			\| \mathds{T} \| \leq 1 \quad \text{and} \quad \| \feps\| \geq M.
		\end{equation}
	\item The Lipschitz estimate
		\begin{equation}\label{eq:ass_mtheta_contraction}
			\left\| \Multiplication_\theta(\mathds{T},\feps_1) - \Multiplication_\theta(\mathds{T},\feps_2) \right\| \leq \left( \|\mathds{T}\| + \delta \right)\, \| \feps_1 - \feps_2\| 
		\end{equation}
		holds for all strains $\feps_1, \feps_2 \in \Sym{d}$ and all tensors $\mathds{T} \in L(\Sym{d})$ obeying the constraint
		\begin{equation}\label{eq:ass_mtheta_contraction_prereq}
			\| \mathds{T} \| \leq 1.
		\end{equation}
	\end{enumerate}
\end{ass}
For the work at hand, we use the rectified linear unit (ReLU) activation function to construct a multiplication operator \eqref{eq:FNO_LS_doubleContractionNN} which satisfies the assumption \ref{ass:mtheta}, see section \ref{sec:activations}. However, there are numerous other activation functions used to construct neural networks \eqref{eq:FNO_LS_doubleContractionNN}. Thus, for the work at hand, we separated the assumption \ref{ass:mtheta} from the concrete activation function to show that our construction of the neural fixed-point operator \eqref{eq:FNO_LS_FLSO} depends on the stated estimates only, and does not involve a specific activation function. In particular, in case the bounds \eqref{eq:ass_mtheta_smallInputEstimate}, \eqref{eq:ass_mtheta_largeInputEstimate} and \eqref{eq:ass_mtheta_contraction} are established for alternative activation functions would  lead to the same conclusions about the expressivity of the involved FNOs for these scenarios, as well.\\
Under assumption \ref{ass:mtheta}, we establish the following result.
\begin{prop}\label{prop:FNO_LS_FLNO}
	Suppose assumption \ref{ass:mtheta} is satisfied, and we fix a field
	\begin{equation}\label{eq:FNO_LS_StiffnessSpace}
		\C \in \mathcal{M}_Y(\alpha_-,\alpha_+)
	\end{equation}
	of linear elastic stiffness tensors and a macroscopic strain $\bar{\feps} \in \Sym{d}$. Then, the following holds:
	\begin{enumerate}[(i)]
		\item The neural operator \eqref{eq:FNO_LS_FLSO} is well-defined and satisfies the estimate
		\begin{equation}\label{eq:FNO_LS_FLNO_estimate}
			\left\| \mathcal{L}_\theta(\feps_1; \bar{\feps},\C,\alpha_0) - \mathcal{L}_\theta(\feps_2; \bar{\feps},\C,\alpha_0) \right\|_{L^2} \leq (\gamma_0 + \delta) \|\feps_1 - \feps_2\|_{L^2}
		\end{equation}
		in case the quantity
		\begin{equation}\label{eq:FNO_LS_FLNO_contraction_constant}
			\gamma_0 = \max \left( \left| \frac{\alpha_+ - \alpha_0}{\alpha_0}\right|, \left| \frac{\alpha_- - \alpha_0}{\alpha_0}\right| \right)
		\end{equation}
		does not exceed unity.
		\item In case the quantity
		\begin{equation}\label{eq:FNO_LS_FLNO_combined_contraction_constant}
			\gamma_\theta = \gamma_0 + \delta
		\end{equation}
		is less than unity, the unique fixed point $\feps_{\theta}^*$ of the contractive neural operator \eqref{eq:FNO_LS_FLSO} satisfies the proximity estimate
		\begin{equation}\label{eq:FNO_LS_FLNO_proximity_estimate}
			\| \feps^* - \feps^*_\theta\|_{L^2} \leq \frac{1}{1 - \gamma_\theta} \left[ \delta\,\| \feps^* \|_{L^2} + C \, \frac{\|\feps^* \|_{L^p}^{\frac{p}{2}}}{M^{\frac{p-2}{2}}} \right]
		\end{equation}
		to the (unique) fixed point $\feps^*$ of the Lippmann-Schwinger operator \eqref{eq:homogenization_basic_LSO2} which resides in an $L^p$ space with an exponent $p>2$ that depends only on the material contrast $\kappa = \alpha_+ / \alpha_-$, see the statement \eqref{eq:homogenization_basic_Lp_bound}.
	\end{enumerate}
\end{prop}
For increased readability, the arguments were out-sourced to Apx.~\ref{apx:LSNO}.

\subsection{The strategy for establishing the main result}
\label{sec:FNO_result}

\begin{figure}
	\centering
	\includegraphics[width=1.0\textwidth]{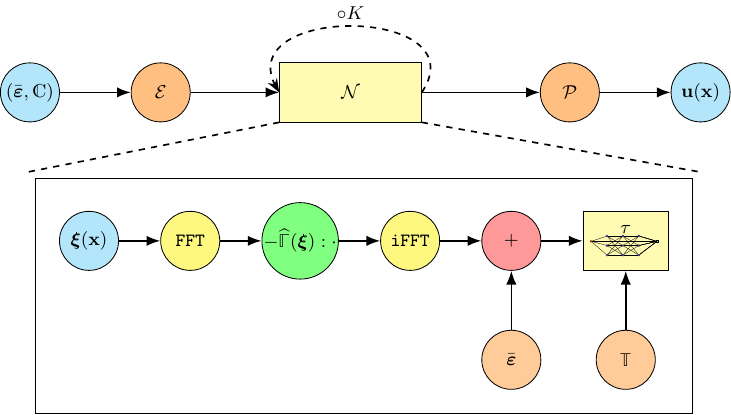}
	\caption{Schematic of the LS-FNO construction \eqref{eq:FNO_result_FNO_construction_finalExpression}.}
	\label{fig:recurrent_fno_schematic}
\end{figure}

The purpose of this section is to provide arguments for the validity of Thm.~\ref{thm:main}, i.e., for every periodic unit cell $Y$, every pair of lower and upper bounds $0<\alpha_- < \alpha_+$ on the material stiffness, every strain bound $\eps_0 > 0$ and every error margin $\delta_{\texttt{target}}>0$, there is an FNO
\begin{equation}\label{eq:FNO_result_contributions_FNO}
	\mathcal{F}_\theta: \Sym{d} \times \mathcal{M}_Y \rightarrow H^1_\#(Y)^d,
\end{equation}
s.t. the estimate 
\begin{equation}\label{eq:FNO_result_contributions_result}
	\left\| \mathcal{F}_\theta(\bar{\feps},\C) - \mathcal{S}(\bar{\feps},\C)\right\|_{H^1_\#(Y)^d} \leq \delta_{\texttt{target}}
\end{equation}
holds for all macroscopic strains $\bar{\feps} \in \Sym{d}$ obeying the bound
\begin{equation}\label{eq:FNO_result_contributions_assumedBarEpsBound}
	\| \bar{\feps} \| \leq \eps_0
\end{equation}
and all stiffness distributions 
\begin{equation}\label{eq:FNO_result_contributions_MAlphaPM}
	\C \in \mathcal{M}_Y(\alpha_-,\alpha_+),
\end{equation}
and where $\mathcal{S}$ refers to the exact solution operator, i.e., for all macroscopic strains $\bar{\feps} \in \Sym{d}$ and all stiffnesses \eqref{eq:intro_contributions_MAlphaPM}, the displacement field $\fu = \mathcal{S}(\bar{\feps},\C)\in H^1_\#(Y)^d$  solves the corrector equation
\begin{equation}\label{eq:FNO_result_contributions_corrector_eqn}
	\div \C:\left( \bar{\feps} + \nabla^s \fu \right) = \bm{0}.
\end{equation}
We will follow the framework of section \ref{sec:FNO_defn} only loosely -- it turns out to be more convenient to work with the full neural network \eqref{eq:FNO_LS_doubleContractionNN} representing the double contraction rather than considering neural network layers individually. We construct the FNO \eqref{eq:FNO_result_contributions_FNO} in the form
\begin{equation}\label{eq:FNO_result_FNO_construction_finalExpression}
	\mathcal{F}_\theta = \mathcal{P} \circ \mathcal{N}^{\circ K} \circ \,\mathcal{E},
\end{equation}
with the embedding operator
\begin{equation}\label{eq:FNO_result_FNO_construction_embedding}
	\begin{array}{rrcl}
		\mathcal{E}: & \Sym{d} \times \mathcal{M}_Y &\rightarrow & Z,\\
		& (\bar{\feps}, \C) & \mapsto & \left(\bar{\feps}, \frac{\C - \C^0}{\alpha_0} , \Multiplication_\theta\left( \frac{\C - \C^0}{\alpha_0}, \feps \right)\right), \quad \alpha_0= \frac{\alpha_- + \alpha_+}{2},
	\end{array}
\end{equation}
mapping to the Banach space
\begin{equation}\label{eq:FNO_result_FNO_construction_BanachSpace}
	Z = \Sym{d} \oplus L^\infty(Y;L(\Sym{d})) \oplus L^2(Y;\Sym{d}),
\end{equation}
the Fourier neural layer
\begin{equation}\label{eq:FNO_result_FNO_construction_layerOperator}
	\begin{array}{rrcl}
		\mathcal{N}: & Z &\rightarrow & Z,\\
		& (\bar{\feta}, \mathds{T}, \fxi) & \mapsto & \left(\bar{\feta}, \mathds{T}, \Multiplication_\theta\left( \mathds{T}, \bar{\feta} - \ffGamma:\fxi \right) \right),
	\end{array}	
\end{equation}
which is applied $K \in \N$ times in the expression \eqref{eq:FNO_result_FNO_construction_finalExpression}, and the projection operator
\begin{equation}\label{eq:FNO_result_FNO_construction_projection}
	\begin{array}{rrcl}
		\mathcal{P}: & Z &\rightarrow & H^1_\#(Y)^d,\\
		& (\bar{\feta}, \mathds{T}, \fxi) & \mapsto & -\fG \div \fxi,
	\end{array}	
\end{equation}
where Green's operator \eqref{eq:homogenization_basic_thm_proof_part1_2}
\begin{equation}\label{eq:FNO_result_FNO_construction_projection_Green}
	\fG \in L(H^{-1}_\#(Y)^d, H^{1}_\#(Y)^d)
\end{equation}
arises as the inverse of the bounded linear operator \eqref{eq:homogenization_basic_thm_proof_part1_3}
\begin{equation}\label{eq:FNO_result_FNO_construction_projection_Green2}
	H^1_\#(Y)^d \rightarrow H^{-1}_\#(Y)^d, \quad \fu \mapsto \div \nabla^s \fu
\end{equation}
and may be represented as a Fourier multiplier. Our construction depends on three parameters:
\begin{equation}\label{eq:FNO_result_parameters}
	(K,M,\delta) \in \N \times \R_{>0} \times \R_{>0},
\end{equation}
which need to be chosen judiciously: $K$ counts the number of times the central operator \eqref{eq:FNO_result_FNO_construction_layerOperator} is applied, $M$ parametrizes that bound on numbers where the double contraction \eqref{eq:FNO_LS_doubleContractionNN} is approximated with fidelity $\delta$. Fig.~\ref{fig:recurrent_fno_schematic} illustrates the construction of the FNO \eqref{eq:FNO_result_contributions_FNO}.\\
The construction \eqref{eq:FNO_result_FNO_construction_finalExpression} is a bit awkward because it needed to fit the setup of FNOs \eqref{eq:FNO_defn_FNO_topology}, where the central neural layer \eqref{eq:FNO_defn_layer_form} applies the activation function last. In the Lippmann-Schwinger neural operator \eqref{eq:FNO_LS_FLSO}, the neural network is applied \emph{first}, and the Fourier action proceeds afterwards. To make the connection and aid in the analysis, for $k=0,1,\ldots,K$ we introduce the quantity
\begin{equation}
	\feps^{k+1}_{\theta} = \bar{\feps} - \ffGamma : \fxi^k_\theta \quad \text{where} \quad \fxi^k_\theta = \texttt{pr}_3 \left( \mathcal{N}^{\circ k} \circ \,\mathcal{E}(\bar{\feps}, \C) \right)
\end{equation}
arises as the third component of an element of the space \eqref{eq:FNO_result_FNO_construction_BanachSpace}. Then, we observe, on the one hand
\begin{equation}\label{eq:FNO_result_FNO_iteration}
	\feps^{k+1}_\theta = \mathcal{L}_\theta\left(\feps^k_\theta; \bar{\feps},\C, \frac{\alpha_- + \alpha_+}{2}\right), \quad \feps^0_\theta = \bar{\feps},
\end{equation}
involving the Lippmann-Schwinger neural operator \eqref{eq:FNO_LS_FLSO},
and, on the other hand, the expression
\begin{equation}\label{eq:FNO_result_FNO_output_as_iterate}
	\feps^{K+1}_\theta = \bar{\feps} + \nabla^s \fu \quad \text{with} \quad \fu = \mathcal{F}_\theta(\bar{\feps}, \C)
\end{equation}
for the output of the FNO \eqref{eq:FNO_LS_FLSO}. For the convenience of the reader, the detailed arguments for establishing the fidelity estimate \eqref{eq:FNO_result_contributions_result} are provided in Appendix \ref{apx:main_result_arguments}. We will just provide a brief sketch here.\\
The FNO \eqref{eq:FNO_result_FNO_construction_finalExpression} is designed in such a way that its output \eqref{eq:FNO_result_FNO_output_as_iterate} equals the $(K+1)$-th iterate of the Lippmann-Schwinger FNO \eqref{eq:FNO_LS_FLSO}, inductively defined by eq.~\eqref{eq:FNO_result_FNO_iteration}. For properly chosen parameters \eqref{eq:FNO_result_parameters}, the Lippmann-Schwinger FNO \eqref{eq:FNO_LS_FLSO} defines an $L^2$-contraction, in particular there is a unique fixed point $\feps^*_\theta$
\begin{equation}\label{eq:FNO_result_FNO_fixed_point}
	\feps^*_\theta = \mathcal{L}_\theta\left(\feps^*_\theta; \bar{\feps},\C,\frac{\alpha_- + \alpha_+}{2}\right).
\end{equation}
Then, we may estimate
\begin{equation}\label{eq:FNO_result_error_decomposition}
	\left\| \feps_\theta^{K+1} - \feps^* \right\|_{L^2} \leq \left\| \feps_\theta^{K+1} - \feps^*_\theta \right\|_{L^2} + \left\| \feps^*_\theta - \feps^* \right\|_{L^2}
\end{equation}
The first term
\begin{equation}\label{eq:FNO_result_error_first}
	\left\| \feps_\theta^{K+1} - \feps^*_\theta \right\|_{L^2}
\end{equation}
equals the \emph{iteration error}. Due to the contractivity property of the Lippmann-Schwinger FNO \eqref{eq:FNO_LS_FLSO}, this error becomes infinitesimal as the parameter $K$ goes to infinity. The second error
\begin{equation}\label{eq:FNO_result_error_second}
	\left\| \feps^*_\theta - \feps^* \right\|_{L^2}
\end{equation}
measures the \emph{approximation error} that was induced by replacing the double contraction \eqref{eq:FNO_LS_doubleContraction} by a neural network \eqref{eq:FNO_LS_doubleContractionNN}. By appropriately choosing the parameters $M$ and $\delta$, this error may be made as small as desired, possibly at the expense of more layers to be used in the network.\\
There are two principal difficulties to be dealt with. For a start, all the arguments need to be \emph{uniform} in the input stiffness $\C$, i.e., \emph{the same} FNO needs has to provide the desired accuracy \eqref{eq:FNO_result_contributions_result}. The second difficulty concerns linear elasticity itself. Solutions to the balance equation \eqref{eq:homogenization_setup_equilibrium} are much less regular than in case of scalar elliptic equations~\cite{Necas1976CounterExample}, in particular the gradients of the displacement are generally unbounded and only $L^p$-estimates for an exponent $p$ slightly larger than two hold~\cite{ShiWright1994,Herzog2011}. Neural networks, on the other hand, are typically only trained on compact sets. Thus, a particular way needs to be found to estimate the "{}remainder"{} not well-approximated by the neural network \eqref{eq:FNO_LS_doubleContractionNN}. As already announced, the details are found in Appendix \ref{apx:main_result_arguments}

\section{Construction of the double-contraction operator with ReLU}
\label{sec:activations}

\begin{figure}[h!]
	\centering
		\begin{minipage}{0.36\textwidth}
			\centering
			\includegraphics[width=\linewidth]{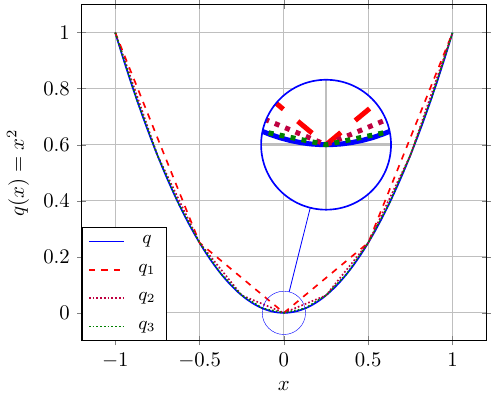}
			\subcaption{}
		\end{minipage}
		\hfill
		\begin{minipage}{0.36\textwidth}
			\centering
			\includegraphics[width=\linewidth]{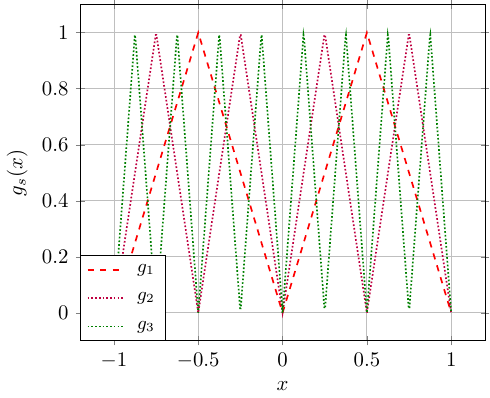}
			\subcaption{}
		\end{minipage}
		\hfill
		\begin{minipage}{0.24\textwidth}
			\centering
			\includegraphics[width=0.6\linewidth]{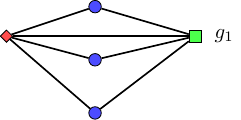}
			\subcaption{}			
			\includegraphics[width=0.8\linewidth]{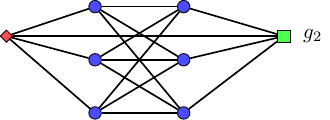}
			\subcaption{}			
			\includegraphics[width=\linewidth]{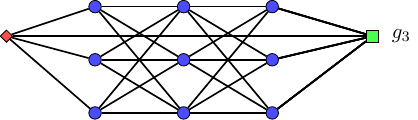}
			\subcaption{}
		\end{minipage}		
		\caption{Approximating the square function $q(x) = x^2$ following Yarotski~\cite{Yarotsky2017ReLU}: (a) ReLU approximations $q_{\mathfrak{m}}(x) = x - \sum_{k=1}^{\mathfrak{m}} \frac{g_k(x)}{4^k}$ by piece-wise linear functions shown in (b), which are constructed by the (c-d-e) ReLU network of depth $\mathfrak{m}=1$, $\mathfrak{m}=2$ and $\mathfrak{m}=3$.}
		\label{fig:ReLUapproximation}
\end{figure}

The purpose of this section is to provide an explicit construction of a neural network \eqref{eq:FNO_LS_doubleContractionNN}
\begin{equation}\label{eq:activations_doubleContractionNN}
	\Multiplication_\theta: L(\Sym{d}) \times \Sym{d} \rightarrow \Sym{d},
\end{equation}
which approximates the double-contraction operator \eqref{eq:FNO_LS_doubleContraction} in such a way that all the conditions in Assumption \ref{ass:mtheta} are satisfied. In a suitable vector notation, the double-contraction operator \eqref{eq:FNO_LS_doubleContraction} is represented by a matrix-vector multiplication. Thus, it appears imperative to \emph{learn} the multiplication of two real numbers. For the sake of being explicit, we consider the ReLU activation function
\begin{equation}\label{eq:activations_ReLU}
	\sigma: \R \rightarrow \R, \quad x \mapsto \left\{
		\begin{array}{rl}
			x, & x \geq 0,\\
			0, & \text{otherwise}.
		\end{array}
	\right.
\end{equation}
To approximate the multiplication of two real numbers, we follow Yarotsky~\cite{Yarotsky2017ReLU} and consider a ReLU approximation of the square function
\begin{equation}\label{eq:activations_square}
	q:[-1,1] \rightarrow [0,1], \quad q(x) = x^2,
\end{equation}
illustrated in Fig.~\ref{fig:ReLUapproximation}.
The following result holds, which is based on Yarotsky~\cite{Yarotsky2017ReLU}, who considered only the uniform norm, and Gühring et al.~\cite{Guehring2020ReLUWsp}, who provided additional estimates for Yarotsky's construction.
\begin{prop}\label{prop:activations_square_approximation}
	For any fidelity level $\delta > 0$, there is a neural network
	\begin{equation}\label{eq:activations_square_approximation}
		q_\theta: [-1,1] \rightarrow [0,1]
	\end{equation}
	based on the ReLU activation function \eqref{eq:activations_ReLU} with the following properties:
	\begin{enumerate}[(i)]
		\item The network $q_\theta$ is Lipschitz continuous and satisfies the estimate
		\begin{equation}\label{eq:activations_square_approximation_estimate}
			\|q - q_\theta\|_{W^{1,\infty}(-1,1)} \leq \delta.
		\end{equation}
		\item The function $q_\theta$ vanishes at the origin
		\begin{equation}\label{eq:activations_square_approximation_origin}
			q_\theta(0) = 0.
		\end{equation}
	\end{enumerate}
\end{prop}
We refer to Gühring et al.~\cite[Prop.~C.1]{Guehring2020ReLUWsp} for a proof.\\
With this result at hand, and based on the polarization identity
\begin{equation}\label{eq:activations_polarization_identity}
	ab = 2M^2 \left[ \left(\frac{a+b}{2M}\right)^2 - \left(\frac{a}{2M}\right)^2 - \left(\frac{b}{2M}\right)^2\right],
\end{equation}
valid for all real numbers $a,b \in \R$ and any positive number $M>0$,
we follow Yarotsky~\cite{Yarotsky2017ReLU} and define the multiplication operator
\begin{equation}\label{eq:activations_multiplicationNN}
	\multiplication_\theta: [-M,M]^2 \rightarrow \R, \quad \multiplication_\theta(a,b) = 2M^2 \left[ q_\theta\left( \frac{a+b}{2M} \right) - q_\theta\left( \frac{a}{2M}\right) - q_\theta\left( \frac{b}{2M} \right) \right],
\end{equation}
associated to the approximation \eqref{eq:activations_square_approximation} of the square function.\\
As a second ingredient for our construction, we introduce the ridge function
\begin{equation}\label{eq:activations_ridge_function}
	r_\theta: \R \rightarrow \R, \quad
	r_\theta (x) = \left\{
		\begin{array}{rl}
			-M, & x \leq -M,\\
			x, & -M < x < M,\\
			M, & x \geq M,
		\end{array}
	\right.
\end{equation}
which may be represented exactly via a ReLU network of depth two. With the networks \eqref{eq:activations_multiplicationNN} and \eqref{eq:activations_ridge_function} at hand, we may construct the desired double-contraction operator \eqref{eq:activations_doubleContractionNN}. Using Mandel's notation, we may represent the strain $\feps \in \Sym{d}$ in the form of a vector $\underline{\eps}$ with components $\eps_i$, and an operator
\begin{equation}\label{eq:activations_T_operator}
	\mathds{T} \in L(\Sym{d})
\end{equation}
is represented by a matrix $\underline{\underline{T}}$, s.t. the norm equivalences
\begin{equation}\label{eq:activations_norm_eq_strain}
	\|\feps\| = \| \underline{\eps} \|
\end{equation}
and
\begin{equation}\label{eq:activations_norm_eq_operator}
	\|\mathds{T}\| = \| \underline{\underline{T}} \|
\end{equation}
hold for all such objects, where we have the Frobenius norm \eqref{eq:homogenization_setup_Frobenius} on the left hand side in eq. \eqref{eq:activations_norm_eq_strain}, and the Euclidean norm on the right hand side. In the identity \eqref{eq:activations_norm_eq_operator}, both hands refer to the respective induced spectral norms. With these preliminary notions at hand, we may construct the double-contraction operator \eqref{eq:activations_doubleContractionNN}:
\begin{equation}\label{eq:activations_double_contraction_construction}
	\Multiplication_\theta: L(\Sym{d}) \times \Sym{d} \rightarrow \Sym{d}, \quad \fxi = \Multiplication_\theta(\mathds{T}, \feps),
\end{equation}
with
\begin{equation}\label{eq:activations_double_contraction_construction2}
	\xi_i = \sum_{j=1}^K \multiplication_{\theta}(T_{ij}, r_\theta(\eps_j)), \quad i=1,2,\ldots, K,
\end{equation}
where we set $K = d(d+1)/2$. In words, the operator \eqref{eq:activations_double_contraction_construction} proceeds by first restricting the input strains component-wise to the interval $[-M,M]$ via the ridge function \eqref{eq:activations_ridge_function}, and subsequently performs the matrix-vector multiplication via the learned multiplication operator \eqref{eq:activations_multiplicationNN}.\\
With these assumptions at hand, we wish to show that Assumption \ref{ass:mtheta} is satisfied for the construction \eqref{eq:activations_double_contraction_construction}. More precisely, the following statement holds, see Apx.~\ref{apx:multiplication} for a proof.

\begin{lem}\label{lem:activations_assumptions}
	The construction \eqref{eq:activations_double_contraction_construction} satisfies the following properties:
	\begin{enumerate}[(i)]
		\item 	The Lipschitz estimate
			\begin{equation}\label{eq:activations_assumptions_lem_Lipschitz}
				\left\| \Multiplication_\theta(\mathds{T},\feps_1) - \Multiplication_\theta(\mathds{T},\feps_2) \right\| \leq \left( \|\mathds{T}\| + M\delta d(d+1)\, \max_{i,j}|T_{ij}| \right) \| \feps_1 - \feps_2\|
			\end{equation}
			holds for all strains $\feps_1, \feps_2 \in \Sym{d}$ and all tensors $\mathds{T} \in L(\Sym{d})$ obeying the constraint
			\begin{equation}\label{eq:activations_assumptions_lem_boundedTensors}
				\| \mathds{T} \| \leq M.
			\end{equation}
		\item The estimate
			\begin{equation}\label{eq:activations_assumptions_lem_multiplication_proximity_small}
				\left\| \Multiplication_\theta(\mathds{T}, \feps) - \mathds{T}:\feps \right\| \leq M\delta d(d+1)\, \max_{i,j}|T_{ij}| \|\feps\|
			\end{equation}
			holds for all tensors $\feps \in \Sym{d}$ and all tensors $\mathds{T} \in L(\Sym{d})$ under the constraints
			\begin{equation}\label{eq:activations_assumptions_lem_boundedStrainsTensors}
				 \| \mathds{T} \| \leq M \quad \text{and} \quad \|\feps\| \leq M.
			\end{equation}
		\item 	The bound 
			\begin{equation}\label{eq:activations_assumptions_lem_multiplication_proximity_large}
				\left\| \Multiplication_\theta(\mathds{T},\feps) - \mathds{T}:\feps \right\| \leq (1 + \delta M d (d+1)) \|\mathds{T}\| \, \|\feps\|
			\end{equation}
			holds for all tensors $\mathds{T} \in L(\Sym{d})$ and $\feps \in \Sym{d}$ satisfying the constraints
				\begin{equation}\label{eq:activations_assumptions_lem_boundedTensors_unbounded strains}
					\| \mathds{T} \| \leq M \quad \text{and} \quad \| \feps\| \geq M.
				\end{equation}
	\end{enumerate}
\end{lem}

We are in the position to establish that the assumption \eqref{ass:mtheta} holds for the construction \eqref{eq:activations_double_contraction_construction}-\eqref{eq:activations_double_contraction_construction2} with constant $C = 2$ (we could have chosen any other positive number greater than unity).
\begin{prop}\label{prop:activations_assumptions}
For any positive constant $M \geq 1$ and every error margin $\delta_0>0$, the constructed neural network \eqref{eq:activations_double_contraction_construction}-\eqref{eq:activations_double_contraction_construction2} satisfies the following conditions:
\begin{enumerate}[(i)]
\item The inequality
	\begin{equation}\label{eq:activations_assumptions_mtheta_smallInputEstimate}
		\left\| \Multiplication_\theta(\mathds{T},\feps) - \mathds{T}:\feps \right\| \leq \delta_0 \, \|\feps\|
	\end{equation}
	holds for all tensors $\mathds{T} \in L(\Sym{d})$ and $\feps \in \Sym{d}$ satisfying the constraints
	\begin{equation}\label{eq:activations_assumptions_mtheta_smallInputEstimate_prereq}
		\| \mathds{T} \| \leq 1 \quad \text{and} \quad \| \feps\| \leq M.
	\end{equation}
\item The bound 
	\begin{equation}\label{eq:activations_assumptions_mtheta_largeInputEstimate}
		\left\| \Multiplication_\theta(\mathds{T},\feps) - \mathds{T}:\feps \right\| \leq 2 \, \|\feps\|
	\end{equation}
	is valid for all tensors $\mathds{T} \in L(\Sym{d})$ and $\feps \in \Sym{d}$ satisfying the constraints
	\begin{equation}\label{eq:activations_assumptions_mtheta_largeInputEstimate_prereq}
		\| \mathds{T} \| \leq 1 \quad \text{and} \quad \| \feps\| \geq M.
	\end{equation}
\item The Lipschitz estimate
	\begin{equation}\label{eq:activations_assumptions_mtheta_contraction}
		\left\| \Multiplication_\theta(\mathds{T},\feps_1) - \Multiplication_\theta(\mathds{T},\feps_2) \right\| \leq \left( \|\mathds{T}\| + \delta_0 \right)\, \| \feps_1 - \feps_2\| 
	\end{equation}
	holds for all strains $\feps_1, \feps_2 \in \Sym{d}$ and all tensors $\mathds{T} \in L(\Sym{d})$ obeying the constraint
	\begin{equation}\label{eq:activations_assumptions_mtheta_contraction_prereq}
		\| \mathds{T} \| \leq 1.
	\end{equation}
\end{enumerate}
\end{prop}
\begin{proof}
	By Lemma \ref{lem:activations_assumptions}, ineq.~\eqref{eq:activations_assumptions_lem_Lipschitz}, the Lipschitz estimate \eqref{eq:activations_assumptions_mtheta_contraction} is satisfied in case the inequality
	\begin{equation}\label{eq:activations_assumptions_proof1}
		M\delta d(d+1) \leq \delta_0
	\end{equation}
	holds in view of the bound
	\begin{equation}\label{eq:activations_assumptions_proof2}
		\max_{i,j}|T_{ij}| \leq  \| \mathds{T} \|.
	\end{equation}
	By Lemma \ref{lem:activations_assumptions}, ineq.~\eqref{eq:activations_assumptions_lem_multiplication_proximity_small}, the small-input estimate \eqref{eq:activations_assumptions_mtheta_smallInputEstimate} holds provided the inequality 
	\begin{equation}\label{eq:activations_assumptions_proof3}
		M\delta d(d+1) \leq \delta_0
	\end{equation}
	is satisfied. For the estimate \eqref{eq:activations_assumptions_mtheta_largeInputEstimate} to hold, a necessary condition reads
	\begin{equation}\label{eq:activations_assumptions_proof4}
		\delta M d (d+1) \leq 1
	\end{equation}
	via Lemma \ref{lem:activations_assumptions}, ineq.~\eqref{eq:activations_assumptions_lem_multiplication_proximity_large}. Assuming $\delta_0 \leq 1$ it is thus sufficient to select the parameter $\delta$ to satisfy the constraint
	\begin{equation}\label{eq:activations_assumptions_proof5}
		\delta \leq \frac{\delta_0}{M d(d+1)}.
	\end{equation}
\end{proof}


\section{Computational examples}
\label{sec:computations}

\subsection{Setup}
\label{sec:computations_setup}

The fixed point $\feps^*$ of the Lippmann-Schwinger FNO \eqref{eq:FNO_LS_FLSO} may be interpreted as a compatible field
\begin{equation}\label{eq:computations_setup_compatibleStrain}
	\feps^*_\theta = \bar{\feps} + \nabla^s \fu^*_\theta,
\end{equation}
which solves the balance of linear momentum
\begin{equation}\label{eq:computations_setup_BLM}
	\div \fsigma_\theta\left( \bar{\feps} + \nabla^s \fu^*_\theta; \C \right) = \bm{0}
\end{equation}
with the neural stress operator
\begin{equation}\label{eq:computations_setup_stressOp}
	\fsigma_\theta: \Sym{d} \times \mathcal{M}_Y(\alpha_-,\alpha_+) \rightarrow \Sym{d}, \quad \fsigma_\theta(\feps;\C) = \alpha_0 \Multiplication_\theta \left( \frac{1}{\alpha_0}(\C - \C^0), \feps \right) + \C^0:\feps,
\end{equation}
involving the reference constant $\alpha_0$ which we select as in eq.~\eqref{eq:homogenization_basic_refMat_basic}. In case the parameter $\delta$ is selected to be small enough, the stress operator \eqref{eq:computations_setup_stressOp} is strongly monotone and Lipschitz continuous in the strain variable, as can be seen from the elementary estimates
\begin{align}
	\left\| \fsigma_\theta(\feps_1;\C) - \fsigma_\theta(\feps_2;\C) \right\| &\leq \alpha_0(1 + \|\mathds{T}_0\| + \delta) \|\feps_1 - \feps_2\|, \label{eq:computations_setup_stressOp_Lipschitz}\\
	\left[ \fsigma_\theta(\feps_1;\C) - \fsigma_\theta(\feps_2;\C) \right]:(\feps_1 - \feps_2) &\geq \alpha_0(1 - \|\mathds{T}_0\| - \delta) \|\feps_1 - \feps_2\|^2, \label{eq:computations_setup_stressOp_monotone}
\end{align}
valid for all strains $\feps_1,\feps_2 \in \Sym{d}$, implied by the condition \eqref{eq:ass_mtheta_contraction} and where we use the abbreviation
\begin{equation}\label{eq:computations_setup_stressOp_T0}
	\mathds{T}_0 = \frac{1}{\alpha_0}(\C - \C^0). 
\end{equation}
In particular, the balance of linear momentum \eqref{eq:computations_setup_BLM} admits a unique solution. Also, for any discretization on a regular grid compatible to the FFT framework, like the Moulinec-Suquet discretization~\cite{MoulinecSuquet1994,MoulinecSuquet1998,MSConvergence2023,Vidyasagar2017}, the finite-difference discretizations~\cite{Willot,StaggeredGrid} or the finite-element discretizations~\cite{FFTFEM,FANS,Ladecky2022FEM}, the corresponding form of the balance equation \eqref{eq:computations_setup_BLM} admits a unique solution. From a mechanical point of view, the constitutive law \eqref{eq:computations_setup_stressOp} is nonlinear elastic, but not hyperelastic, and approximates a linear elastic material model. We study the effects of the approximation \eqref{eq:computations_setup_BLM} to the proper model \eqref{eq:homogenization_setup_equilibrium} in the section at hand.\\
For this purpose, we integrated the constitutive law \eqref{eq:computations_setup_stressOp} into an in-house FFT-based computational micromechanics code~\cite{BB2019,QuasiNewton2019}. This allows us to benchmark the proposed FNO model with the FFT-based framework. In both methods, the volume element is discretized by the rotated staggered grid~\cite{Willot,StaggeredGrid} and solved by the basic scheme~\cite{MoulinecSuquet1994,MoulinecSuquet1998} with a termination error of $10^{-5}$, unless stated otherwise. 

\subsection{Influence of ReLU network depth}

\begin{figure}
	\centering
	\begin{subfigure}{.24\textwidth}
		\centering
		\includegraphics[width=\textwidth]{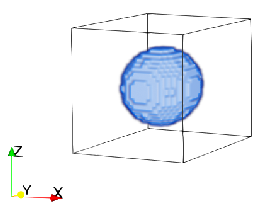}
		\caption{$32 \times 32 \times 32$}
		\label{fig:balls_32}
	\end{subfigure}
	\begin{subfigure}{.24\textwidth}
		\centering
		\includegraphics[width=\textwidth]{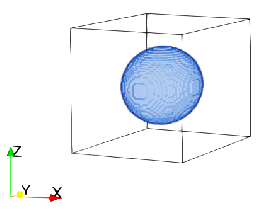}
		\caption{$64 \times 64 \times 64$}
		\label{fig:balls_64}       
	\end{subfigure}
	\begin{subfigure}{.24\textwidth}
		\centering
		\includegraphics[width=\textwidth]{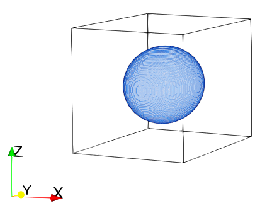}
		\caption{$128 \times 128 \times 128$}
		\label{fig:balls_128}      
	\end{subfigure}
	\begin{subfigure}{.24\textwidth}
		\centering
		\includegraphics[width=\textwidth]{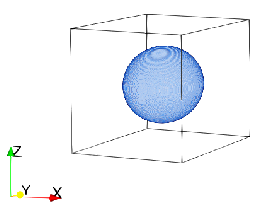}
		\caption{$256 \times 256 \times 256$}
		\label{fig:balls_256}
	\end{subfigure}
	\caption{Volume elements of edge length $32 \mu m$ with a spherical inclusion of radius $10\mu m$.}
	\label{fig:ball}
\end{figure}	

In this section, we demonstrate the influence of the ReLU network depth, denoted by the letter $\mathfrak{m}$, of the double-contraction operator \eqref{eq:FNO_LS_doubleContraction} for a simple microstructure. Specifically, we consider a single spherical inclusion of radius 10 $\mu m$ embedded in a matrix material, modeled as a cubic cell of edge length 32 $\mu m$. The volume element is discretized by $32^3$, $64^3$, $128^3$ and $256^3$ voxels, corresponding to voxel lengths of $1\mu m$, $0.5\mu m$, $0.25\mu m$, $0.125\mu m$, respectively, as shown in Fig.~\ref{fig:ball}. The matrix material is assumed to have fixed Young's modulus and Poisson's ratio as given in Tab. \ref{tab:materialParameters_ball} together with the fixed Poisson's ratio of the inclusion. We assume the material contrast $\kappa$, defined by the ratio between the Young's modulus of the inclusion and the matrix, ranges from 12 to 96. All numerical examples in this section were computed on desktop computer with 12-core Intel i7-1360P CPU.
\begin{table}[H]
	\centering
	\caption{Material parameters of the single spherical inclusion. The Young's modulus of the ball is varied by the material contrast $\kappa=12,24,48,96$.}
	\begin{tabular}{l l l}
		\hline
		Matrix & $E=3.0$ GPa, & $\nu=0.3$\\
		Inclusion(s) & \review{$\kappa E$},  & $\nu=0.22$\\		
		\hline
	\end{tabular}
	\label{tab:materialParameters_ball}
\end{table}

\subsubsection{On the effective stiffness}

This example investigates the performance of the FNO model in computing the effective stiffness matrix of the volume element shown in Fig.~\ref{fig:balls_32}, which resolves the volume element with $32^3$ voxels. To determine the effective stiffness matrix, we prescribe six independent macroscopic strains  $\bar{\feps}$ of magnitude $0.1\%$ unless stated otherwise, and solve eq.\,\eqref{eq:computations_setup_BLM} with the neural stress operator $\fsigma_\theta$, see eq.~\eqref{eq:computations_setup_stressOp}, defined via the double-contraction operator $\Multiplication_\theta$ given in eq.~\eqref{eq:activations_double_contraction_construction2}. As described in Section \ref{sec:activations}, we construct the double-contraction operator by the ReLU network. The expressivity of the ReLU network is characterized by its depth $\mathfrak{m}$, which we consider in this paper as $\mathfrak{m}=7$, $\mathfrak{m}=9$ and $\mathfrak{m}=11$, denoted henceforth as FNO7, FNO9 and FNO11, respectively. In addition, we also determine the effective stiffness of the proper model \eqref{eq:homogenization_setup_equilibrium} by the FFT-based homogenization method and use the obtained results as the reference to quantify the \textit{approximation error} \eqref{eq:FNO_result_error_second} of the FNO models.
\begin{table}[h!]
	\caption{Effective elastic moduli and the relative approximation errors corresponding to different FNO models and varying material contrasts.}
	\label{tab:effective_stiffness}
	\centering
	\begin{tabular}{cc|r|rrr|rrr}
		\cline{4-9}
		\multicolumn{1}{c}{}  & \multicolumn{1}{c}{} & \multicolumn{1}{r}{}  & \multicolumn{3}{c|}{FNO prediction}  &\multicolumn{3}{c}{Approximation error ($\%$)} \\ \cline{4-9}
		\hline
		& Contrast &FFT & FNO7 & FNO9 & FNO11   & FNO7 & FNO9 & FNO11 \\
		\hline \hline 
		$C_{11}$ (GPa)& 12 & 5.0083 & 4.5010 & 4.9903 & 5.0046 & 10.1292 & 0.3594 & 0.0739\\
		& 24 & 5.1309 & 4.3648 & 4.9710 & 5.1254 & 14.9311 &  3.1164 & 0.1072 \\
		& 48 & 5.2002 & 4.9845 & 5.1100 & 5.1824 & 4.1479 & 1.7345 & 0.3422 \\
		& 96 & 5.2376 & 5.1286 & 5.1398 & 5.2191 & 2.0811 & 1.8673 &  0.3532  \\	
		\hline \hline 		 			 		
		$C_{12}$ (GPa) & 12 &  1.9884 & 1.7670 & 1.9365 & 1.9850 & 11.1346 & 2.6101 & 0.1710 \\
		& 24 & 2.0208 & 1.1810 & 1.8410 & 2.0159 & 41.5578 & 8.8975 &  0.2425 \\
		& 48 & 2.0378 & 0.0622 & 1.7679 &2.0202 &96.9477 & 13.2447 & 0.8637\\
		& 96 & 2.0465 & 0.0409 & 1.3164 & 2.0254 & 98.0015 & 35.6755& 1.0310\\	
		\hline \hline 
		$C_{44}$ (GPa)  & 12 & 2.8770 & 2.6125& 2.7677& 2.8765& 9.1936 & 3.7991 & 0.0174  \\
		& 24 & 2.9338 & 2.6198 &2.9264 &2.9293& 10.7028 & 0.2522& 0.1534   \\
		& 48 & 2.9654 & 2.6365 &2.6223& 2.9501 &11.0913 & 11.5701& 0.5160   \\
		& 96 &2.9822 &0.0002& 2.6378& 2.9604& 99.9933 & 11.5485& 0.7310\\						
		\hline  	
	\end{tabular}
\end{table}

The results of effective elastic coefficients $C_{11}, C_{12}$ and $C_{44}$  and their corresponding approximation errors are reported in Tab. \ref{tab:effective_stiffness} and depicted in Fig.~\ref{fig:singlesphere_stiffness}.
We observe the following. The depth $\mathfrak{m}$ of the ReLU network is decidedly important in approximating the double contraction operation and hence the microscopic solution of eq.\,\eqref{eq:computations_setup_BLM}. For the depth $\mathfrak{m}=11$, the FNO11 model achieves at most $1\%$ approximation error as compared to the FFT-based homogenization method. The error magnitudes are also observed to be monotonically increasing with higher material contrasts $\kappa$. In the case $\mathfrak{m}=9$, the constructed FNO9 model may still yield good predictions of effective elastic coefficients with less than $5\%$ error for material contrast $\kappa=12$. For higher contrast values, for instance $\kappa=96$, the FNO9 model's approximation leads to an error of $35.7\%$ for the effective coefficient $C_{12}$. For shallower depth $\mathfrak{m}=7$, the FNO7 model essentially fails to predict the effective elastic tensor. Even for a microstructure at low contrast $\kappa=12$, this model produces an approximation error of around $10\%$. Meanwhile, at the contrast levels $\kappa=48$ or $\kappa=96$, the FNO7 model only resolves the normal coefficient $C_{11}$ with the approximation errors of $4.1\%$ and $2.1\%$, respectively. On the other hand, it is unable to approximate the elastic coefficients $C_{12}$ and $C_{44}$, yielding values close to zeros and approximation errors close to $100\%$.
\begin{figure}[h!]
	\centering
	\begin{subfigure}{.3\textwidth}
		\centering
		\includegraphics[width=\textwidth]{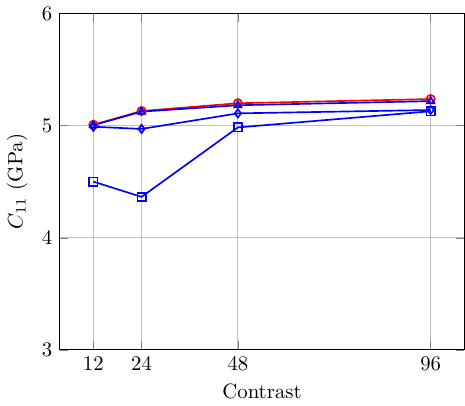}
		\caption{$C_{11}$}
		\label{fig:singlesphere_c11}
	\end{subfigure}
	\hfill
	\begin{subfigure}{.3\textwidth}
		\centering
		\includegraphics[width=\textwidth]{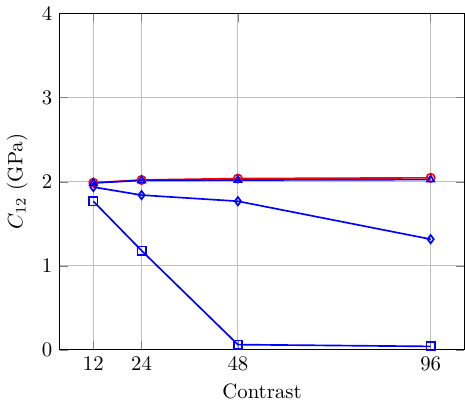}
		\caption{$C_{12}$}        
		\label{fig:singlesphere_c12}        
	\end{subfigure}
	\hfill
	\begin{subfigure}{.3\textwidth}
		\centering
		\includegraphics[width=\textwidth]{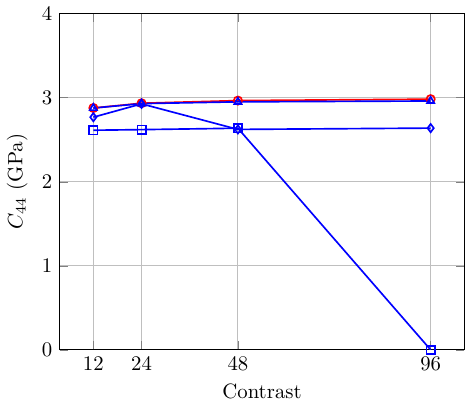}
		\caption{$C_{44}$}
		\label{fig:singlesphere_c44}        
	\end{subfigure}
	
	\begin{subfigure}{.3\textwidth}
		\centering
		\includegraphics[width=1.0\linewidth]{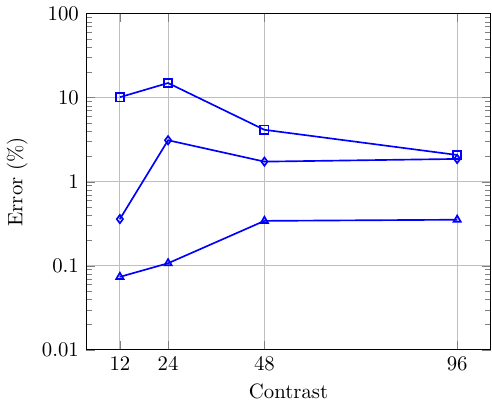}
		\caption{$C_{11}$}
		\label{fig:singlesphere_c11_error}
	\end{subfigure}
	\hfill
	\begin{subfigure}{.3\textwidth}
		\centering
		\includegraphics[width=1.0\linewidth]{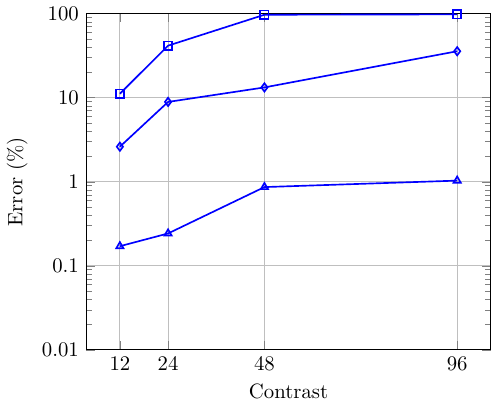}
		\caption{$C_{12}$}        
		\label{fig:singlesphere_c12_error}        
	\end{subfigure}
	\hfill
	\begin{subfigure}{.3\textwidth}
		\centering
		\includegraphics[width=1.0\linewidth]{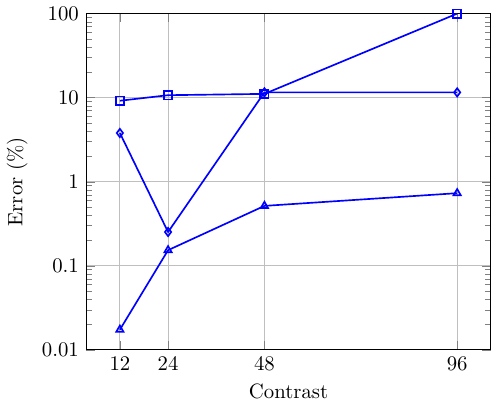}
		\caption{$C_{44}$}
		\label{fig:singlesphere_c44_error}        
	\end{subfigure}
	
	\begin{center}
		\begin{subfigure}{\textwidth}
			\centering
			\includegraphics[height=.022\textheight]{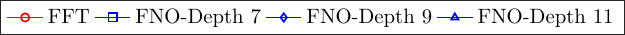}
		\end{subfigure}
	\end{center}
	\caption{Performance of FNO models for varying material contrasts. (a-b-c) Effective elastic coefficients. (d-e-f) Relative approximation errors.}
	\label{fig:singlesphere_stiffness}
\end{figure}

\subsubsection{On the number of iterations}

We investigate the influence of the ReLU-network depth $\mathfrak{m}$ on the number of Fourier neural layers $K$ or, equivalently, the number of iterations of the basic scheme. We re-use the computational results of the previous example by extracting the number of iterations of the first load case, i.e., uniaxial extension in the $x-$direction at $0.1\%$ strain, as well as the computational time of the application of the double-contraction operator $\Multiplication_\theta$. The obtained results from both the FFT-based method and different FNO models are summarized in Tab. \ref{tab:effective_iterations} and depicted in Fig.~\ref{fig:singlesphere_iterations_materialTime}. We make the following observations. The number of required iterations in the FNO models is doubled when doubling the material contrast, just as the FFT-based method, except for the FNO7 model. As observed in the previous example, the FNO7 model not only predicts less accurate solutions but also consistently requires more iterations than the FFT-based method. Especially, in case of high contrast $\kappa=96$, this model needs almost 21 times the number of iterations and 50 times the runtime than the FFT-based method. The deficiency of the FNO7 model in this scenario is attributed by its large approximation error that makes the constructed FNO no longer a contraction, or - at least - converges more slowly. Using deeper ReLU networks improve the performance of the FNO model significantly. Indeed, in the case $\mathfrak{m}=9$, the FNO9 model requires the idential number of iterations as the FFT-based method up to material contrast of 24, with additional $10\%$ to $15\%$ more iterations for the higher contrasts $\kappa=48$ and $\kappa=96$, respectively. Meanwhile, the FNO11 model performs practically as good as the FFT method in terms of number of iterations. However, both models FNO9 and FNO11 consume approximately the same computational time, which is three times higher than for the FFT-based method for material contrasts $\kappa=48$ and $\kappa=96$.
\begin{table}
	\caption{Counted numbers of iterations and runtimes of the FNO models for varying material contrasts.}
	\label{tab:effective_iterations}
	\centering
	\begin{tabular}{crrrr|rrrr}
		\cline{2-9}
		& \multicolumn{4}{c|}{Iterations}  & \multicolumn{4}{c}{Runtime (seconds)} \\ \cline{2-9}
		\hline
		Contrast &FFT & FNO7 & FNO9 & FNO11 & FFT  & FNO7 & FNO9 & FNO11\\
		\hline \hline 
		12 & 92 & 104 & 92 & 92 & 7.3 & 41.3  & 26.9  & 40.3\\
		24 & 181 & 207 & 183 & 181 & 8.4 & 68.2 & 51.6	& 52.1	 	\\
		48 & 360 & 426 & 393 & 360 & 32.3 & 104.0 & 102.3 & 96.7			\\
		96 & 716 & 14675 & 825 & 719  & 64.1 & 3234.8 & 208.3	&206.3	\\	 	
		\hline  	
	\end{tabular}
\end{table}
\begin{figure}[h!]
	\centering
	\begin{subfigure}{.4\textwidth}
		\centering
		\includegraphics[width=\textwidth]{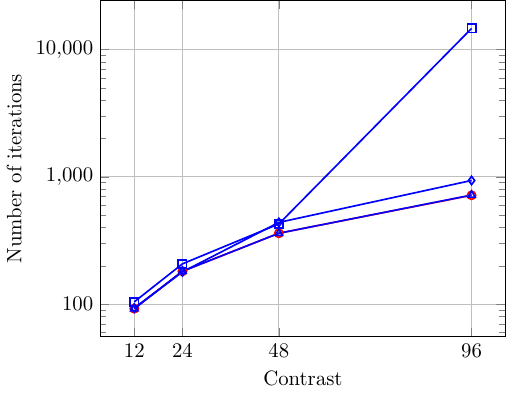}
		\caption{Iterations}
		\label{fig:singlesphere_iterations_load1}
	\end{subfigure}
	\hfill
	\begin{subfigure}{.4\textwidth}
		\centering
		\includegraphics[width=\textwidth]{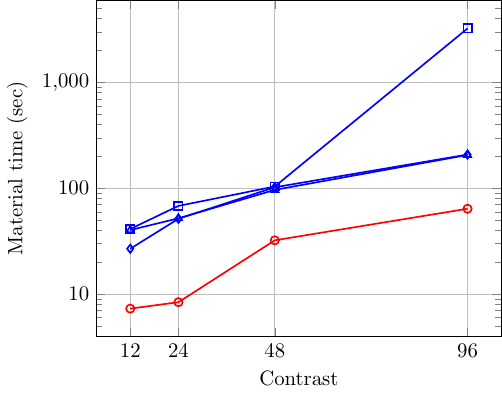}
		\caption{Material time}
		\label{fig:singlesphere_materialTime_loadcase1}
	\end{subfigure}	
	\begin{center}
		\begin{subfigure}{\textwidth}
			\centering
			\includegraphics[height=.022\textheight]{results/singleSphere/effective/legend_stiffness.pdf}
		\end{subfigure}
	\end{center}
	\caption{Iteration counts and runtimes of the FNO models for varying material contrasts.}
	\label{fig:singlesphere_iterations_materialTime}
\end{figure}

\subsubsection{On the resolution}

We turn our attention to the influence of the resolutions of the discretized microstructure, described in Fig.~\ref{fig:ball}, on the performance of the neural stress operator $\fsigma_\theta$ or, equivalently, the double-contraction operator $\Multiplication_\theta$. For a fixed material contrast $\kappa=24$, we study the solution of eq.\,\eqref{eq:computations_setup_BLM} for a prescribed uniaxial macroscopic strain $\bar{\feps}$ in $x-$direction and subsequently compute the effective stiffness coefficient with its respective error as well as monitor the number of iterations and runtime. Similar to the previous example, we take the results from the FFT-based method as the reference. We compute the effective elastic coefficient $C_{11}$ and its approximation errors for varying resolutions and report them in Tab. \ref{tab:resolution_effective} and in Fig.~\ref{fig:singlesphere_resolution_effective}. We observe that the volume element resolution has little influence on the FNO models as the effective elastic coefficient obtained from different FNO models does not change significantly from the coarse  to fine discretizations. Nevertheless, the accuracy of these models are similar to previous observations, i.e., FNO model with deeper ReLU network yields better accuracy. Specifically, regardless of the resolution, the approximation error decreases from about $15\%$ to $3\%$ and $0.1\%$ as the ReLU network depth $\mathfrak{m}$ increases from 7 to 9 and 11, respectively.
\begin{table}[h!]
	\caption{Effective coefficient $C_{11}$ and its relative approximation errors of FNO models for different resolutions.}
	\label{tab:resolution_effective}
	\centering
	\begin{tabular}{c|r|rrr|rrr}
		\cline{3-8}
		\multicolumn{1}{c}{} & \multicolumn{1}{r}{}  & \multicolumn{3}{c|}{FNO prediction}  &\multicolumn{3}{c}{Approximation error ($\%$)} \\ \cline{3-8}
		\hline
		Resolution &FFT &FNO7 & FNO9 & FNO11   & FNO7 & FNO9 & FNO11\\
		\hline \hline 
		32 & 5.1309 & 4.3648 & 4.9710 & 5.1254 & 14.9303 &  3.1155 & 0.1060 \\		
		64 & 5.1036 &  4.3382 & 4.9441 & 5.0982 & 14.9977 & 3.1253 & 0.1052	 \\
		128 & 5.0912 & 4.3259 & 4.9318 & 5.0858 & 15.0303& 3.1308& 0.1050	 \\
		256 & 5.0856 & 4.3206 &4.92631& 5.0803& 15.0443& 3.1331&0.1048	 \\			
		\hline  	
	\end{tabular}
\end{table}
\begin{figure}
	\centering
	\begin{subfigure}{.4\textwidth}
		\centering
		\includegraphics[width=\linewidth]{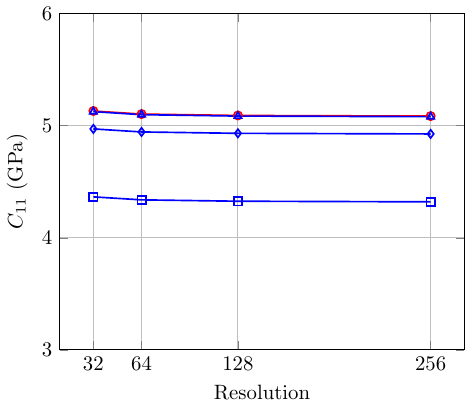}
		\caption{$C_{11}$}
		\label{fig:plot_c11_resolution}
	\end{subfigure}
	\hfill
	\begin{subfigure}{.4\textwidth}
		\centering
		\includegraphics[width=1.04\linewidth]{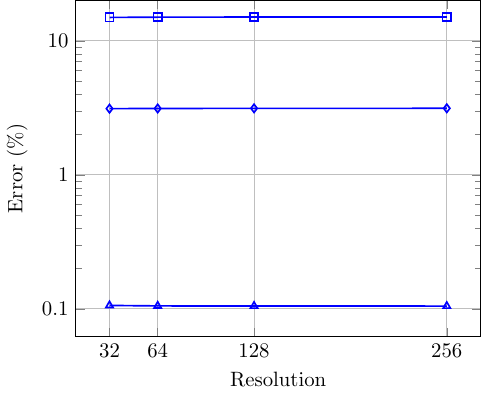}
		\caption{$C_{11}$ approximation error}
		\label{fig:plot_c11_error_resolution}
	\end{subfigure}
	\begin{center}
		\begin{subfigure}{\textwidth}
			\centering
			\includegraphics[height=.022\textheight]{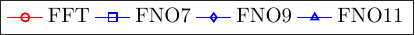}
		\end{subfigure}
	\end{center}
	\caption{Performance of FNO models for varying resolutions.}
	\label{fig:singlesphere_resolution_effective}
\end{figure}

We also evaluate the number of iterations and the computation time of the double-contraction operations for different FNO models and varying resolutions with the summarized results in Tab. \ref{tab:resolution_iterations}, shown in Fig.~\ref{fig:singlesphere_resolution_iterations}. The number of iterations required by different FNO models is not significantly sensitive to the volume element resolution. As observed in the previous example, the FNO7 model consistently requires more iterations than the other models. At the resolution of $32^3$ voxels, this model needs 207 iterations to reach the convergence, whereas the FNO9, the FNO11 and the FFT-based method need approximately 25 less iterations. As the resolution is doubled, the number of required iterations for all considered models approximately remain the same as for the resolution of $32^3$ voxels. This observation, however, does not hold for the runtime. Since the double-contraction operator $\Multiplication_\theta$ is invoked at each voxel, the increasing resolution leads to the increase in runtime. As the resolution is doubled in each direction of the volume element, the number of voxels increases by 8 times, which agrees with the observed increased runtime for the FNO models. The FFT-based method does not submit to this scaling as the double contraction is evaluated exactly. Therefore, although the FNO11 model perform essentially as good as the FFT-based method, its computational time is about 3, 11, 40 and 85 times higher than the FFT-based method for the resolution of $32^3$, $64^3$, $128^3$ and $256^3$, respectively.
\begin{table}
	\caption{Iteration counts and runtimes of the FNO models for varying resolutions.}
	\label{tab:resolution_iterations}
	\centering
	\begin{tabular}{crrrr|rrrr}
		\cline{2-9}
		& \multicolumn{4}{c|}{Iterations}  & \multicolumn{4}{c}{runtime (seconds)} \\ \cline{2-9}
		\hline
		Resolution &FFT &FNO7 & FNO9 & FNO11 & FFT  & FNO7 & FNO9 & FNO11\\
		\hline \hline 
		32 & 181 & 207 & 183 & 181 & 1.0 & 2.8  & 3.0  & 3.3\\
		64 & 178 & 206 & 179 & 178 & 2.2 & 21.0 & 22.4	& 24.6	 	\\
		128 & 175 & 205 & 177 & 175 & 4.8 & 160.6 & 169.8 & 194.3			\\
		256 & 173 & 205 & 175 & 174  & 22.5 & 1760.6 & 2040.7 & 1916.3	\\	 	
		\hline  	
	\end{tabular}
\end{table}
\begin{figure}[h!]
	\centering
	\begin{subfigure}{.4\textwidth}
		\centering
		\includegraphics[width=1.04\linewidth]{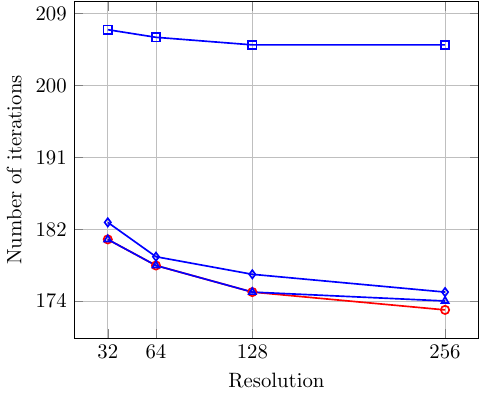}
		\caption{Number of iterations}
		\label{fig:iterations_loadcase1_plot}
	\end{subfigure}	\hfill
	\begin{subfigure}{.4\textwidth}
		\centering
		\includegraphics[width=1.08\linewidth]{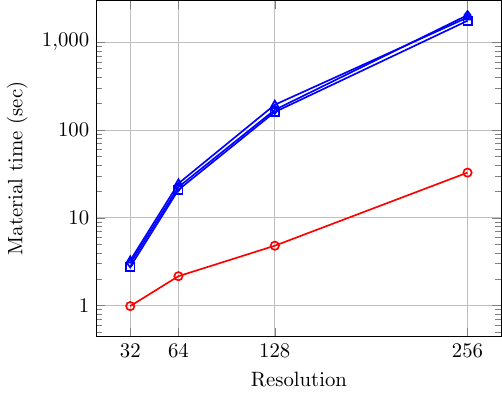}
		\caption{Runtime}        
		\label{fig:materialTime_resolution_plot}        
	\end{subfigure}
	\begin{center}
		\begin{subfigure}{\textwidth}
			\centering
			\includegraphics[height=.022\textheight]{results/singleSphere/resolutions/legend_stiffness.pdf}
		\end{subfigure}
	\end{center}
	\caption{Iteration counts and runtimes of the FNO models for varying resolutions.}
	\label{fig:singlesphere_resolution_iterations}
\end{figure}

\subsubsection{On the applied strain magnitude}

This section examines the performance of the FNO models under various prescribed macroscopic strain magnitudes. We consider the microstructure in Fig.~\ref{fig:balls_32} resolved by $32^3$ voxels with a low material contrast $\kappa=12$. The macroscopic strain component is applied in the $x-$direction with varied magnitudes $0.1\%, 1\%, 50\%$ and $100\%$. Although linear elasticity is considered in this article, the neural approximation \eqref{eq:computations_setup_stressOp} is non-linear. Thus, we study the accuracy of the double-contraction operator $\Multiplication_\theta$ with respect to the cut-off bound value $M$ which we fix as 1.0 in our construction. As usual, we employ FFT-based method as the reference. In case of linear material, the resulting effective elastic coefficient obtained by the FFT-based method is independent of the prescribed strain magnitudes. The resulting effective stiffness coefficient and its corresponding approximation errors are summarized in Tab. \ref{tab:magnitude_effective} and plotted in Fig.~\ref{fig:singlesphere_magnitude_effective}. We make the following observations. At the small strain magnitude of $0.1\%$, as seen from the previous examples, the FNO7 model exhibits the largest approximation error of $10.13\%$ as compared to $0.36\%$ and $0.07\%$ for the FNO9 and FNO11 models, respectively. As the prescribed strain magnitude increases, the approximation errors of all three FNO models decline. The smallest approximation errors are achieved for the prescribed strain magnitude of $50\%$, where the FNO11 model reaches $0.0003\%$ and even the FNO7 model attains close to $0.1\%$ errors. Such improvements of the approximation errors stem from the construction mechanism of the double-contraction operator: Because of the cutoff $M=1.0$, the low strain values are not well-approximated by the ReLU network, see Fig.~\ref{fig:ReLUapproximation}. This observation showcases the possibility to select a more robust FNO architecture for an a priori strain field. Nevertheless, all three models result in about $1\%$ of approximation error at $100.0\%$ applied strain. For a closer look, we present in Fig.~\ref{fig:singlesphere_magnitude_sigmavM} the distribution of the equivalent stress $\sigma^{\texttt{ev}}$ for the prescribed $50\%$ and $100\%$ macroscopic strain, whereas their corresponding absolute error distributions are presented in Fig.~\ref{fig:singlesphere_magnitude_absError_sigmavM}. Similar to the effective elastic coefficient, an excellent agreement between the FNO models and the FFT-based method is obtained in the case of $50\%$ strain, where the stress distributions are almost identical with the maximum absolute error of 1 MPa for the FNO7 model, corresponds to $0.1\%$ relative error. In the case of $100\%$ strain, the stress distributions of the FNO models become more visible from the FFT-based method. All three FNO models exhibit the highest error at the center of the sphere with absolute equivalent stress difference of about 550 MPa. However, it should be noted that such extreme load case of $100\%$ strain is infeasible in practical elastic homogenization problems and we merely employ it to evaluate the behavior of the double-contraction operator. Also, we do not survey the number of iterations and computational time in this example because they are independent of the applied strain magnitude.

\begin{table}[h!]
	\caption{Effective coefficient $C_{11}$ and its relative approximation errors of FNO models for prescribed macroscopic strain magnitudes.}
	\label{tab:magnitude_effective}
	\centering
	\begin{tabular}{c|r|rrr|rrr}
		\cline{3-8}
		\multicolumn{1}{c}{} & \multicolumn{1}{r}{}  & \multicolumn{3}{c|}{FNO prediction}  &\multicolumn{3}{c}{Approximation error ($\%$)} \\ \cline{3-8}
		\hline
		Magnitude ($\%$) &FFT & FNO7 & FNO9 & FNO11   & FNO7 & FNO9 & FNO11\\
		\hline \hline 
		0.1 & 5.0083 & 4.5010 & 4.9903 & 5.005 & 10.1295  & 0.3602 & 0.0749	 \\		
		1.0 & 5.0083 & 4.8880 & 4.9989 &5.0076& 2.4019& 0.1889& 0.0151			 \\
		50.0 & 5.0083 & 5.0034& 5.0081& 5.0083 & 0.0993 &0.0047& 0.0003			 \\
		100.0 & 5.0083 & 5.0563& 5.0569 &5.0566 &0.9574& 0.9703 &0.9639 	 \\			
		\hline  	
	\end{tabular}
\end{table}

\begin{figure}
	\centering
	\begin{subfigure}{.4\textwidth}
		\centering
		\includegraphics[width=\linewidth]{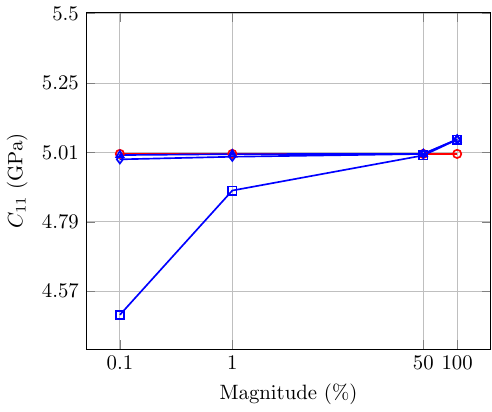}
		\caption{$C_{11}$}
		\label{fig:plot_c11_magnitude}
	\end{subfigure}
	\hfill
	\begin{subfigure}{.4\textwidth}
		\centering
		\includegraphics[width=1.04\linewidth]{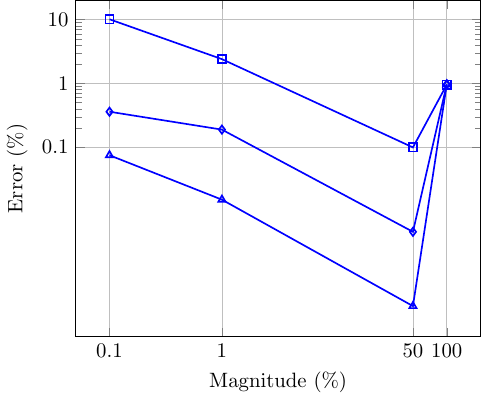}
		\caption{$C_{11}$ approximation error}
		\label{fig:plot_c11_error_magnitude}
	\end{subfigure}
	\begin{center}
		\begin{subfigure}{\textwidth}
			\centering
			\includegraphics[height=.022\textheight]{results/singleSphere/resolutions/legend_stiffness.pdf}
		\end{subfigure}
	\end{center}
	\caption{Performance of FNO models for varying prescribed strain magnitudes.}
	\label{fig:singlesphere_magnitude_effective}
\end{figure}
\begin{figure}
		\begin{subfigure}{.24\textwidth}
		\centering
		\includegraphics[width=\textwidth,trim=2cm 4.35cm 39cm 7cm,clip]{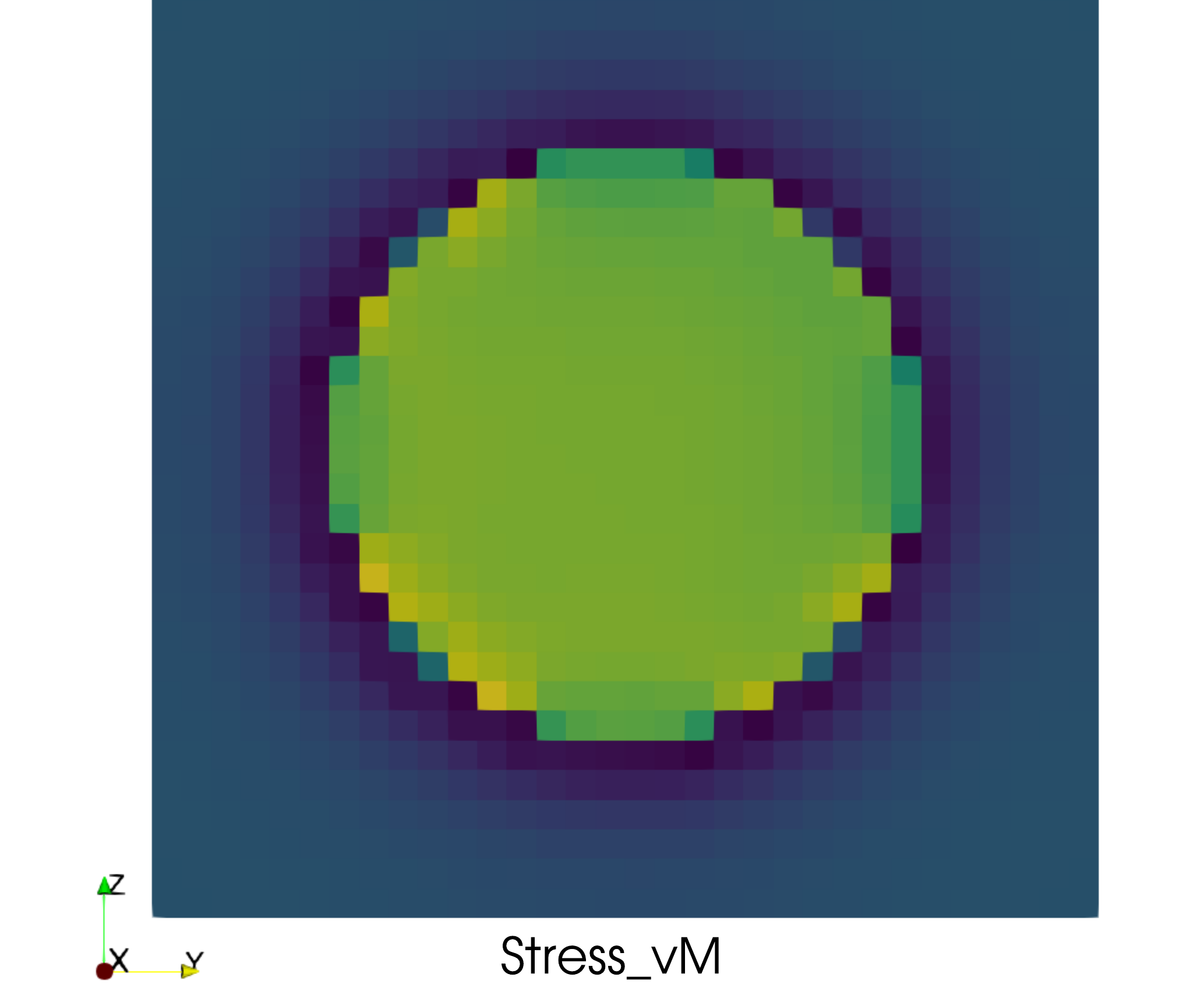}
		\caption{FFT - $\bar{\varepsilon}_{11}=50\%$}
		\label{fig:ball_FFT_stressvM_load50}
	\end{subfigure}
	\hfill
	\begin{subfigure}{.24\textwidth}
		\centering
		\includegraphics[width=\textwidth,trim=2cm 4.35cm 39cm 7cm,clip]{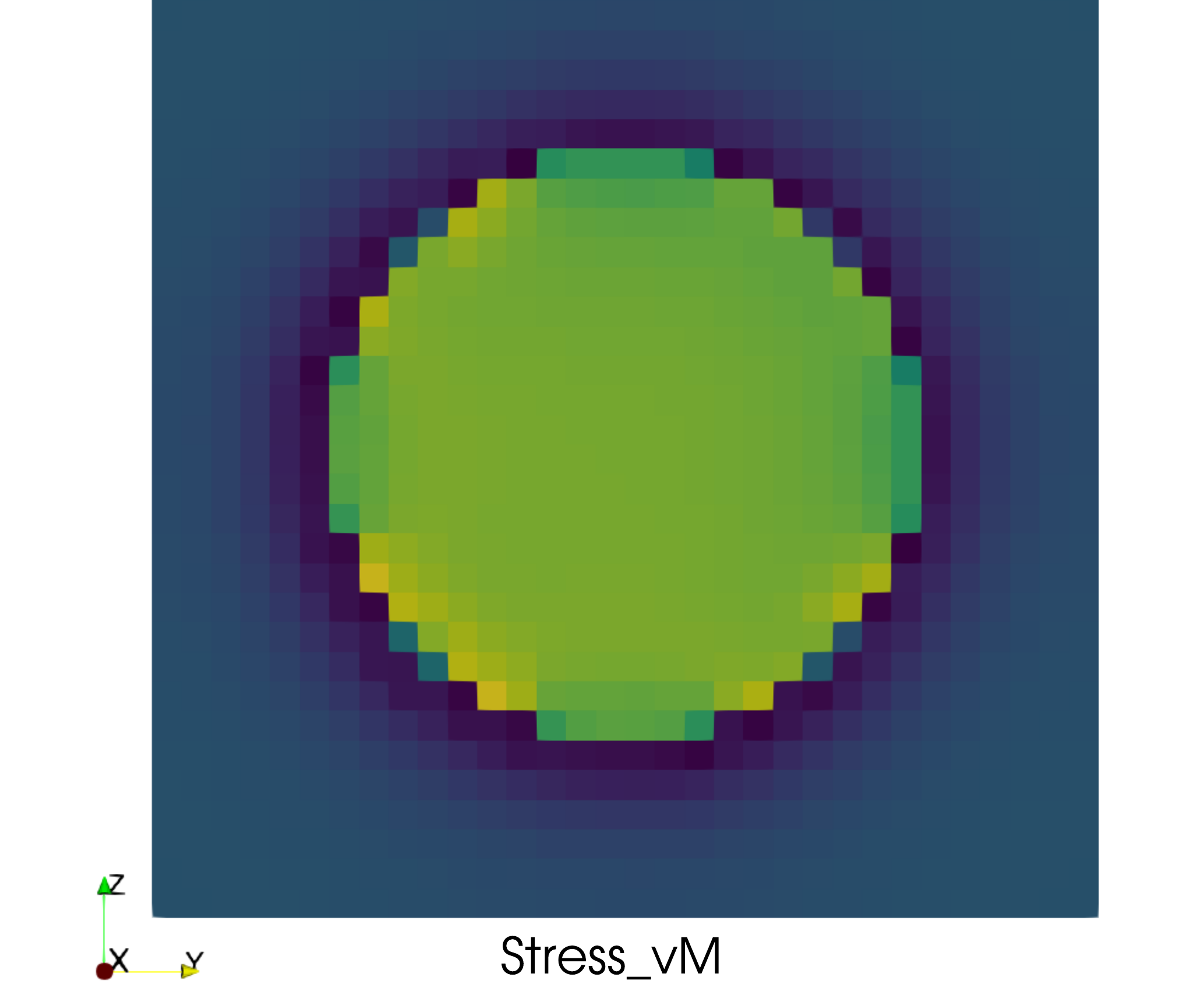}
		\caption{FNO7 - $\bar{\varepsilon}_{11}=50\%$}
		\label{fig:ball_FNO7_stressvM_load50}
	\end{subfigure}
	\hfill
	\begin{subfigure}{.24\textwidth}
		\centering
		\includegraphics[width=\textwidth,trim=2cm 4.35cm 39cm 7cm,clip]{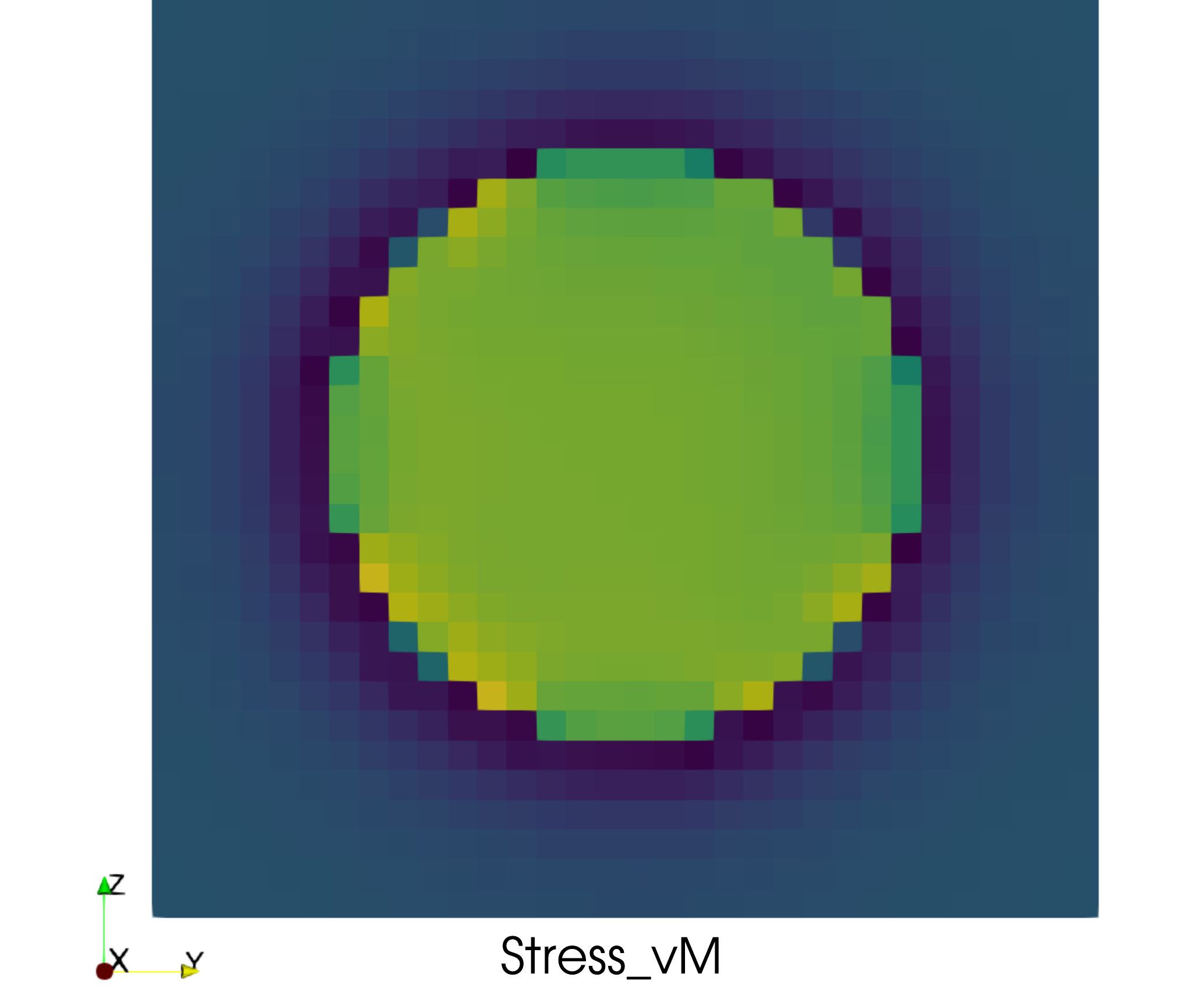}
		\caption{FNO9 - $\bar{\varepsilon}_{11}=50\%$}
		\label{fig:ball_FNO9_stressvM_load50}
	\end{subfigure}
	\hfill
	\begin{subfigure}{.24\textwidth}
		\centering
		\includegraphics[width=\textwidth,trim=2cm 4.35cm 39cm 7cm,clip]{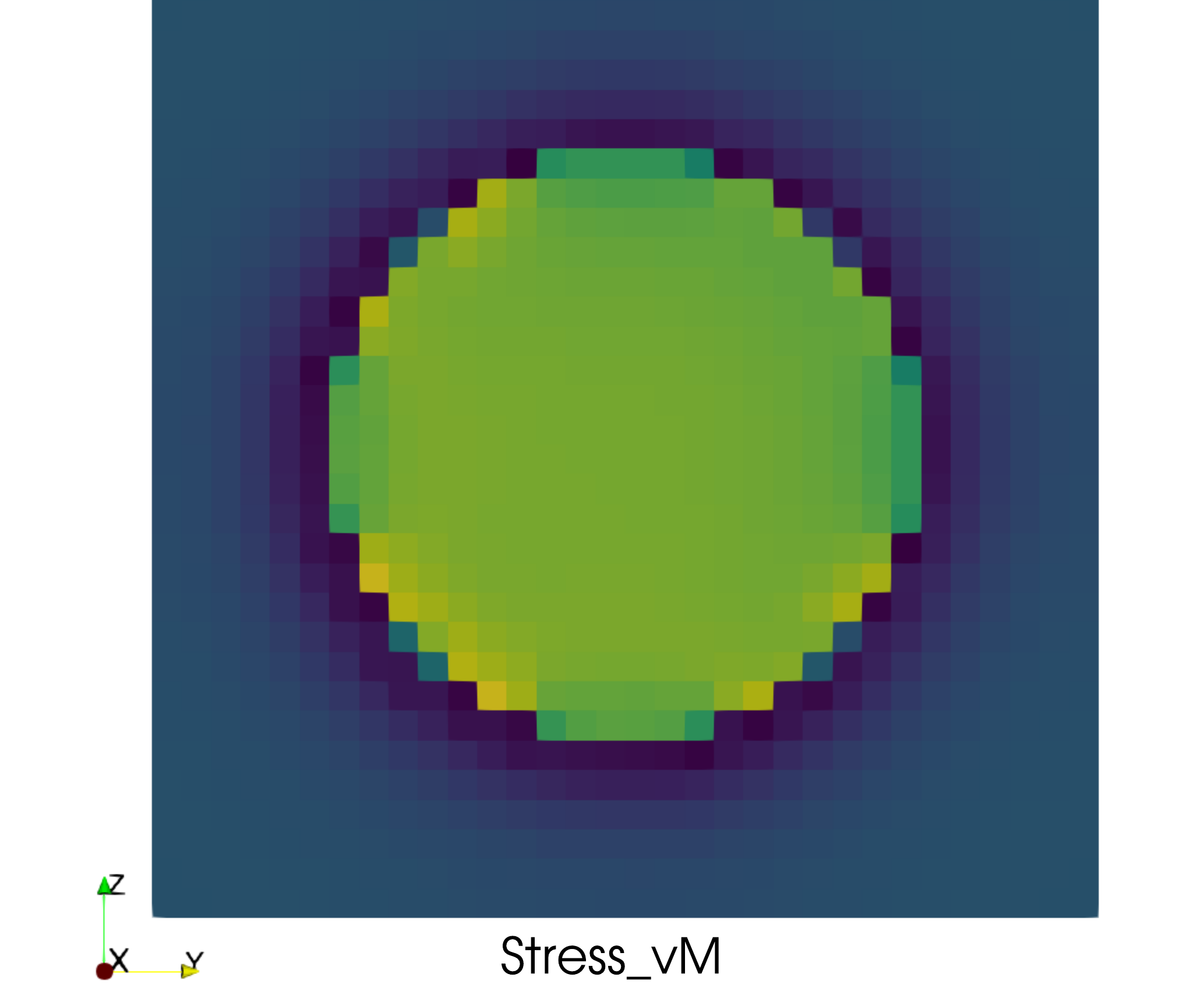}
		\caption{FNO11 - $\bar{\varepsilon}_{11}=50\%$}
		\label{fig:ball_FNO1_stressvM_load50}
	\end{subfigure}
	\vspace{-1em}
	\begin{center}
		\begin{subfigure}{\textwidth}
			\centering
			\includegraphics[height=.033\textheight]{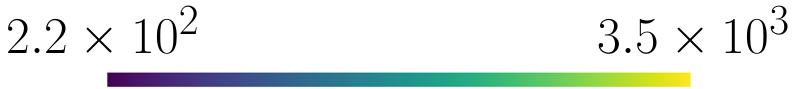}
		\end{subfigure}
	\end{center}
	
	\begin{subfigure}{.24\textwidth}
		\centering
		\includegraphics[width=\textwidth,trim=2cm 4.35cm 39cm 7cm,clip]{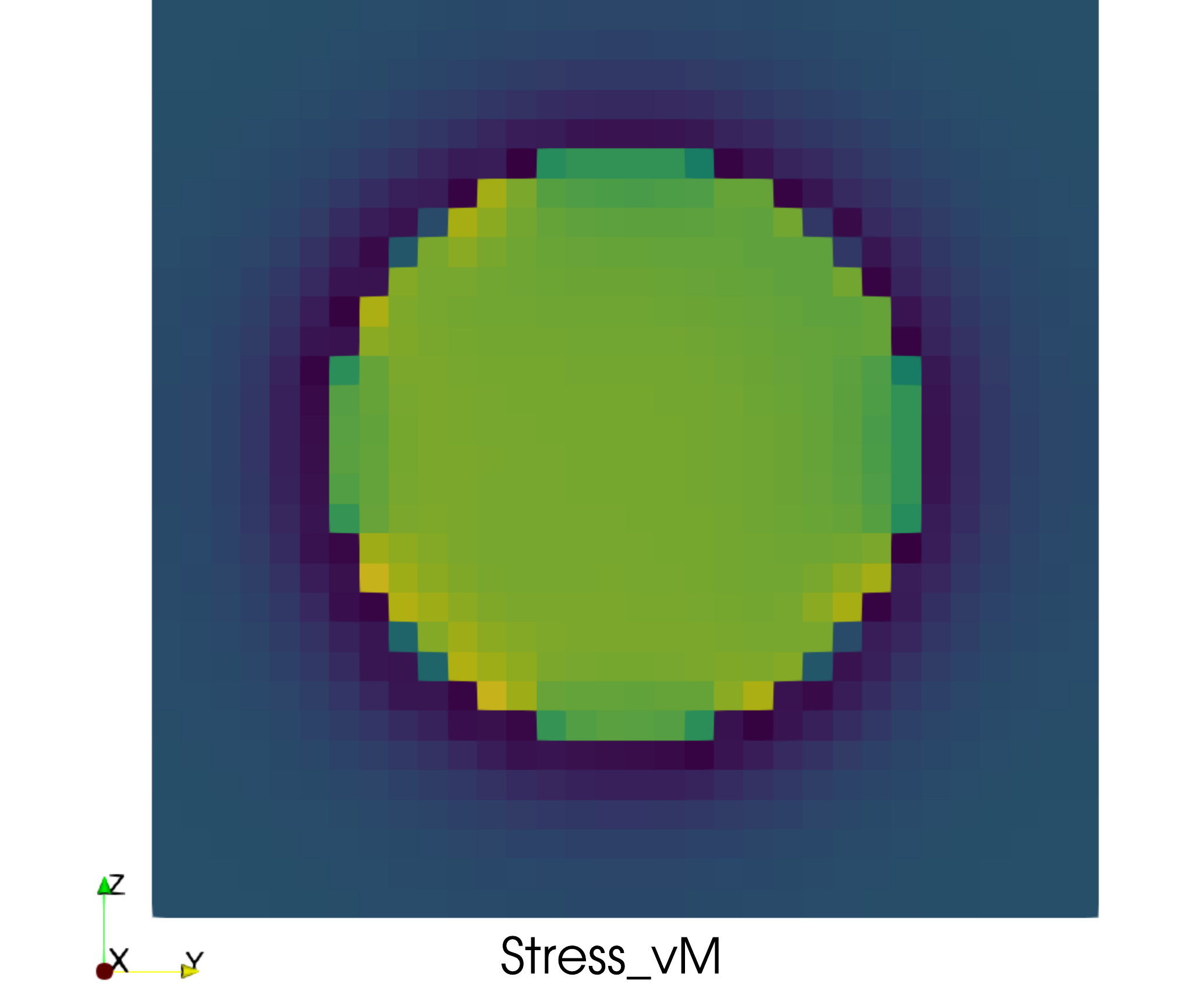}
		\caption{FFT - $\bar{\varepsilon}_{11}=100.0\%$}
		\label{fig:ball_FFT_stressvM_load100}
	\end{subfigure}
	\hfill
	\begin{subfigure}{.24\textwidth}
		\centering
		\includegraphics[width=\textwidth,trim=2cm 4.35cm 39cm 7cm,clip]{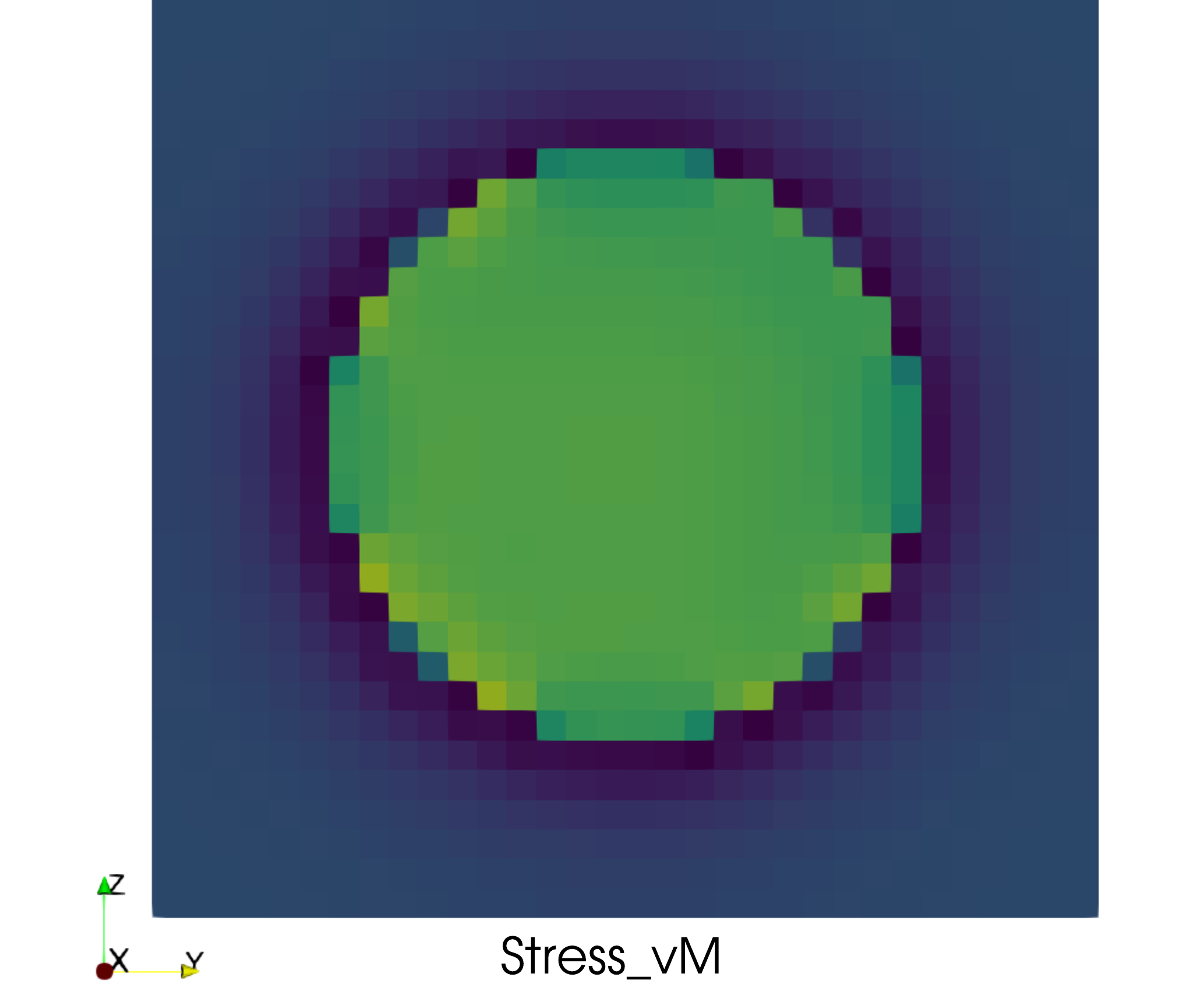}
		\caption{FNO7 - $\bar{\varepsilon}_{11}=100.0\%$}
		\label{fig:ball_FNO7_stressvM_load100}
	\end{subfigure}
	\hfill
	\begin{subfigure}{.24\textwidth}
		\centering
		\includegraphics[width=\textwidth,trim=2cm 4.35cm 39cm 7cm,clip]{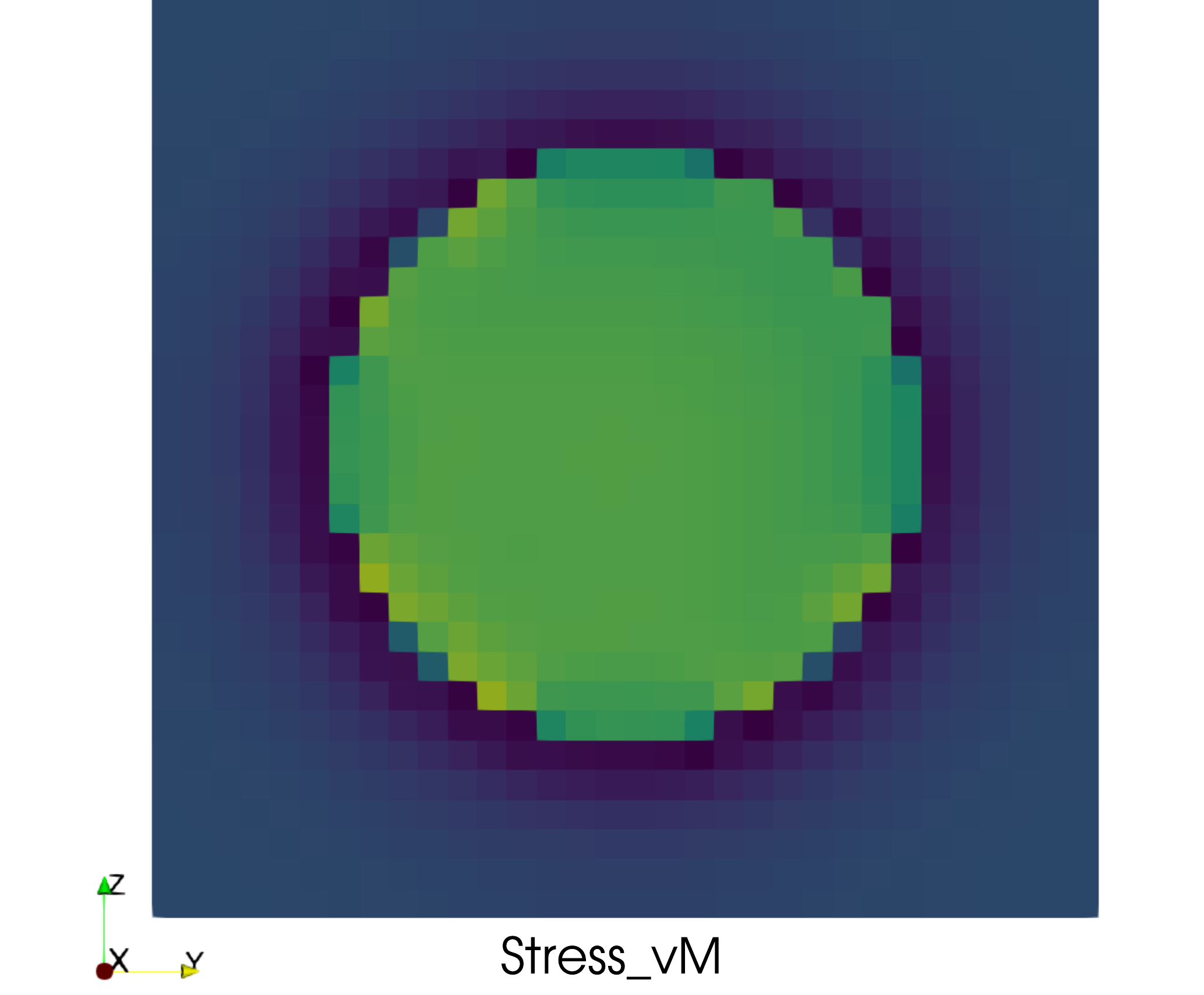}
		\caption{FNO9 - $\bar{\varepsilon}_{11}=100.0\%$}
		\label{fig:ball_FNO9_stressvM_load100}
	\end{subfigure}
	\hfill
	\begin{subfigure}{.24\textwidth}
		\centering
		\includegraphics[width=\textwidth,trim=2cm 4.35cm 39cm 7cm,clip]{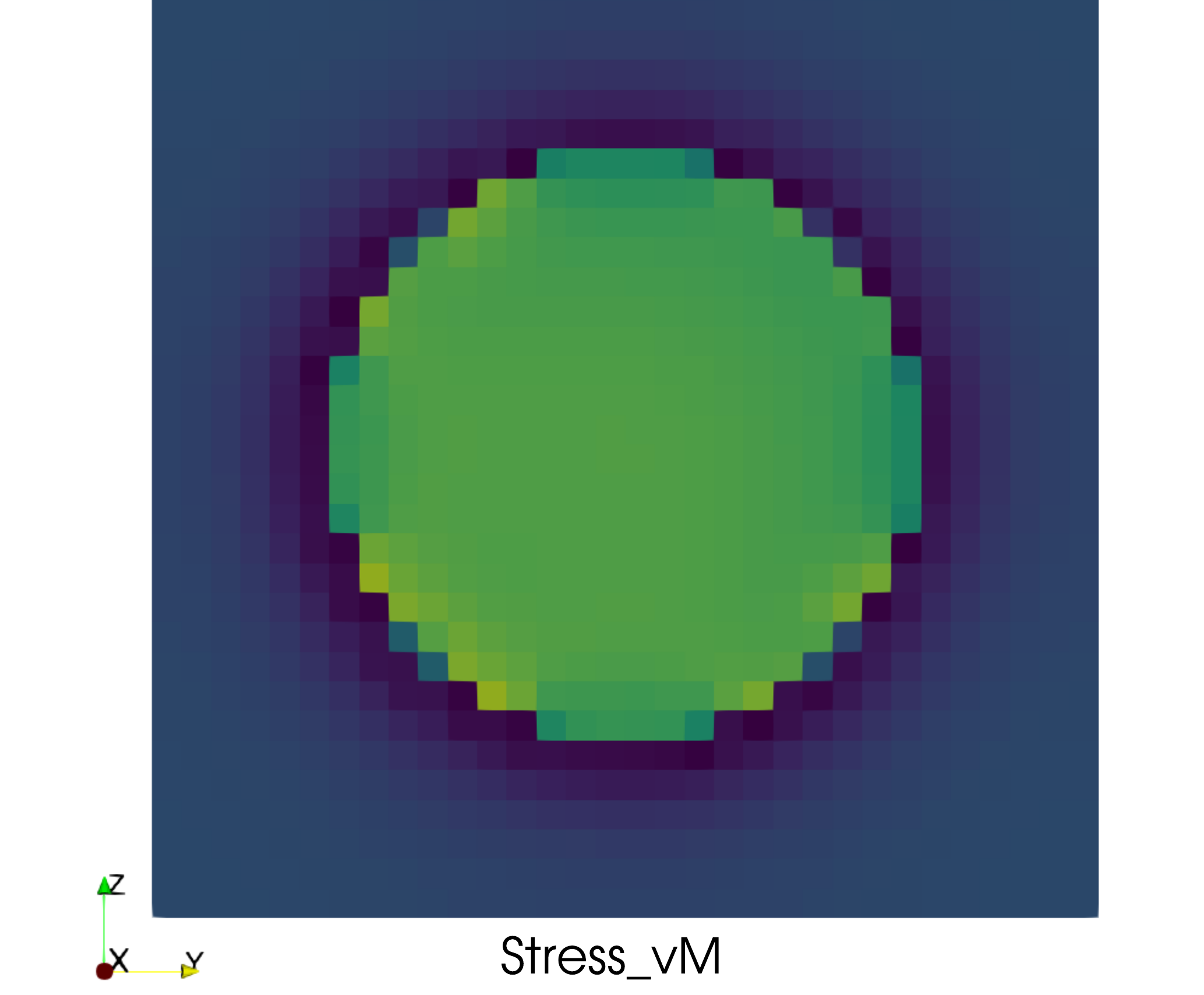}
		\caption{FNO11 - $\bar{\varepsilon}_{11}=100.0\%$}
		\label{fig:ball_FNO11_stressvM_load100}
	\end{subfigure}
	\vspace{-1em}
	\begin{center}
		\begin{subfigure}{\textwidth}
			\centering
			\includegraphics[height=.033\textheight]{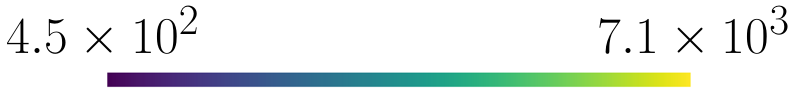}
		\end{subfigure}
	\end{center}
	\caption{Equivalent stress distributions in the microstructures described in Fig.~\ref{fig:ball} for varying prescribed strain magnitudes.}
	\label{fig:singlesphere_magnitude_sigmavM}
\end{figure}

\begin{figure}
	\centering
\begin{subfigure}{.24\textwidth}
	\centering
	\includegraphics[width=\textwidth,trim=2cm 4.4cm 39cm 7cm,clip]{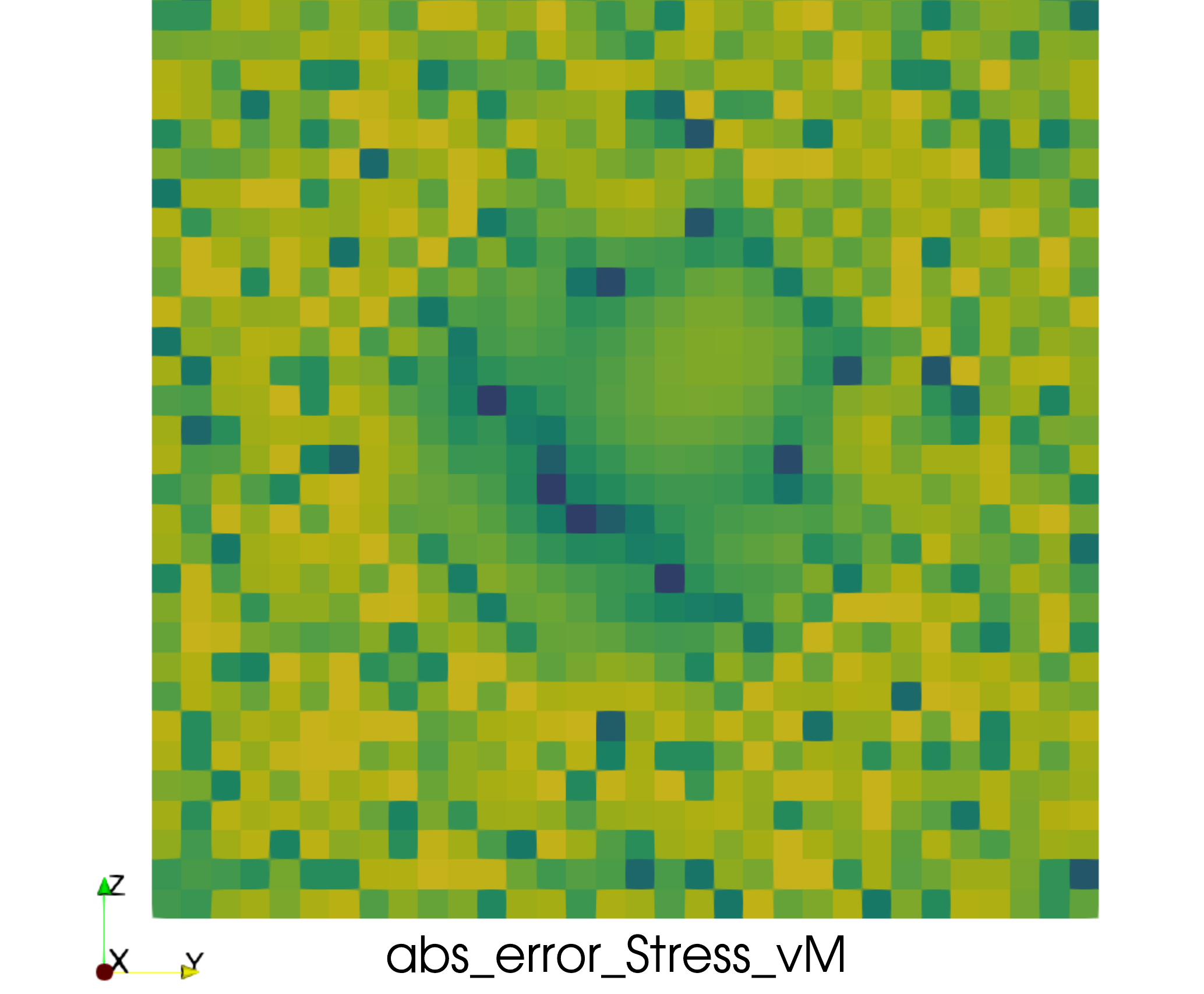}
	\caption{FNO7 - $\bar{\varepsilon}_{11}=50\%$}
	\label{fig:ball_absError_stressvM_m7_load50}
\end{subfigure}
\hfill
\begin{subfigure}{.24\textwidth}
	\centering
	\includegraphics[width=\textwidth,trim=2cm 4.4cm 39cm 7cm,clip]{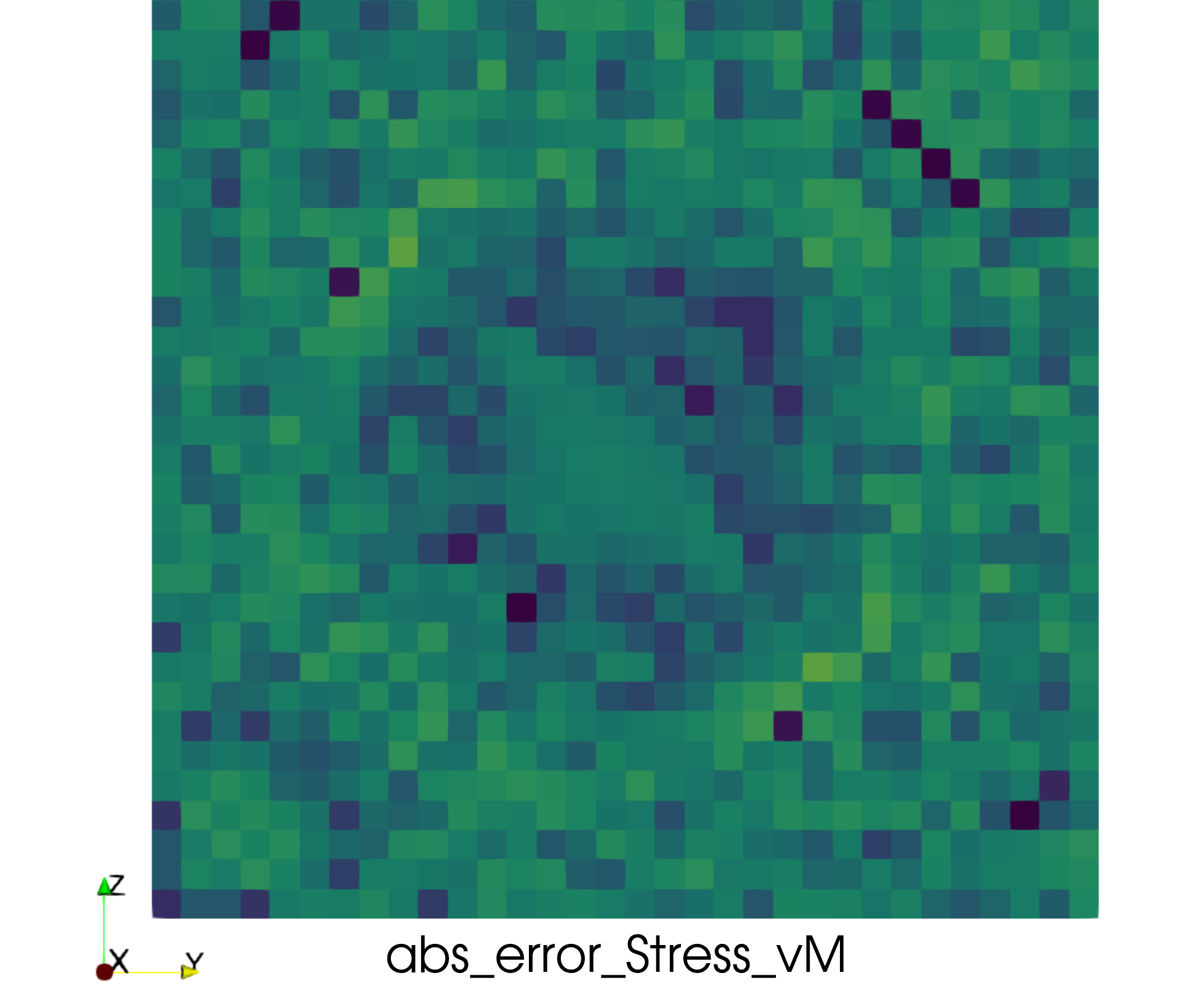}
	\caption{FNO9 - $\bar{\varepsilon}_{11}=50\%$}
	\label{fig:ball_absError_stressvM_m9_load50}
\end{subfigure}
\hfill
\begin{subfigure}{.24\textwidth}
	\centering
	\includegraphics[width=\textwidth,trim=2cm 4.4cm 39cm 7cm,clip]{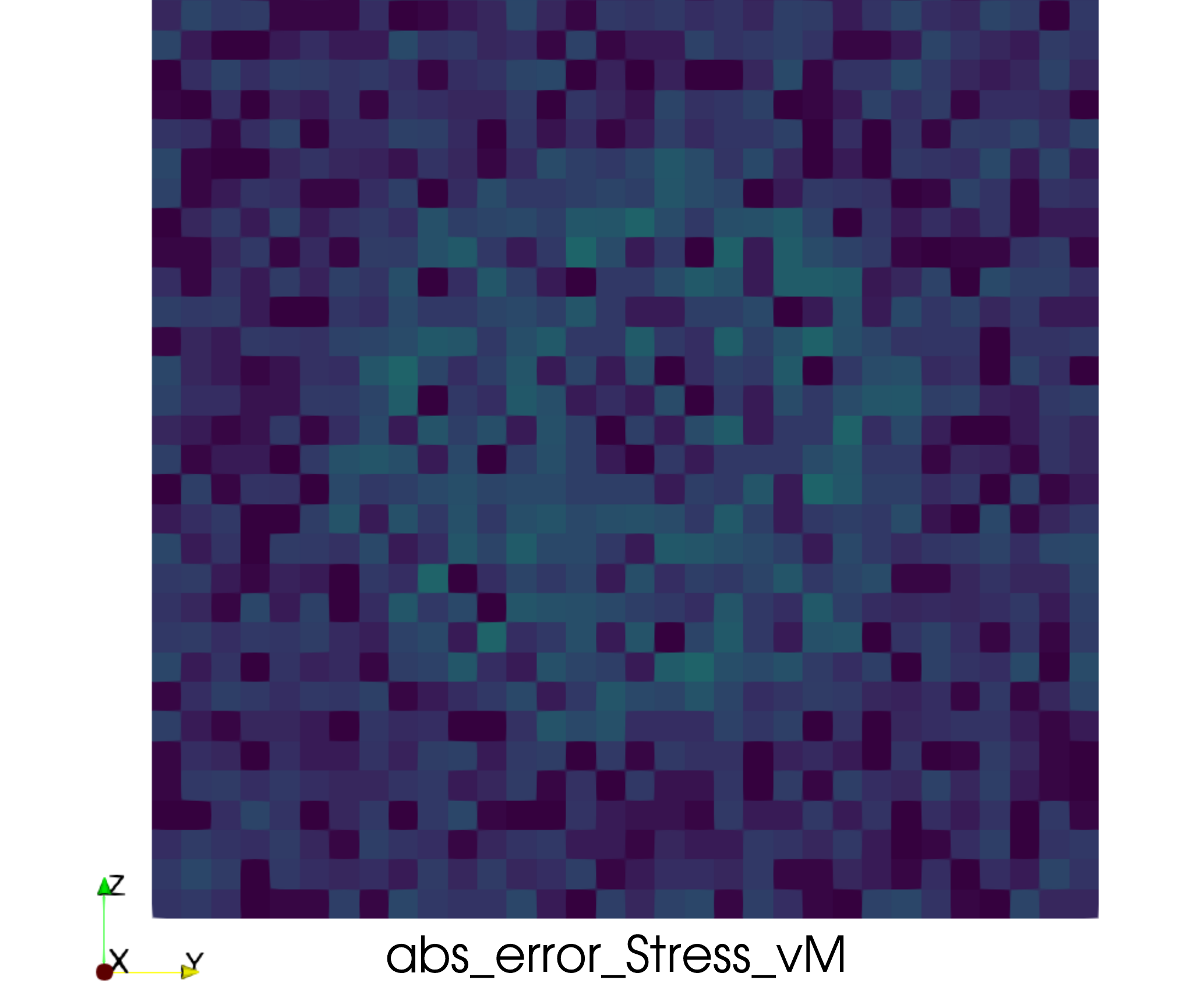}
	\caption{FNO11 - $\bar{\varepsilon}_{11}=50\%$}
	\label{fig:ball_absError_stressvM_m11_load50}
\end{subfigure}
\vspace{-1em}
\begin{center}
	\begin{subfigure}{\textwidth}
		\centering
		\includegraphics[height=.033\textheight]{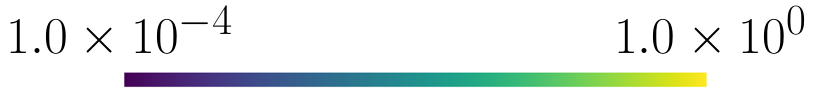}
	\end{subfigure}
\end{center}
\begin{subfigure}{.24\textwidth}
	\centering
	\includegraphics[width=\textwidth,trim=2cm 4.4cm 39cm 7cm,clip]{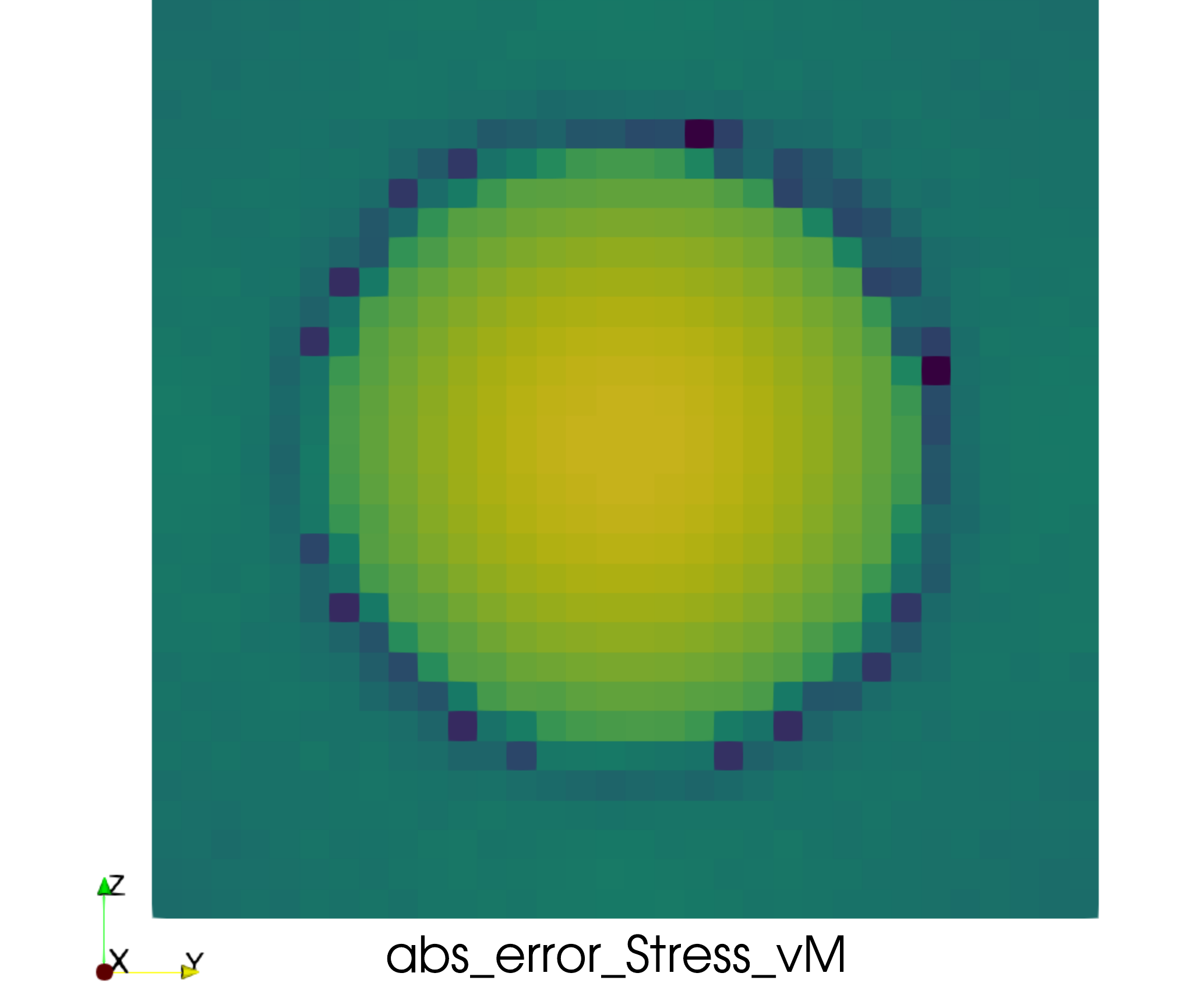}
	\caption{FNO7 - $\bar{\varepsilon}_{11}=100.0\%$}
	\label{fig:ball_absError_stressvM_m7_load100}
\end{subfigure}
\hfill
\begin{subfigure}{.24\textwidth}
	\centering
	\includegraphics[width=\textwidth,trim=2cm 4.4cm 39cm 7cm,clip]{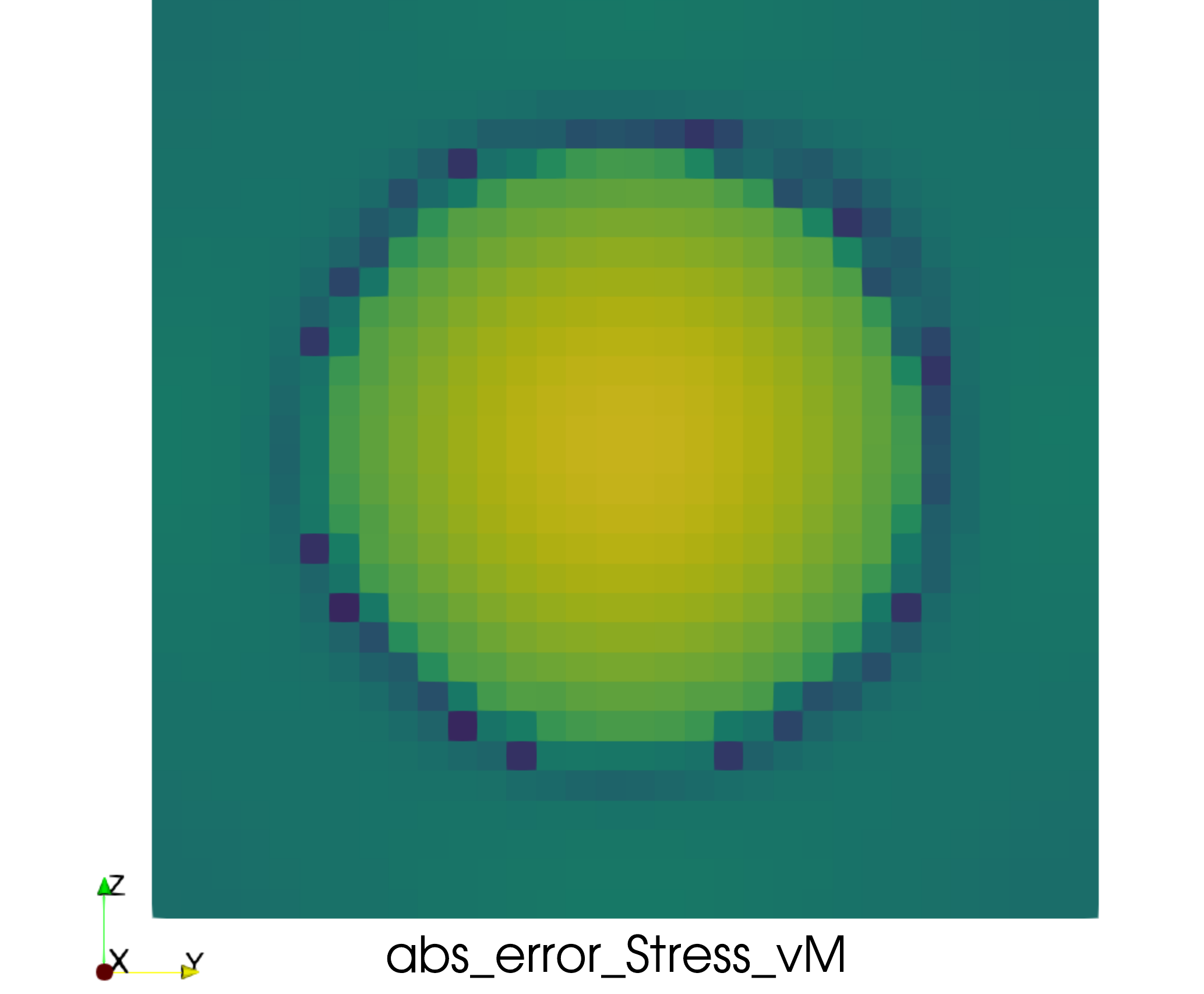}
	\caption{FNO9 - $\bar{\varepsilon}_{11}=100.0\%$}
	\label{fig:ball_absError_stressvM_m9_load100}
\end{subfigure}
\hfill
\begin{subfigure}{.24\textwidth}
	\centering
	\includegraphics[width=\textwidth,trim=2cm 4.4cm 39cm 7cm,clip]{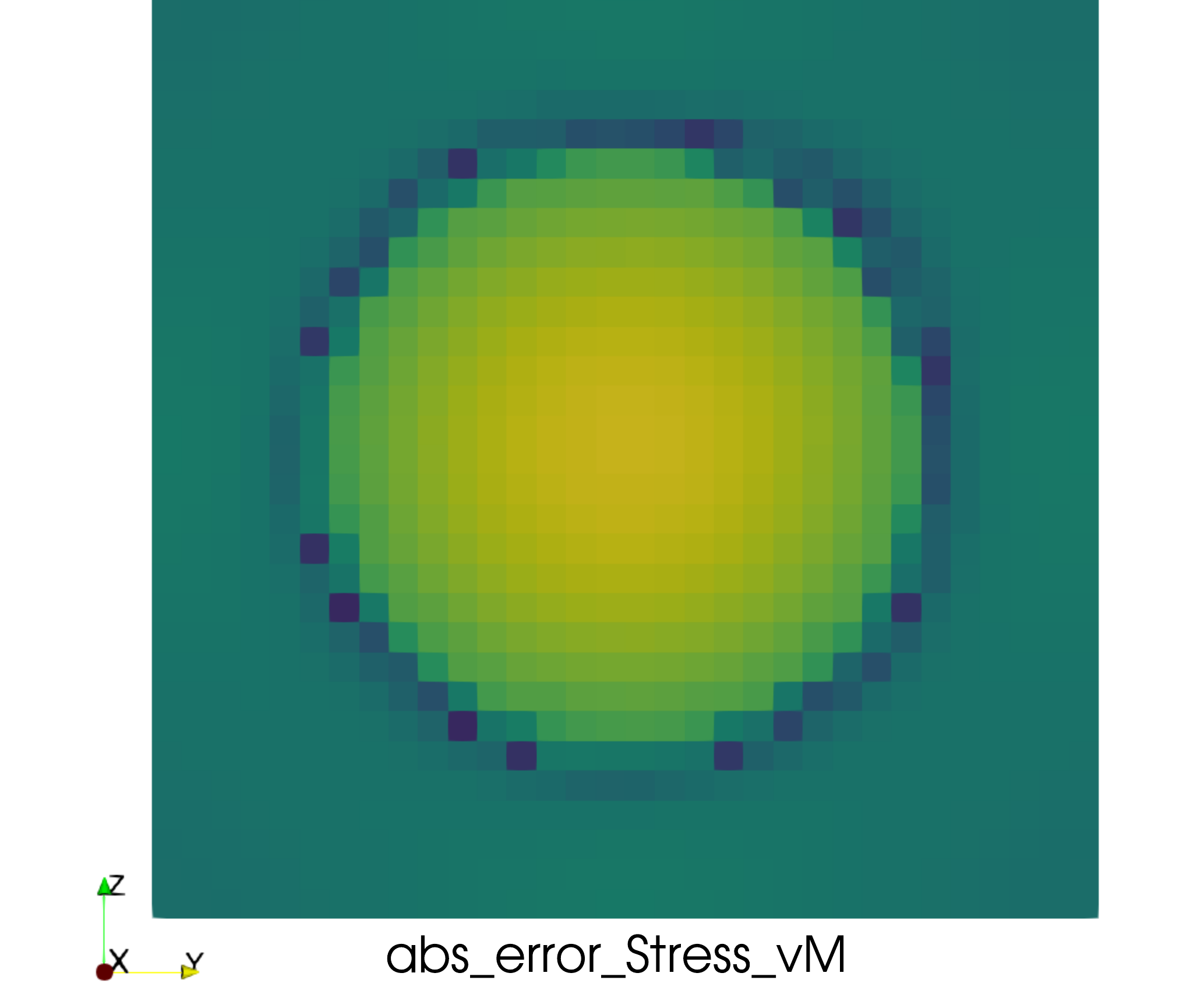}
	\caption{FNO11 - $\bar{\varepsilon}_{11}=100.0\%$}
	\label{fig:ball_absError_stressvM_m11_load100}
\end{subfigure}
\vspace{-1em}
\begin{center}
	\begin{subfigure}{\textwidth}
		\centering
		\includegraphics[height=.033\textheight]{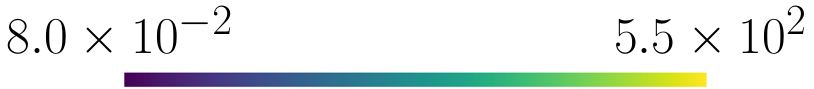}
	\end{subfigure}
\end{center}
	\caption{Absolute error of equivalent stress distributions in the microstructures described in Fig.~\ref{fig:ball} for varying  prescribed strain magnitudes.}
	\label{fig:singlesphere_magnitude_absError_sigmavM}
\end{figure}

\subsection{Universality}

In this section, we aim to illustrate the FNO model as a universal approximator for homogenization problems that does not require training and is applicable to \textit{any} microstructure which belongs to the elastic phase space $\mathcal{M}$. For this reason, we consider three distinct types of microstructures, namely ceramic spherical particles in an aluminum matrix, E-glass fibers embedded in a polyamide matrix and bound quartz-sand grains, resolved by $512^3$ voxels each. A detailed description of the microstructure-generation procedure for each of these structures is found in the work of Schneider~\cite{BB2019}, which also provides the material parameters summarized in Tab. \ref{tab:materialParameters_fancy_examples}. The computations for this example were carried out on a computing workstation with two AMD EPYC 9354 32 core processors, 1.12 TB of RAM and Ubuntu 22.04.4.
\begin{table}[H]
	\centering
	\caption{Material parameters of various microstructures described in Fig.~\ref{fig:universality_microstructures}}
	\begin{tabular}{l l l}
		\hline
		Material & E (GPa) & $\nu$ \\
		\hline \hline
		Ceramic particles & 400  & 0.2 \\		
		Aluminum matrix & 75  & 0.3 \\		
		E-glass fibers & 72  & 0.22 \\		
		Polyamide matrix & 2.1  & 0.3 \\		
		Quartz sand grains & 66.9  & 0.25 \\		
		\hline
	\end{tabular}
	\label{tab:materialParameters_fancy_examples}
\end{table}
Based on the previous benchmark examples, the FNO11 constructed from a ReLU network of depth $\mathfrak{m}=11$ shows the performance comparable to the FFT-based method. Therefore we choose the FNO11 model to study the expressivity through the predictions of effective stiffness tensor across a variety of microstructures. We employ the same computational setup as in previous example with one exception for the sand grain microstructure. Due to its porosity, for the purpose of comparison, we choose a larger termination error at $10^{-3}$ in the iterative solver for the sand grain microstructure in both FNO11 model and the FFT-based method. In Tab. \ref{tab:universality_effective_elastic}, we report the computed effective Young's modulus and Poisson's ratio of three microstructures. The relative errors of the FNO11 model's predictions are given in parentheses, showing excellent agreements as compared to the FFT-based method with at most $0.05\%$ error. Additionally, for the prescribed macroscopic strain in $x-$direction, the Mises-equivalent stresses $\sigma^{\texttt{ev}}$ of three different microstructures obtained from FNO11 model are compared to the corresponding results from FFT-based method in Fig.~\ref{fig:universality_microstructures}, and an excellent agreement is observed. For a closer look, the 2D plots of equivalent stress distributions and their absolute errors are shown in Fig.~\ref{fig:universality_microstructures_plane}.

\begin{table}
	\caption{Computed isotropic engineering constants for various microstructures. The relative errors of the FNO11 model are given in parentheses.}
	\label{tab:universality_effective_elastic}
	\centering
	\begin{tabular}{crc|rc}
		\cline{2-5}
		& \multicolumn{2}{c|}{FFT}  & \multicolumn{2}{c}{FNO11} \\ \cline{2-5}
		\hline
		Microstructures  & $E^{\textrm{eff}}$ (GPa) & $\nu^{\textrm{eff}}$ & $E^{\textrm{eff}}$ (GPa)  & $\nu^{\textrm{eff}}$ \\
		\hline \hline 
		Spherical inclusions & 115.5675 & 0.2436 & 115.5309 (0.03$\%$) & 0.2435	 (0.04$\%$)	\\
		\review{Fiber} reinforcements & 71.2393 & 0.1837 & 71.2262  (0.02$\%$) & 0.1837  (0.00$\%$)		\\
		Sand grains & 15.2933 & 0.2186 & 15.2905  (0.02$\%$) & 0.2185  (0.05$\%$)	\\
		\hline  	
	\end{tabular}
\end{table}

\begin{figure}
	\centering
	\begin{subfigure}{.3\textwidth}
		\centering
		\includegraphics[width=\textwidth,trim=5cm 0cm 38cm 11cm,clip]{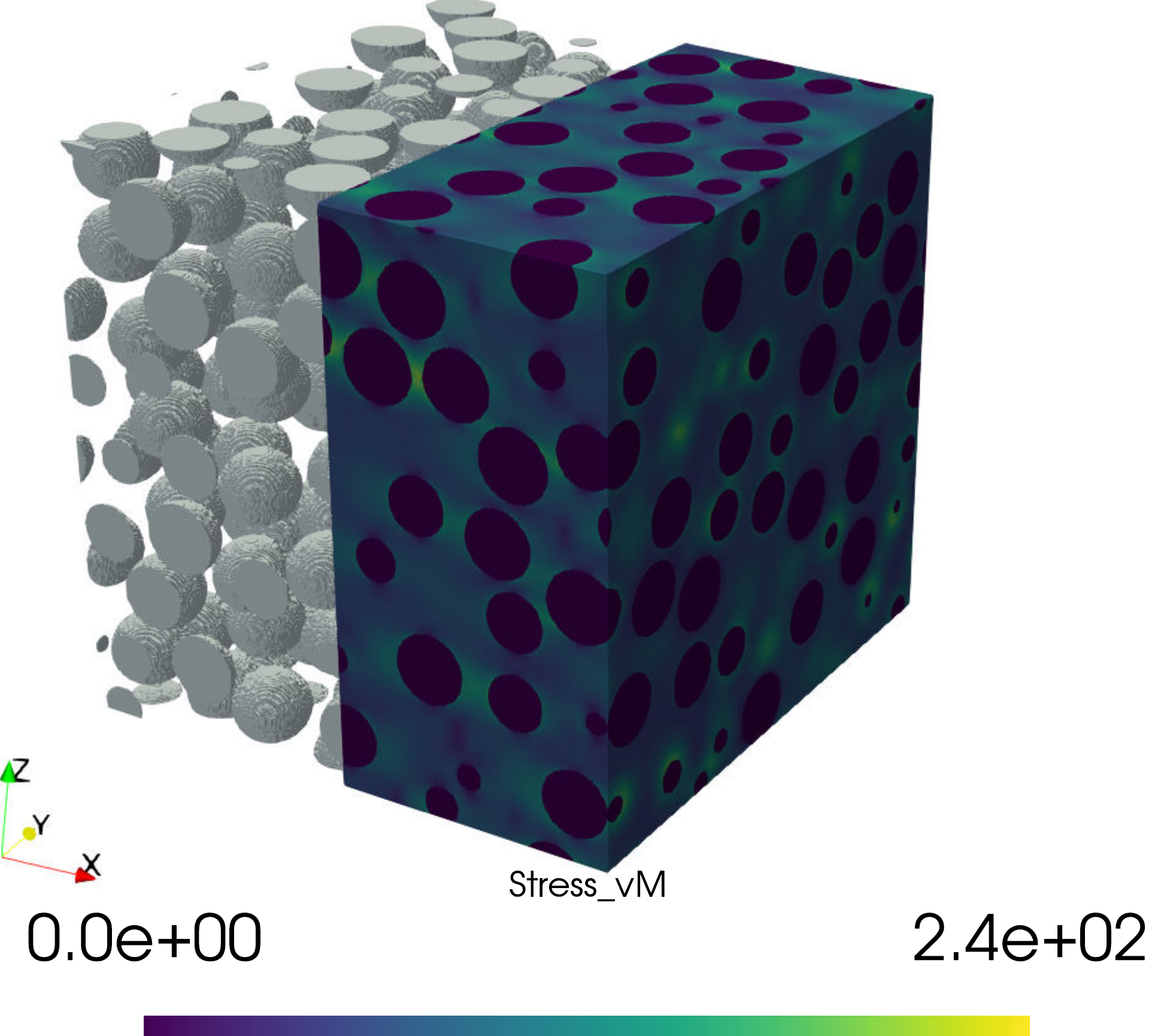}
		\caption{FNO}
		\label{fig:spheres_stressvM_FNO}
	\end{subfigure}
	\hfill
	\begin{subfigure}{.3\textwidth}
		\centering
		\includegraphics[width=\textwidth,trim=5cm 0cm 38cm 11cm,clip]{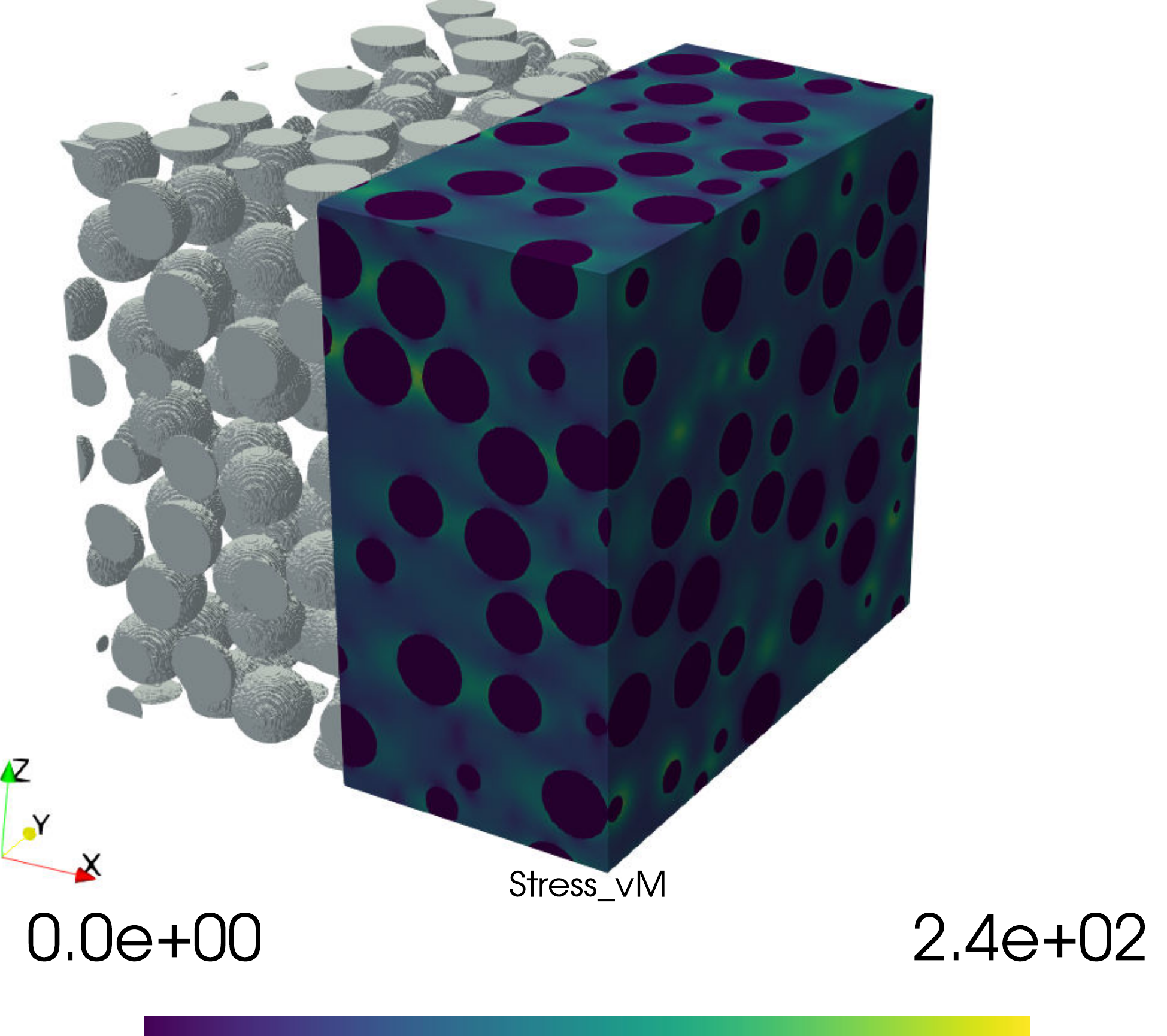}
		\caption{FFT}
		\label{fig:spheres_stressvM_FFT}
	\end{subfigure}	
	\hfill
	\begin{subfigure}{.3\textwidth}
		\centering
		\includegraphics[width=\textwidth,trim=4cm 0cm 38cm 11cm,clip]{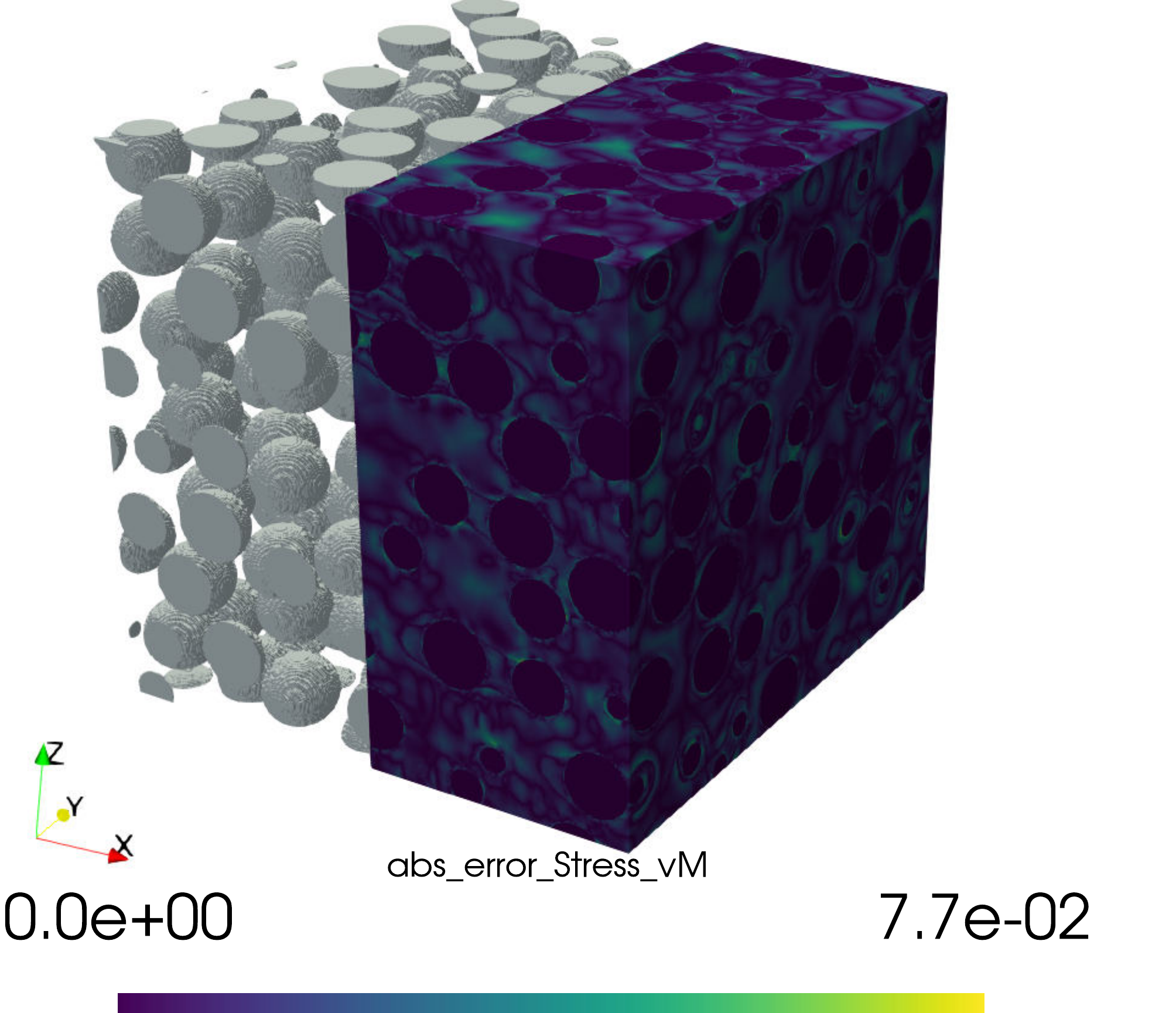}
		\caption{Error}
		\label{fig:spheres_stressvM_error}
	\end{subfigure}
	
	\begin{subfigure}{.3\textwidth}
		\centering
		\includegraphics[width=\textwidth,trim=5cm 0cm 38cm 10cm,clip]{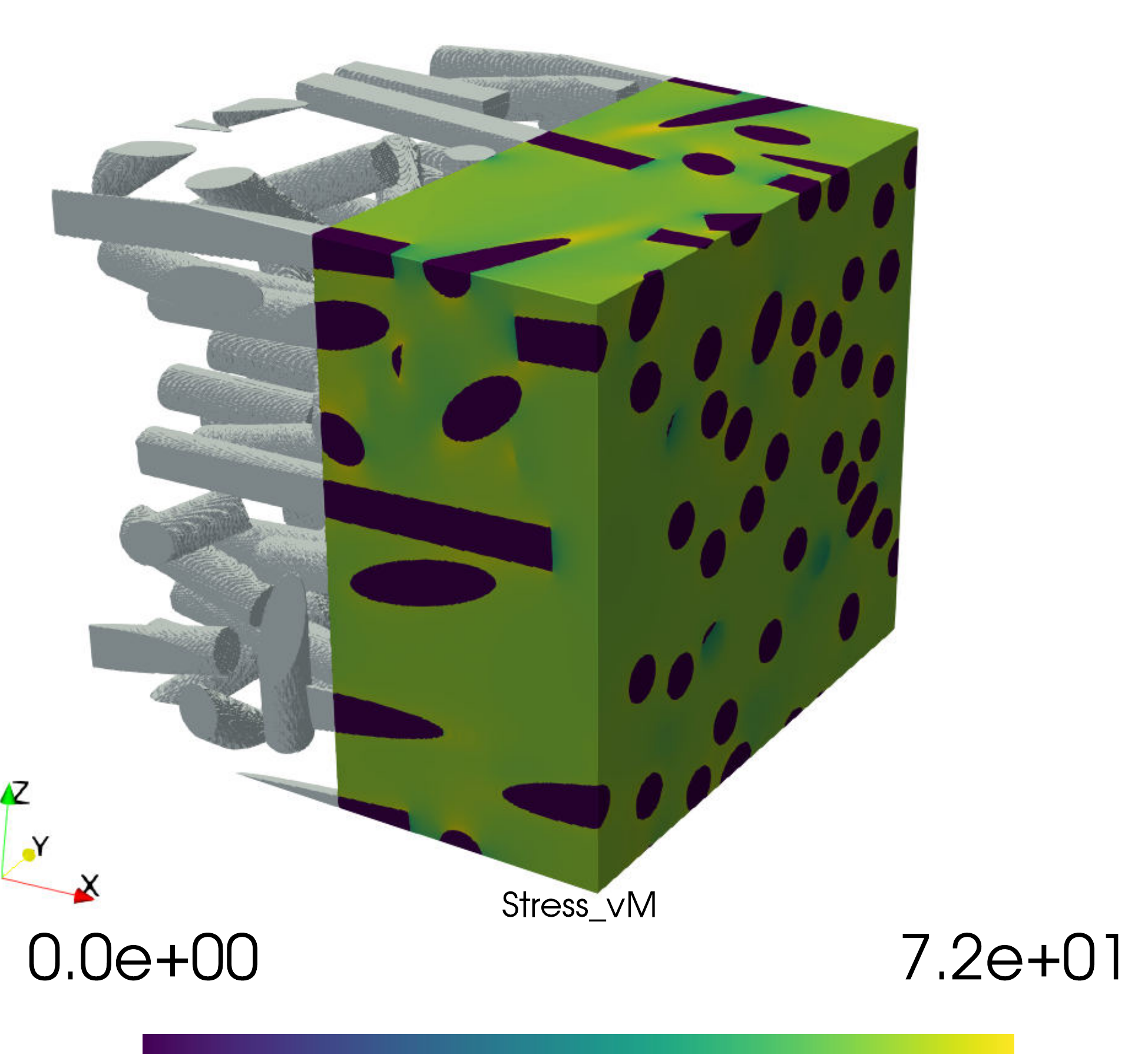}
		\caption{FNO}
		\label{fig:fibers_stressvM_FNO}
	\end{subfigure}
	\hfill
	\begin{subfigure}{.3\textwidth}
		\centering
		\includegraphics[width=\textwidth,trim=5cm 0cm 38cm 10cm,clip]{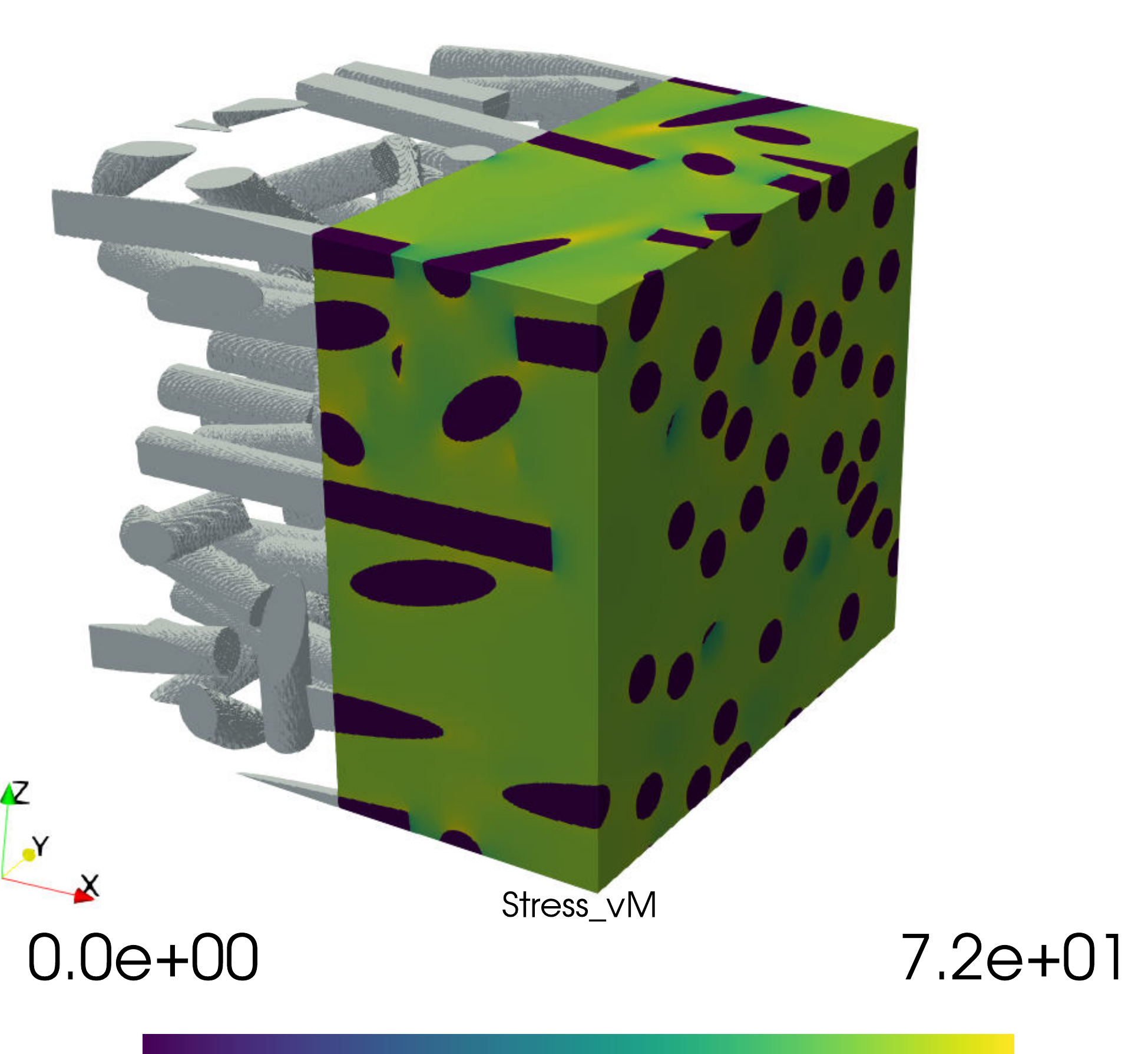}
		\caption{FFT}
		\label{fig:fibers_stressvM_FFT}
	\end{subfigure}	
	\hfill
	\begin{subfigure}{.3\textwidth}
		\centering
		\includegraphics[width=\textwidth,trim=4cm 0cm 38cm 10cm,clip]{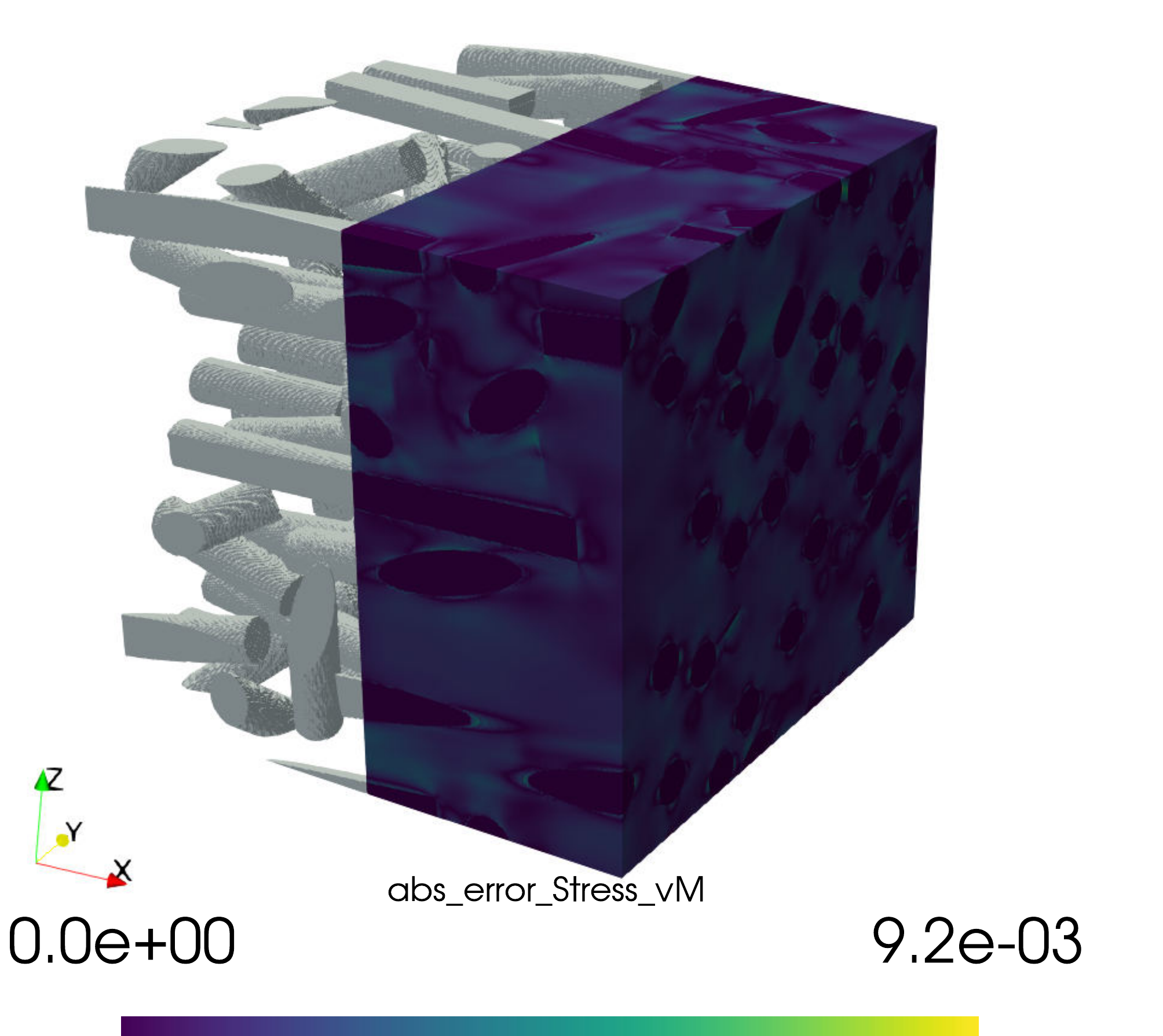}
		\caption{Error}
		\label{fig:fibers_stressvM_error}
	\end{subfigure}
	
	\begin{subfigure}{.3\textwidth}
		\centering
		\includegraphics[width=\textwidth,trim=5cm 0cm 38cm 10cm,clip]{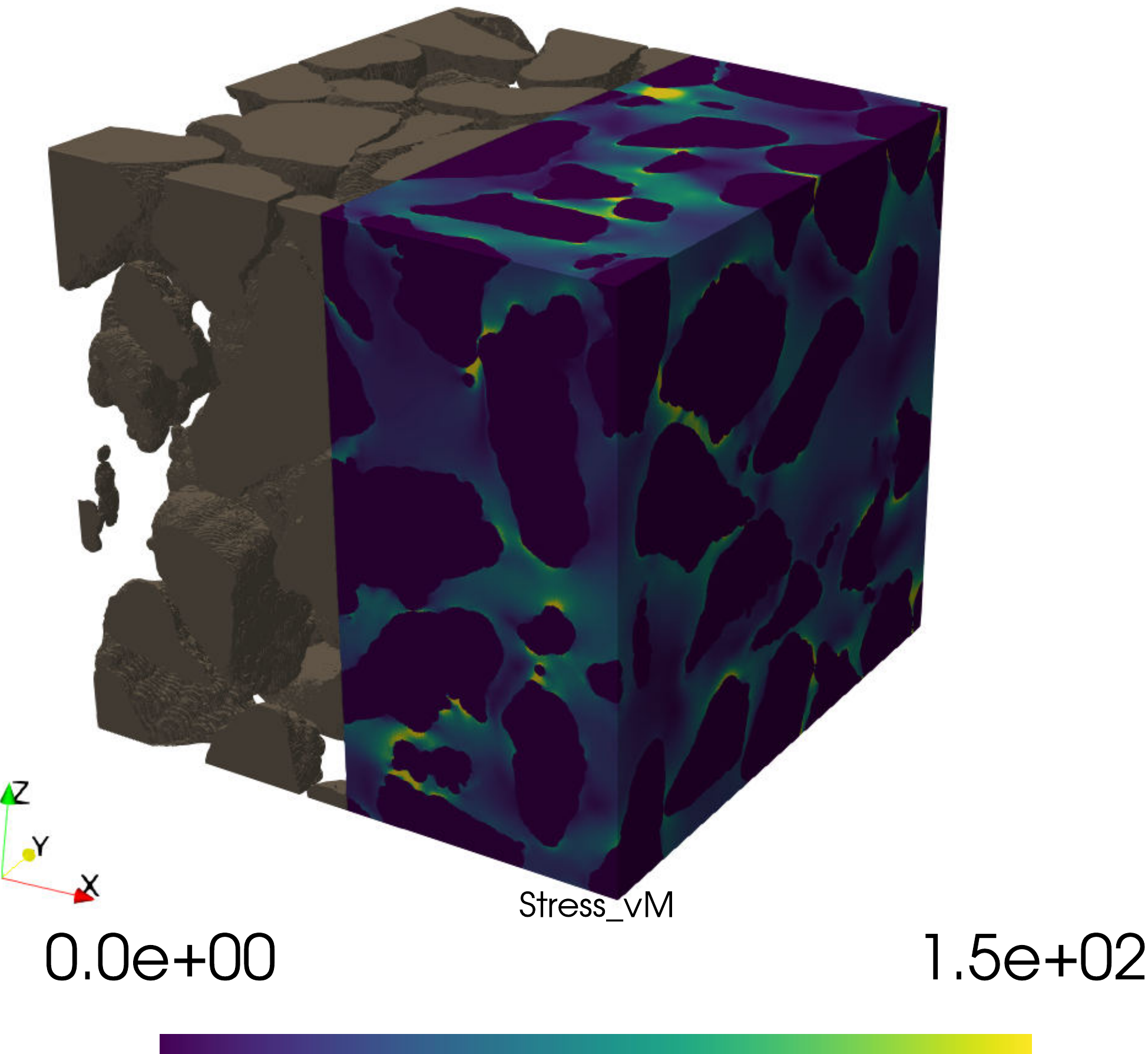}
		\caption{FNO}
		\label{fig:grains_stressvM_FNO}
	\end{subfigure}
	\hfill
	\begin{subfigure}{.3\textwidth}
		\centering
		\includegraphics[width=\textwidth,trim=5cm 0cm 38cm 10cm,clip]{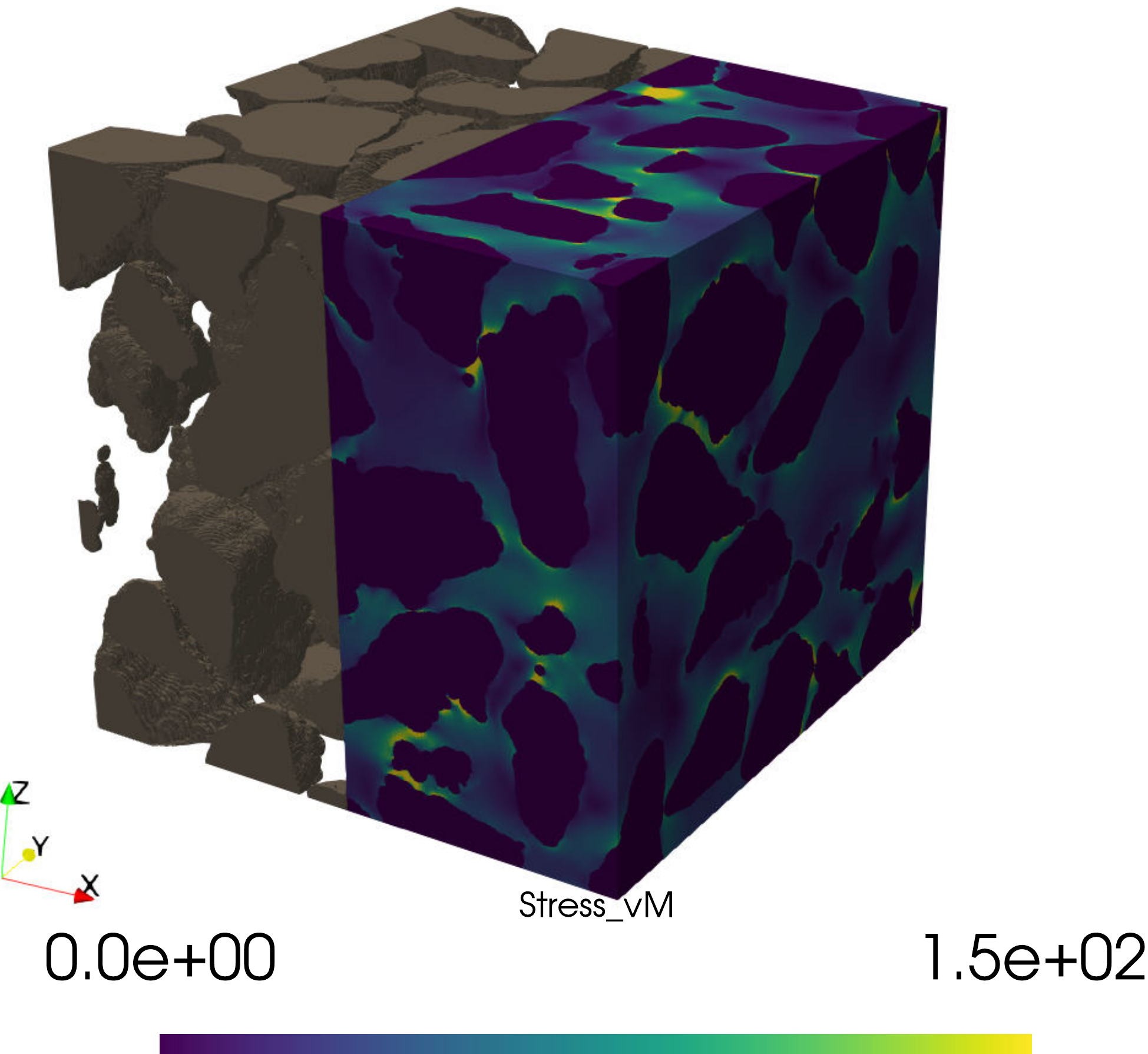}
		\caption{FFT}
		\label{fig:grains_stressvM_FFT}
	\end{subfigure}	
	\hfill
	\begin{subfigure}{.3\textwidth}
		\centering
		\includegraphics[width=\textwidth,trim=4cm 0cm 38cm 10cm,clip]{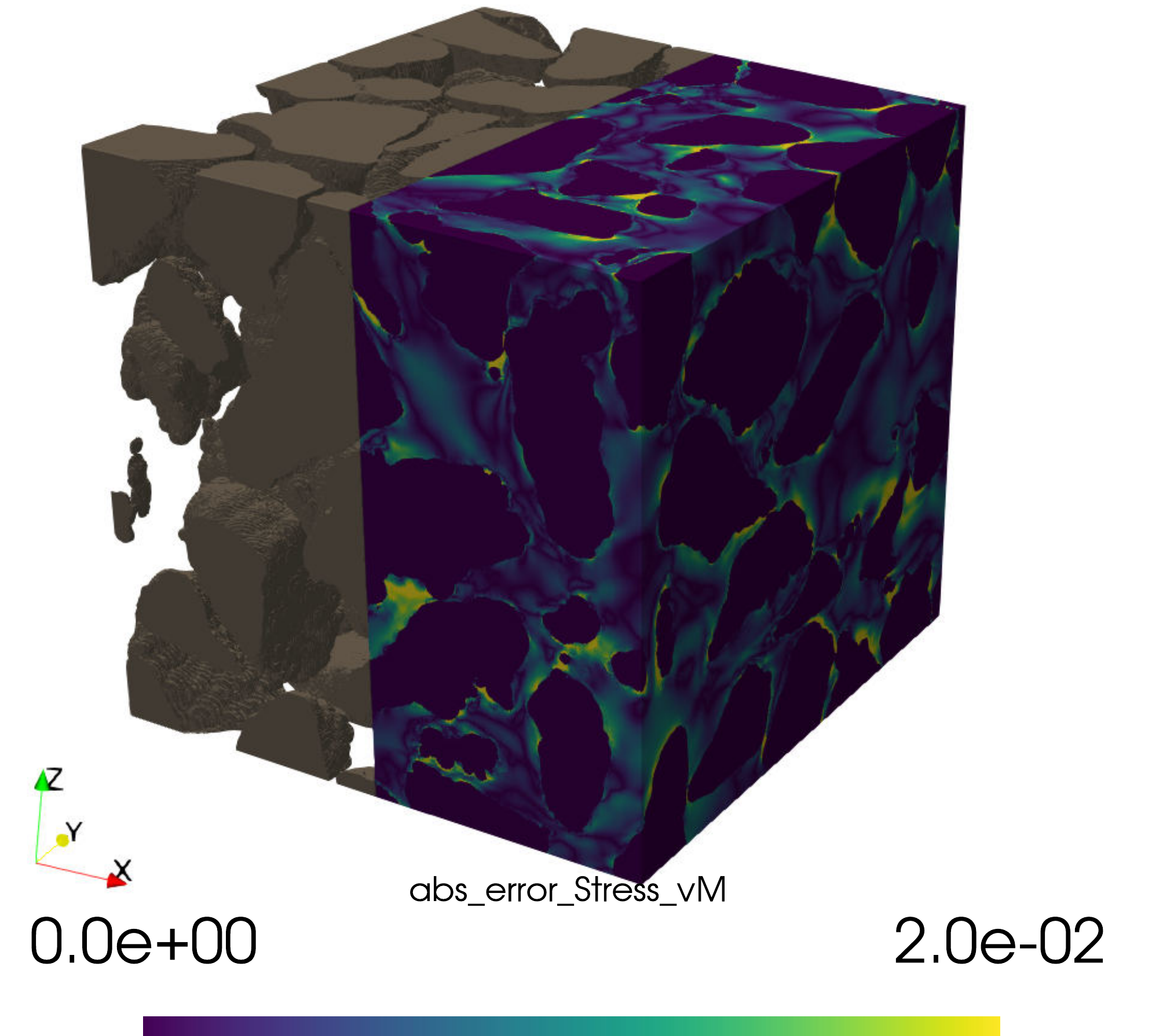}
		\caption{Error}
		\label{fig:grains_stressvM_error}
	\end{subfigure}	
	\caption{Mises-equivalent stresses $\sigma^{\texttt{ev}}$.}
	\label{fig:universality_microstructures}
\end{figure}	

\begin{figure}
	\centering
	\begin{subfigure}{.3\textwidth}
		\centering
		\includegraphics[width=\textwidth,trim=3cm 0cm 37cm 0cm,clip]{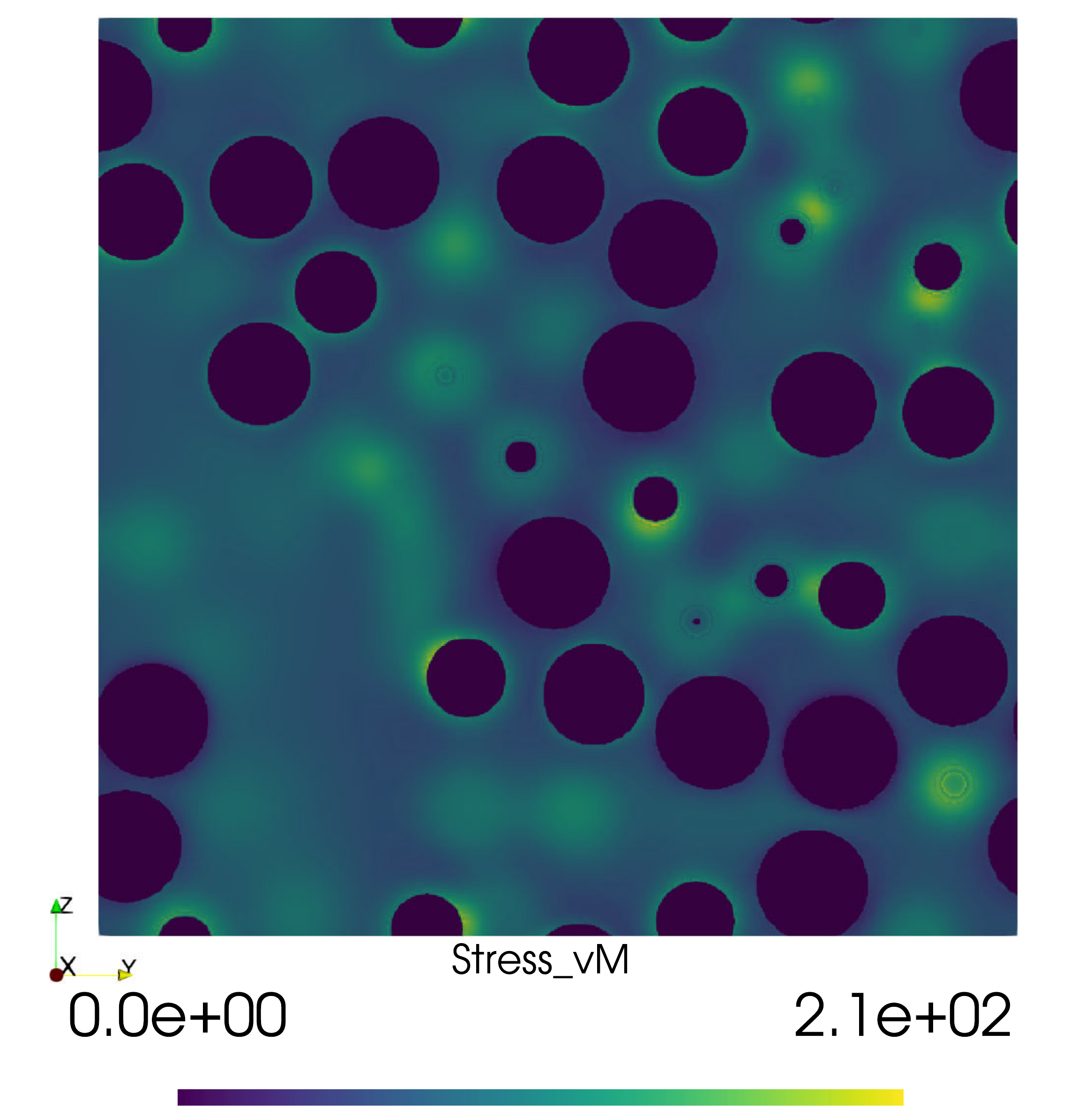}
		\caption{FNO}
		\label{fig:spheres_FNO11_stressvM_512}
	\end{subfigure}	
	\hfill
	\begin{subfigure}{.3\textwidth}
		\centering
		\includegraphics[width=\textwidth,trim=3cm 0cm 37cm 0cm,clip]{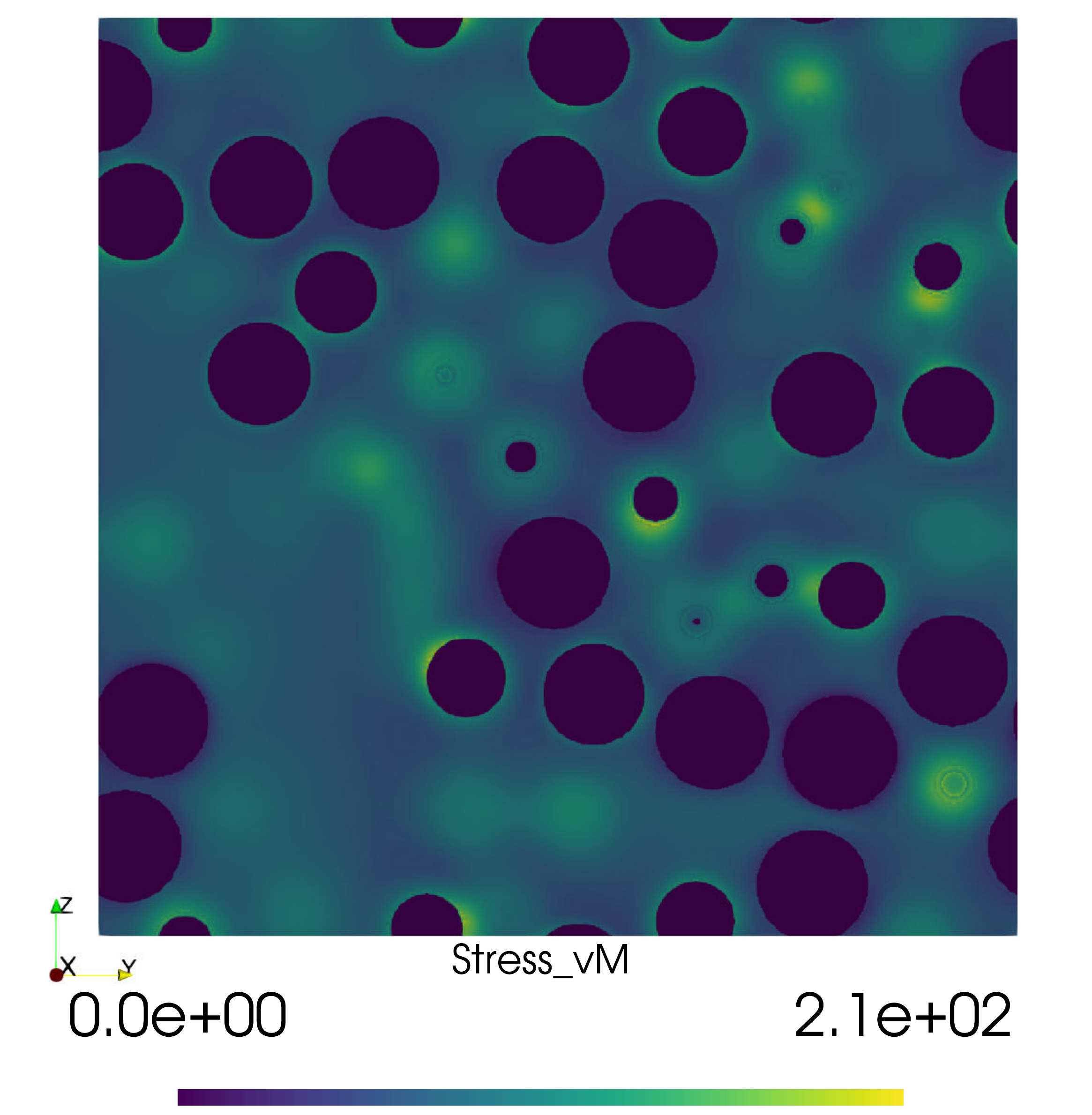}
		\caption{FFT}
		\label{fig:spheres_FFT_stressvM_512}
	\end{subfigure}
	\hfill
	\begin{subfigure}{.3\textwidth}
		\centering
		\includegraphics[width=\textwidth,trim=3cm 0cm 37cm  0cm,clip]{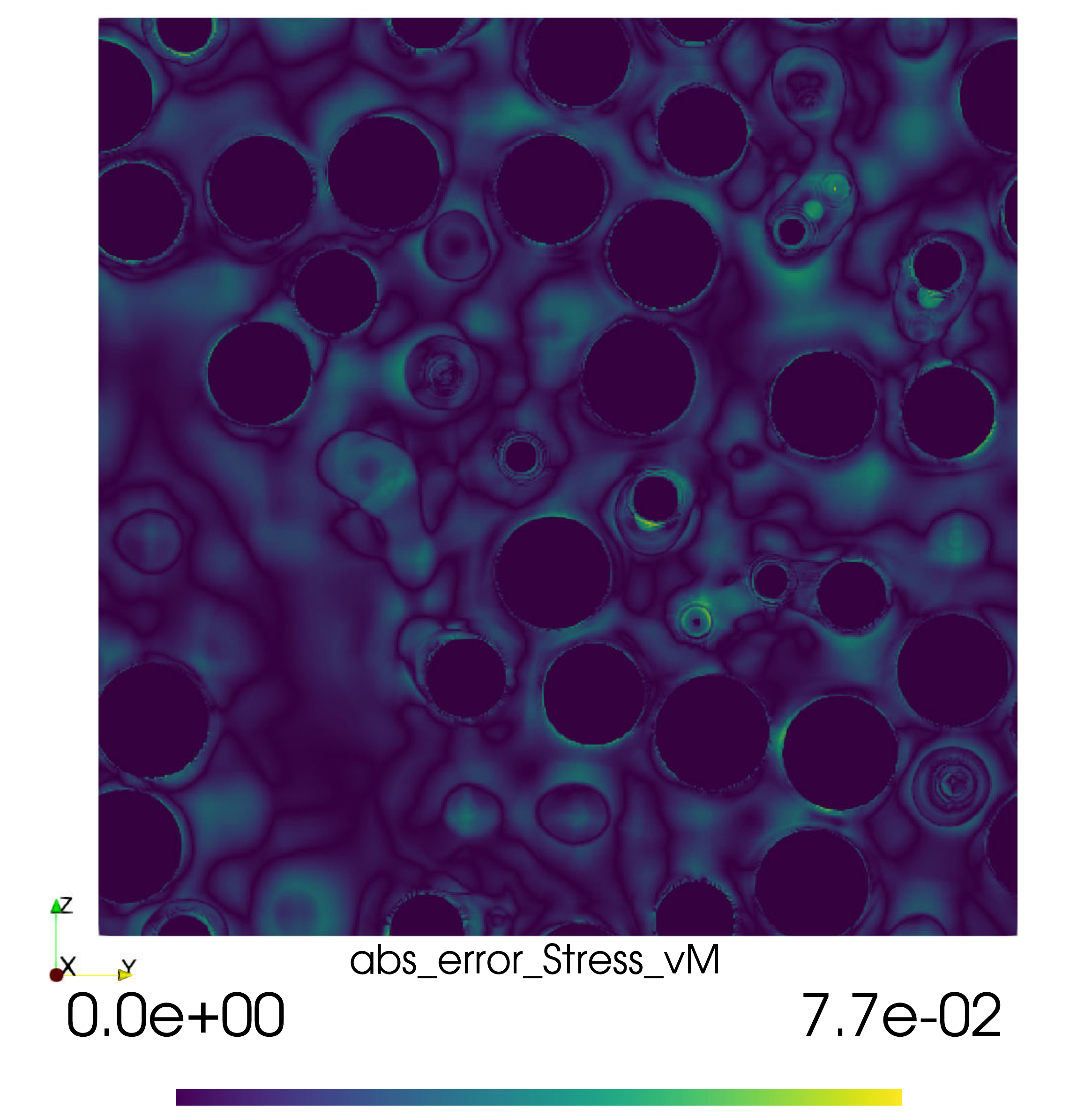}
		\caption{Error}
		\label{fig:spheres_absError_stressvM_512}
	\end{subfigure}
	
	\begin{subfigure}{.3\textwidth}
		\centering
		\includegraphics[width=\textwidth,trim=3cm 0cm 37cm 0cm,clip]{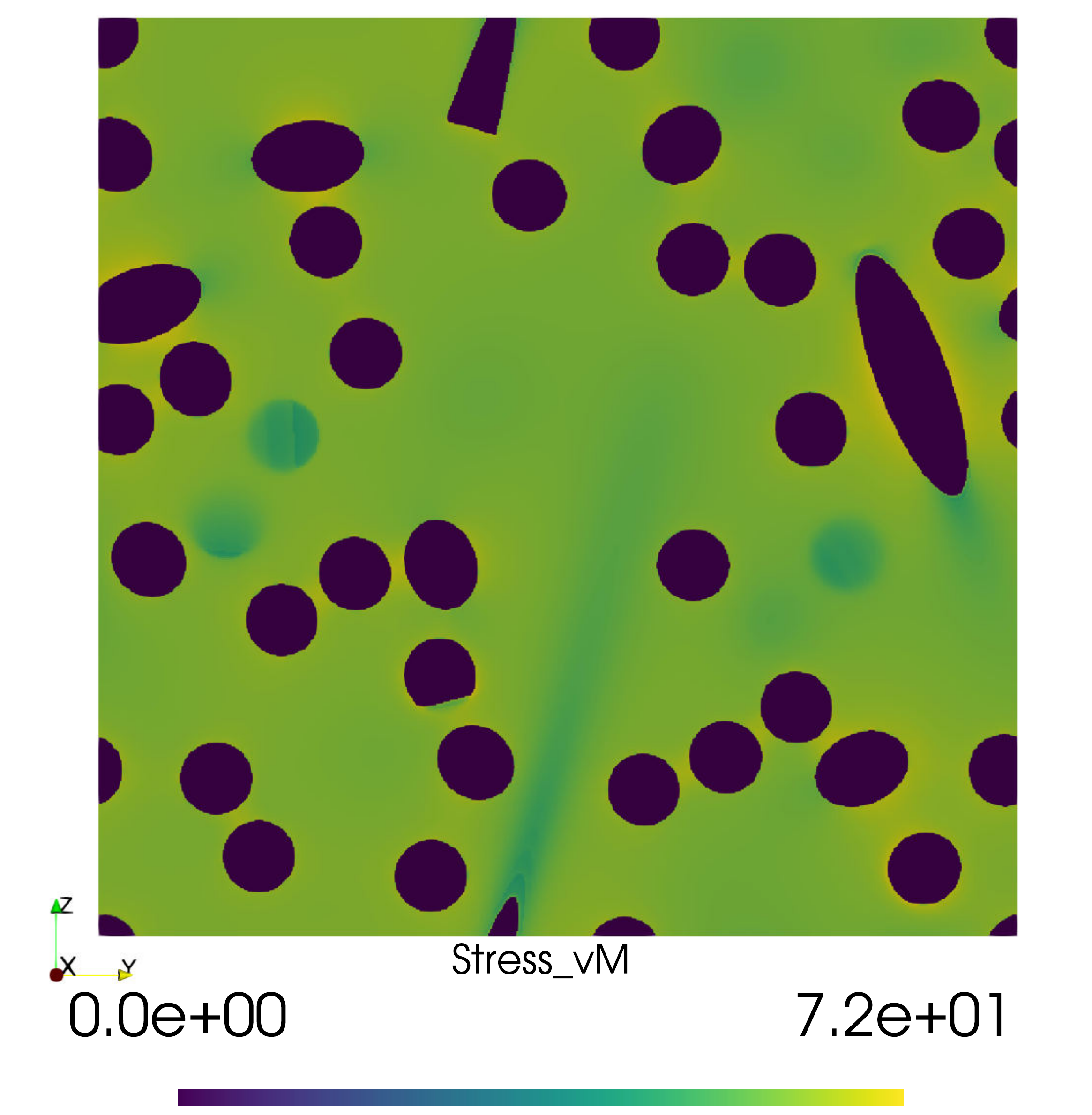}
		\caption{FNO}
		\label{fig:fibers_FNO11_stressvM_512}
	\end{subfigure}	
	\hfill
	\begin{subfigure}{.3\textwidth}
		\centering
		\includegraphics[width=\textwidth,trim=3cm 0cm 37cm 0cm,clip]{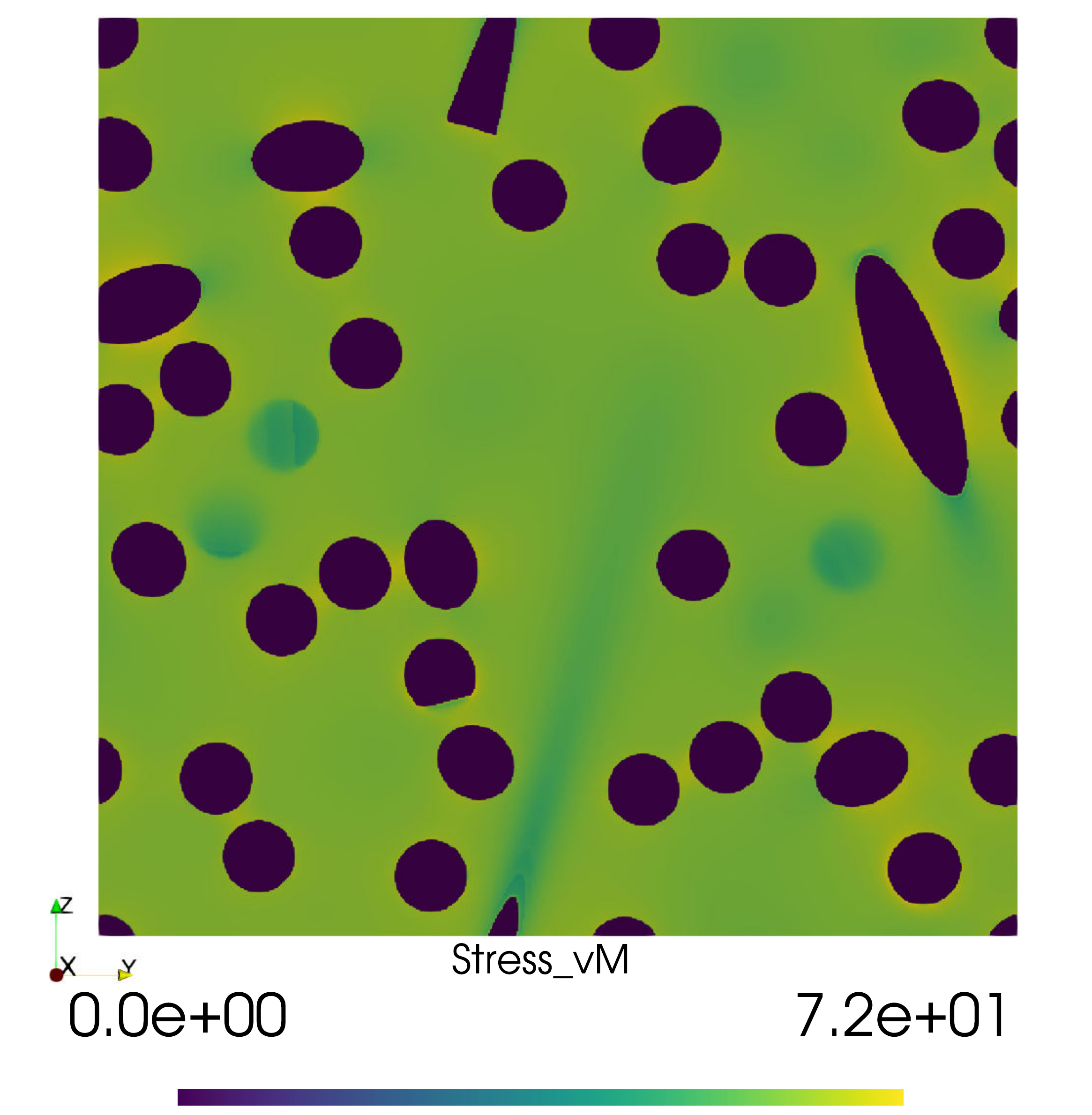}
		\caption{FFT}
		\label{fig:fibers_FFT_stressvM_512}
	\end{subfigure}
	\hfill
	\begin{subfigure}{.3\textwidth}
		\centering
		\includegraphics[width=\textwidth,trim=3cm 0cm 37cm  0cm,clip]{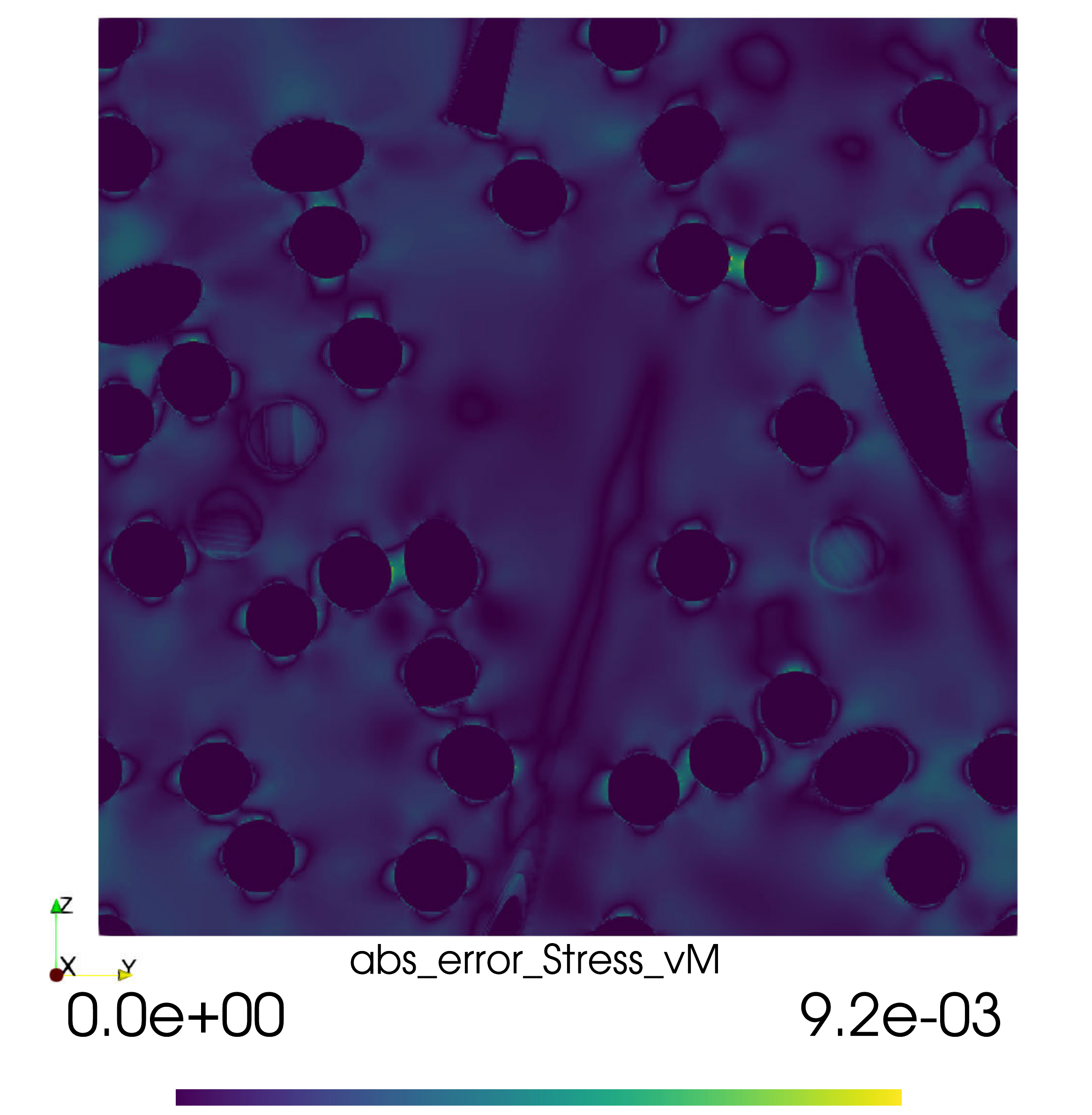}
		\caption{Error}
		\label{fig:fibers_absError_stressvM_512}
	\end{subfigure}
	
	\begin{subfigure}{.3\textwidth}
		\centering
		\includegraphics[width=\textwidth,trim=3cm 0cm 37cm 0cm,clip]{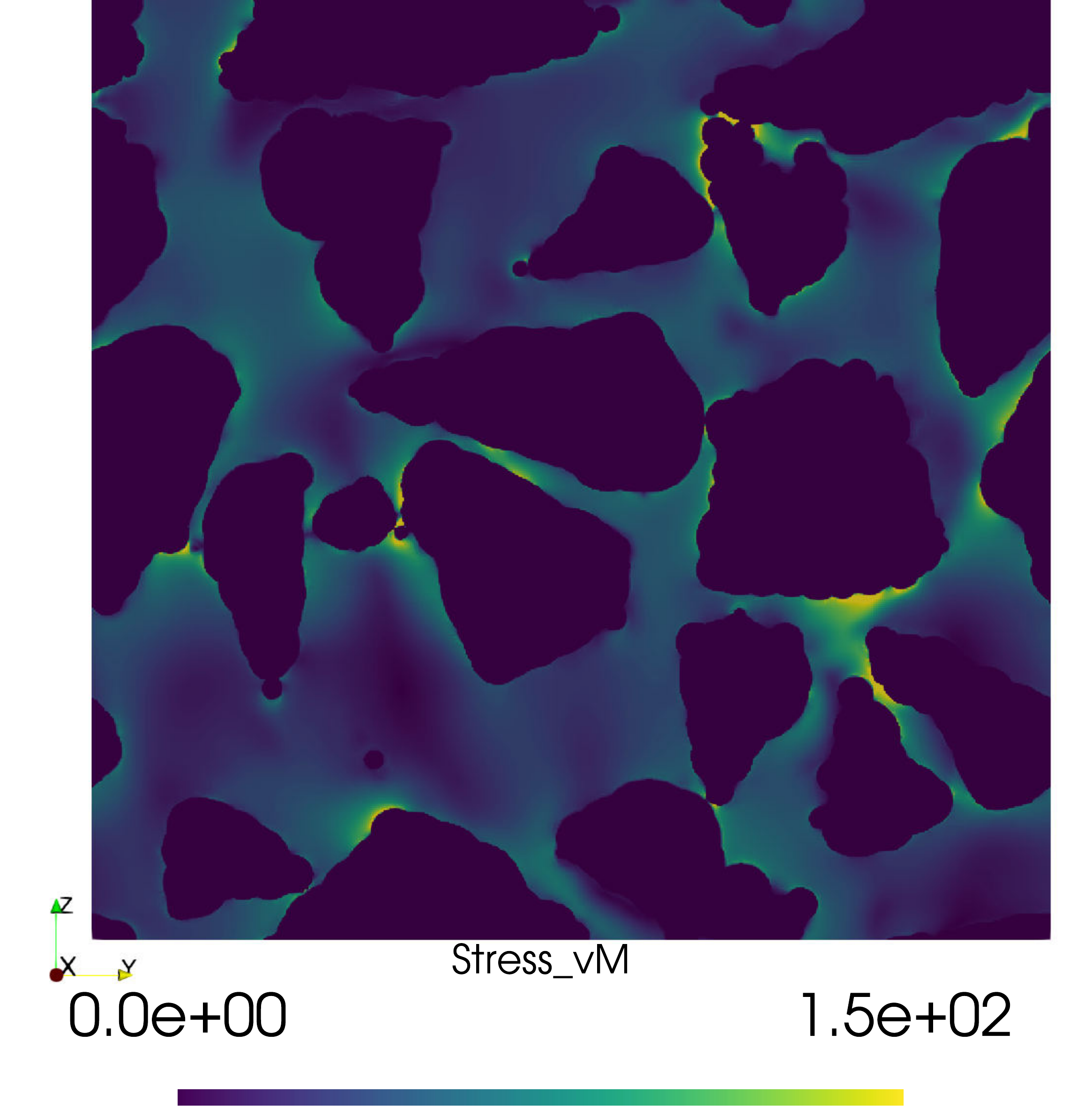}
		\caption{FNO}
		\label{fig:grains_FNO11_stressvM_512}
	\end{subfigure}	
	\hfill
	\begin{subfigure}{.3\textwidth}
		\centering
		\includegraphics[width=\textwidth,trim=3cm 0cm 37cm 0cm,clip]{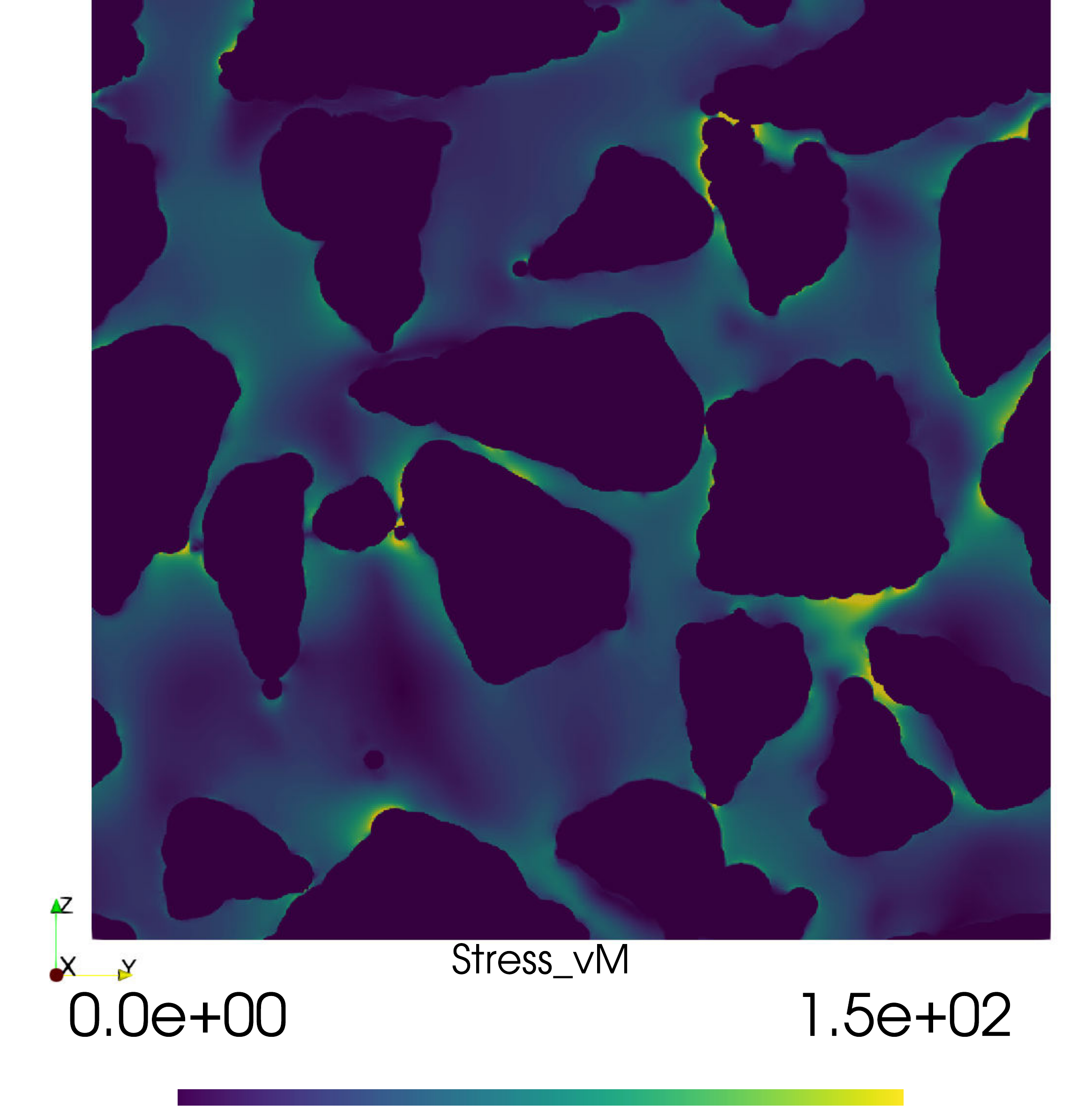}
		\caption{FFT}
		\label{fig:grains_FFT_stressvM_512}
	\end{subfigure}
	\hfill
	\begin{subfigure}{.3\textwidth}
		\centering
		\includegraphics[width=\textwidth,trim=3cm 0cm 37cm  0cm,clip]{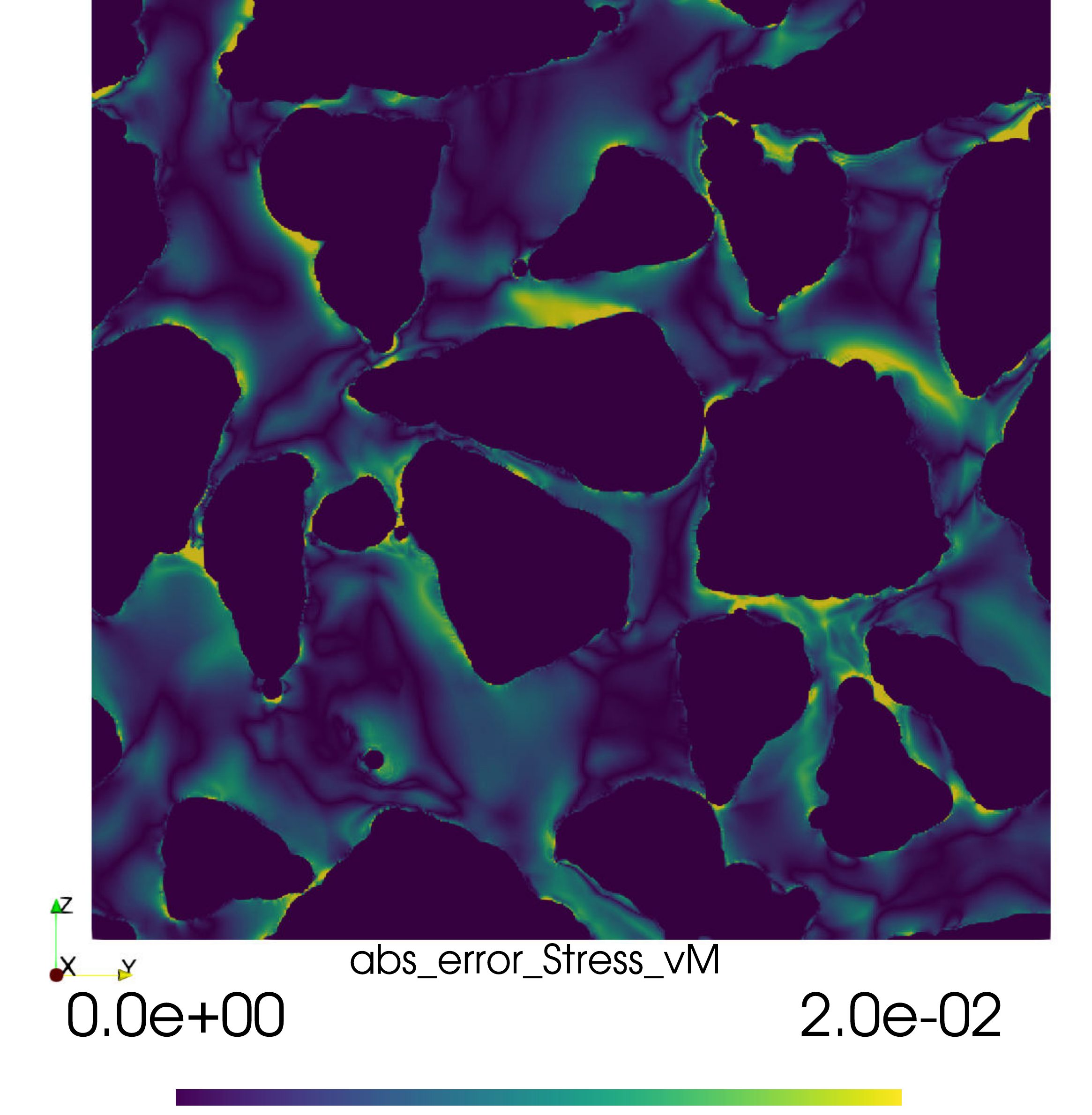}
		\caption{Error}
		\label{fig:grains_absError_stressvM_512}
	\end{subfigure}
	\caption{Equivalent stress $\sigma^{\texttt{ev}}$ at the mid-plane of the microstructures.}
	\label{fig:universality_microstructures_plane}
\end{figure}	

\section{Conclusions}
\label{sec:conclusions}

This work was concerned with the use of Fourier neural operators (FNOs) for computational micromechanics. Although deep learning and operator-learning frameworks are available for homogenization tasks, these do not appear to outperform the most efficient discretization-solver combinations available via traditional (non-neural) methods. We attribute this shortcoming to the desire of creating a "{}new"{} tool, repeating mistakes that may be avoided and falling into traps that are unnecessary.\\
Instead, we recommend to generalize the available framework~\cite{Segurado2020CPReview,LebensohnRolletCPReview,FFTReview2020} of solvers based on the fast Fourier transform (FFT) to include FNOs. The work at hand emulated the basic scheme~\cite{MoulinecSuquet1994,MoulinecSuquet1998}, the simplest FFT solver, by an FNO. By this simple trick, we were able to remove most of the limitations of the approximation result~\cite{Bhattacharya2024LearningHomogenization} by Bhattacharya and coworkers on FNOs in homogenization. Rather than relying upon continuity properties of the solution mapping $(\bar{\feps},\C) \mapsto \fu$, we emulate the constructive approximation of the solution to design a recurrent FNO \eqref{eq:FNO_LS_FLSO}, see also You et al.~\cite{you2022learning} and Kelly-Kalidindi~\cite{Kelly2025}. The obtained approximation result stated in Theorem \ref{thm:main} is astonishingly general. It appears that the traditional supervised training procedure serves as the most severe restriction for FNOs in homogenization.\\
Interestingly, the link that we established between FFT solvers and FNO came with a dictionary that permits us to translate between terminology available for both methods. Most importantly, "{}depth"{} of the operator network corresponds to the iteration count required by the FFT solver. Apart from this remark, we made a number of observations:
\begin{enumerate}
	\item Traditionally, FFT methods encounter two different kinds of errors~\cite{Ferrier2023}: the resolution error, a consequence of inaccurate discretization, and the iteration error, i.e., the "{}distance"{} to discrete equilibrium. In our construction of the FNO discussed in section \ref{sec:FNO_result}, both errors are mingled. The iteration error is reduced by increasing the depth of the FNO. There is no "{}classical"{} discretization error because we permit infinitely many Fourier coefficients to appear in the Fourier layer \eqref{eq:FNO_defn_convolution_Fourier}. There is no alternative, however, if general microstructures are to be treated, see the discussion surrounding eqns.~\eqref{eq:FNO_defn_rescaled_stiffness}--\eqref{eq:FNO_defn_vanishing strains}.
	\item Instead of the basic scheme~\cite{MoulinecSuquet1994,MoulinecSuquet1998}, more advanced solution schemes developed for FFT schemes could be used: polarization schemes~\cite{EyreMilton1999,MichelMoulinecSuquet2001,Donval2024Polarization}, fast gradient~\cite{Schneider2017Nesterov,Ernesti2020PhasefieldFracture} and conjugate gradient methods~\cite{ZemanCG2010,BrisardDormieux2010,NonlinearCG2020} or Quasi-Newton strategies~\cite{Chen2019Anderson,QuasiNewton2019,BB2019}. Some of these methods require much less than the $O(\kappa)$ iterations necessary for the basic scheme, where $\kappa = \alpha_+ / \alpha_-$ stands for the material contrast. However, no solver requires less than $O(\sqrt{\kappa})$ iterations. This fact is a consequence of general results on the minimum number of gradient evaluations required by iterative methods for approximating the minimizer of a strongly convex function with Lipschitz gradient, see Nesterov's book~\cite[Thm.~2.1.13]{NesterovBook}. It would be interesting to see whether FFT solvers could be improved by deep learning techniques, in particular whether the $O(\sqrt{\kappa})$ bound could be overcome. Indeed, FNOs do not require to compute the gradient of the objective function exactly, i.e., there is room for improvement.
	\item The FNO which we constructed was used for a fixed discretization scheme in section \ref{sec:computations}. However, a variety of discretization methods is available for FFT methods, including spectral~\cite{MoulinecSuquet1994,MoulinecSuquet1998,Bonnet2007}, finite difference~\cite{Willot2015,StaggeredGrid,Finel2025TET}, finite element~\cite{FFTFEM,FANS,Ladecky2022FEM}, finite volume~\cite{DornSchneider2019,CCMF2021} and more exotic discretizations~\cite{BrisardDormieux2012,BSplinesFFTFromWuhan,Jabs2024Radon}. The FFT solvers were designed under certain conditions on the discretization, i.e., that the symbol of the Eshelby-Green operator \eqref{eq:homogenization_basic_EshelbyGreen} is self-adjoint and non-expansive. Therefore, deep learning could be employed to learn better discretizations, as well.
	\item Composite-voxel methods~\cite{Gelebart2015CompositeVoxels,CompositeVoxelsElastic} serve as an indispensable tool for practical computational micromechanics. These associate an appropriate "{}mixing rule"{} to voxels comprising more than one microstructure phase. It is conceivable that FNOs might profit from such constructions, as well. Conversely, deep-learning techniques could help constructing more advanced "{}mixing rules"{} than the ones currently used.
	\item Our construction \eqref{eq:FNO_result_FNO_construction_finalExpression} was purely qualitative, but could be made \emph{quantitative} with some effort. More precisely, the number of Fourier layers and the number of neural-network layers for the ReLU network could be estimated as a function of the input parameters $\kappa$, $\eps_0$ and $\delta$. However, it is not clear whether this estimate would be optimal. The principal technical difficulty here is the optimality of the power $p$ in the a-priori $L^p$-bound \eqref{eq:homogenization_basic_Lp_bound}.
	\item It is known that FFT solvers may be applied to solve unit-cell problems involving porous materials as long as the microstructure is mechanically stable and a numerically robust discretization is used~\cite{PorousSLE2020}. Thus, FNOs may be constructed to solve such problems, as well. However, the fidelity of such operators depends strongly on specific quantities related to the interface, see Schneider~\cite{PorousSLE2020} for the porous case, and Sab et al.~\cite{Sab2024} for microstructures involving both pores and rigid inclusions. In particular, problems with "{}infinite material contrast"{} may be treated.
	\item The "{}F"{}-part in FNO is responsible for obtaining a uniform bound on the depth, a.k.a. iteration count. It is hard to imagine a uniform approximation result for homogenization problems using the Deep-O-Net~\cite{Lu2021DeepONet} similar to ours which was obtained for FNOs. It thus appears that FNOs are more suited to homogenization problems with their low-regularity solutions.
	\item Traditionally, FNOs were used to solve linear homogenization problems. Of course, the primary problem here is the dependence of the solution to the homogenization problem on the coefficients, which is strongly nonlinear. Still, some of the most interesting questions in computational mechanics involve nonlinear and inelastic constitutive laws. These requirements were taken into consideration when designing FFT methods~\cite{MoulinecSuquet1994,MoulinecSuquet1998}. Via the established link between FFT schemes and FNOs, the latter may be infused with knowledge how to handle inelasticity in an efficient manner.
\end{enumerate}

\section*{Acknowledgements}

BHN and MS acknowledge support from the European Research Council within the Horizon Europe program - project 101040238. MS is grateful to Burigede Liu and Michael Ortiz for inviting him to the IUTAM symposium on data driven mechanics 2025 where he was introduced to FNOs in homogenization. \review{We thank the anonymous referees for taking the time to review the manuscript and for providing constructive feedback.}

\section*{Data availability statement}

The data that support the findings of this study are available from the corresponding author upon reasonable request.

\appendix

\section{Elementary estimates for fixed-point methods}
\label{apx:fixed_point_estimates}

This appendix collects two rather elementary a-priori estimates valid for general fixed-point methods. We state them in this generality because of their inherent simplicity.\\
We start with a simple a-priori estimate for fixed-point methods
\begin{lem}
	For a contraction mapping $F$ on a Banach space $X$ with unique fixed point $x^* \in X$ and contraction constant $\gamma \in (0,1)$, the estimate
	\begin{equation}\label{eq:fixed_point_estimates_simpleAPrioriEstimateFixedPoint}
		\| x^* - x^0 \|_X \leq \frac{ \| F(x^0) - x^0\|_X }{ 1 - \gamma}
	\end{equation}
	holds for any element $x^0\in X$.
\end{lem}
\begin{proof}
	By the proof of the contraction mapping theorem, the inductively defined sequence
	\begin{equation}
		x^{k+1} = F(x^k), \quad k=0,1,2,\ldots,
	\end{equation}
	converges to the unique fixed point $x^*$. Re-writing the difference
	\begin{equation}
		x^{m+1} - x^0 = \sum_{k=0}^m (x^{k+1} - x^k), \quad m > 0,
	\end{equation}
	as a telescopic sum, we estimate
	\begin{equation}
		\| x^{m} - x^0 \|_X \leq  \sum_{k=0}^{m-1} \|x^{k+1} - x^k\|_X \leq \sum_{k=0}^{m-1} \gamma^k \|x^{1} - x^0\|_X = \frac{1 - \gamma^m}{1 - \gamma} \| F(x^0) - x^0 \|_X.
	\end{equation}
	Thus, we obtain the bound
	\begin{equation}
		\begin{split}
			\| x^* - x^0 \|_X &\leq \| x^* - x^m \|_X + \| x^{m} - x^0 \|_X\\
			&\leq \| x^* - x^m \|_X + \frac{1 - \gamma^m}{1 - \gamma} \| F(x^0) - x^0 \|_X \\
			&\rightarrow \frac{1}{1 - \gamma} \| F(x^0) - x^0 \|_X
		\end{split}
	\end{equation}
	for $m \rightarrow \infty$, which was to be shown.
\end{proof}

We turn our attention to an estimate for the fixed points associated to a perturbed fixed-point mapping.
\begin{lem}
	Let $X$ be a Banach space. Suppose that two Lipschitz mappings $F$ and $F_\theta$ are given with respective contraction constants $\gamma \in (0,1)$ and $\gamma_\theta \in (0,1)$. Then the estimate
	\begin{equation}\label{eq:fixed_point_estimates_main_estimate}
		\| x^* - x^*_\theta\|_X \leq \frac{\|F(x^*) - F_\theta(x^*)\|_X}{1 - \gamma_\theta}
	\end{equation}
	holds for the (unique) fixed points
	\begin{equation}
		x^* = F(x^*) \quad \text{and} \quad x^*_\theta = F(x^*_\theta).
	\end{equation}
\end{lem}
\begin{proof}
	We write
	\begin{equation}
		x^* - x^*_\theta = F(x^*) - F_\theta(x^*_\theta) = F(x^*) - F_\theta(x^*) + F_\theta(x^*) - F_\theta(x^*_\theta)
	\end{equation}
	and estimate
	\begin{equation}
		\|x^* - x^*_\theta\|_X \leq \| F(x^*) - F_\theta(x^*)\|_X + \|F_\theta(x^*) - F_\theta(x^*_\theta)\|_X
		\leq \| F(x^*) - F_\theta(x^*)\|_X + \gamma_\theta \|x^* - x^*_\theta\|_X.
	\end{equation}
	Thus, the inequality 
	\begin{equation}
		(1 - \gamma_\theta) \|x^* - x^*_\theta \|_X \leq \| F(x^*) - F_\theta(x^*)\|_X
	\end{equation}	
	holds, which was to be shown.
\end{proof}

\section{Arguments for the continuous basic scheme}
\label{apx:basic}

\subsection{Estimates for the Eshelby-Green operator}
\label{apx:basic_EshelbyGreen}

The purpose of this appendix is to provide the necessary arguments to show the validity of Lemma \eqref{lem:homogenization_basic_Lp_estimates_EshelbyGreen} which comprises three sub-statements:
\begin{enumerate}
	\item The tensor-valued Fourier multiplier gives rise to a bounded linear operator on $L^p(Y;\Sym{d})$ for exponents $1<p<\infty$, i.e., 
	\begin{equation}\label{eq:homogenization_basic_LpBound_EshelbyGreen_apx}
		\ffGamma \in L(L^p(Y;\Sym{d})), \quad 1 < p < \infty,
	\end{equation}
	\item For the exponent $p=2$, the Eshelby-Green operator \eqref{eq:homogenization_basic_EshelbyGreen} defines an orthogonal projector on the $L^2$-space \eqref{eq:homogenization_setup_L2} onto the (closed) subspace of compatible strains. In particular, the norm equality
	\begin{equation}\label{eq:homogenization_basic_L2Bound_EshelbyGreen_apx}
		\| \ffGamma \|_{L(L^2(Y;\Sym{d}))} = 1
	\end{equation}
	holds.
	\item For every constant $C > 1$, there is an exponent $\check{p}(C) > 2$, s.t., the inequality
	\begin{equation}\label{eq:homogenization_basic_ExplicitLpBound_EshelbyGreen_apx}
		\| \ffGamma \|_{L(L^p(Y;\Sym{d}))} \leq C \quad \text{holds for all} \quad 2 \leq p \leq \check{p}(C).
	\end{equation}
\end{enumerate}
We provide the arguments in the same order as the statements:
\begin{enumerate}
	\item For any $j=1,2,\ldots,d$, we define the Riesz transform $R_j$~\cite{SteinWeiss1971} as a Fourier multiplier acting on any function $f \in L^2(Y)$ via
	\begin{equation}\label{eq:homogenization_basic_lemma_proof_Riesz_defn}
		R_j f (\fx) = -\texttt{i} \sum_{\fxi \in \Z^d \backslash\{\bm{0}\}} \frac{\tilde{\xi}_j}{\|\tilde{\fxi}\|} \hat{f}(\fxi) \, e^{\texttt{i} \, \fx \cdot \tilde{\fxi}}, \quad \fx \in Y,
	\end{equation}
	Then, by the result of Iwaniec-Martin~\cite{IwaniecMartin1996Riesz} and the transference principle~\cite{SteinWeiss1971}, the operator-norm bound
	\begin{equation}\label{eq:homogenization_basic_lemma_proof_Riesz_bound}
		\|R_j\|_{L(L^p(Y))} \leq \cot \left( \frac{\pi}{2p^*} \right) \quad \text{for} \quad 1<p<\infty \quad \text{and} \quad p^* = \max\left( p, \frac{p}{p-1} \right)
	\end{equation}
	is valid. By the definitions \eqref{eq:homogenization_basic_EshelbyGreen_Fourier_coefficients}, \eqref{eq:homogenization_basic_EshelbyGreen_Fourier_coefficients_explicit} and \eqref{eq:homogenization_basic_lemma_proof_Riesz_defn}, we may express the action of the Eshelby-Green operator \eqref{eq:homogenization_basic_EshelbyGreen} in the component form
	\begin{equation}\label{eq:homogenization_basic_lemma_proof_Gamma_via_Riesz}
		\left(\ffGamma:\ftau\right)_{ij} = -\sum_{k=1}^d R_i R_k \tau_{jk} - \sum_{k=1}^d R_j R_k \tau_{ij} - \sum_{k,\ell = 1}^d R_i R_j R_k R_l \tau_{kl}, \quad i,j=1,\ldots,d.
	\end{equation}
	By the $L^p$-boundedness of the Riesz transform, the Eshelby-Green operator \eqref{eq:homogenization_basic_EshelbyGreen} is also $L^p$-bounded, as was to be shown.
	\item We only provide a sketch, because this property is well-known~\cite[Lemma 3.2]{Schneider2015Convergence}. It is straightforward to show that the Eshelby-Green operator \eqref{eq:homogenization_basic_EshelbyGreen} defines a projector. Moreover, the operator is $L^2$-self\review{-}adjoint. In particular, it defines an orthogonal projector. The formula \eqref{eq:homogenization_basic_L2Bound_EshelbyGreen_apx} follows directly, as every orthogonal projector satisfies the Pythagorean Theorem.
	\item This statement is a direct consequence of the Riesz-Thorin interpolation theorem~\cite{SteinWeiss1971}, the $L^2$-bound \eqref{eq:homogenization_basic_L2Bound_EshelbyGreen} and the $L^p$-boundedness \eqref{eq:homogenization_basic_LpBound_EshelbyGreen} for $p>2$. Indeed, fix some $\tilde{p} > 2$. Then, by the Riesz-Thorin interpolation theorem, we have
	\begin{equation}\label{eq:homogenization_basic_lemma_proof_RieszThorinStatement}
		\|\ffGamma \|_{L(L^{p_\theta}(Y;\Sym{d}))} \leq \|\ffGamma \|_{L(L^{2}(Y;\Sym{d}))}^{1-\theta} \|\ffGamma \|_{L(L^{\tilde{p}}(Y;\Sym{d}))}^\theta \equiv \|\ffGamma \|_{L(L^{\tilde{p}}(Y;\Sym{d}))}^\theta
	\end{equation}
	for $0 \leq \theta \leq 1$ and the exponent
	\begin{equation}\label{eq:homogenization_basic_lemma_proof_RieszThorinExponent}
		\frac{1}{p_\theta} = \frac{1 - \theta}{2} + \frac{\theta}{\tilde{p}}.
	\end{equation}
	We assume that the number
	\begin{equation}\label{eq:homogenization_basic_lemma_proof_misc1}
		A_{\tilde{p}} := \|\ffGamma \|_{L(L^{\tilde{p}}(Y;\Sym{d}))}
	\end{equation}
	is greater than the prescribed constant $C > 1$. Otherwise, the validity of the desired statement is immediate. The function
	\begin{equation}\label{eq:homogenization_basic_lemma_proof_misc2}
		f: [0,1] \rightarrow \R, \quad f(\theta) = A_{\tilde{p}}^\theta,
	\end{equation}
	is continuous, strictly monotonically increasing, for instance validated by studying the derivative, and attains the values
	\begin{equation}\label{eq:homogenization_basic_lemma_proof_misc3}
		f(0) = 1 \quad \text{and} \quad f(1) = A_{\tilde{p}}.
	\end{equation}
	Thus, by the intermediate value theorem, there is precisely one value $\check{\theta} \in [0,1]$, s.t. the condition
	\begin{equation}\label{eq:homogenization_basic_lemma_proof_misc4}
		A^{\check{\theta}} \equiv f(\check{\theta}) = C
	\end{equation}
	is satisfied. By the monotonicity of the function $f$, the estimate
	\begin{equation}\label{eq:homogenization_basic_lemma_proof_misc5}
		1 \leq f(\theta) \leq f(\check{\theta}) = C, \quad 0 \leq \theta \leq \check{\theta},
	\end{equation}
	holds, which we may translate into the form
	\begin{equation}\label{eq:homogenization_basic_lemma_proof_misc6}
		\|\ffGamma \|_{L(L^{p_\theta}(Y;\Sym{d}))} \leq C, \quad 0 \leq \theta \leq \check{\theta}.
	\end{equation}
	As the function $[0,1] \ni \theta \mapsto p_\theta \in [2,\tilde{p}]$ is strictly monotonically increasing, the statement \eqref{eq:homogenization_basic_lemma_proof_misc6} implies that the following estimate
	\begin{equation}\label{eq:homogenization_basic_lemma_proof_misc7}
		\|\ffGamma \|_{L(L^{p}(Y;\Sym{d}))} \leq C, \quad p \in [2, p_{\check{\theta}}],
	\end{equation}	
	holds, as claimed.
\end{enumerate}

\subsection{Contraction properties of the standard Lippmann-Schwinger operator}
\label{apx:basic_LSO}

This appendix provides the arguments for the validity of the statements claimed in Theorem \ref{thm:homogenization_basic_contraction_estimates_EshelbyGreen}, more precisely:
\begin{enumerate}
	\item For any prescribed strain $\bar{\feps} \in \Sym{d}$, stiffness field $\C \in \mathcal{M}(\alpha_-,\alpha_+)$ and reference constant $\alpha_0$, the Lippmann-Schwinger operator \eqref{eq:homogenization_basic_LSO1} has precisely one fixed point $\feps^*$ which has the form
	\begin{equation}\label{eq:homogenization_basic_fixed_point}
		\feps^* = \bar{\feps} + \nabla^s \fu,
	\end{equation}
	where the displacement fluctuation $\fu \in H^1_\#(Y)^d$ is the unique solution to the equilibrium equation \eqref{eq:homogenization_setup_equilibrium}. In particular, the fixed point \eqref{eq:homogenization_basic_fixed_point} does not depend on the reference constant $\alpha_0$.
	\item 
	For the choice
	\begin{equation}\label{eq:homogenization_basic_refMat_basic_apx}
		\alpha_0 = \frac{\alpha_+ - \alpha_-}{2},
	\end{equation}
	the Lippmann-Schwinger operator defines an $L^2$-contraction with constant
	\begin{equation}\label{eq:homogenization_basic_contraction_constant_apx}
		\gamma = \frac{\alpha_+ - \alpha_-}{\alpha_+ + \alpha_-} \equiv 1 - \frac{2}{\kappa + 1},
	\end{equation}
	i.e., the estimate
	\begin{equation}\label{eq:homogenization_basic_contraction_estimate_apx}
		\left\| \mathcal{L}(\feps_1; \bar{\feps},\C,\alpha_0) - \mathcal{L}(\feps_2; \bar{\feps},\C,\alpha_0) \right\|_{L^2} \leq \gamma\,\|\feps_1 - \feps_2\|_{L^2}
	\end{equation}
	holds for all $\feps_1, \feps_2 \in L^2(Y;\Sym{d})$.
	\item There are numbers $p \in (2,\infty)$ and $C_p > 0$, depending only on the pair $(\alpha_-,\alpha_+)$, s.t. the estimate
	\begin{equation}\label{eq:homogenization_basic_Lp_bound_apx}
		\|\feps^*\|_{L^p} \leq C_p \, \|\bar{\feps}\|
	\end{equation}
	holds for the fixed point \eqref{eq:homogenization_basic_fixed_point}.
\end{enumerate}
We provide the arguments in the same order:
\begin{enumerate}
	\item The Lippmann-Schwinger operator \eqref{eq:homogenization_basic_LSO2} may be written in the form
	\begin{equation}\label{eq:homogenization_basic_thm_proof_part1_1}
		\ffGamma = \nabla^s \circ \fG \circ \div
	\end{equation}
	with Green's operator
	\begin{equation}\label{eq:homogenization_basic_thm_proof_part1_2}
		\fG \in L(H^{-1}_\#(Y)^d, H^{1}_\#(Y)^d),
	\end{equation}
	the inverse of the bounded linear operator
	\begin{equation}\label{eq:homogenization_basic_thm_proof_part1_3}
		H^{-}_\#(Y)^d \rightarrow H^{-1}_\#(Y)^d, \quad \fu \mapsto \div \nabla^s \fu.
	\end{equation}
	Suppose that $\feps^* \in L^2(Y;\Sym{d})$ is a fixed point of the Lippmann-Schwinger equation		
	\begin{equation}\label{eq:homogenization_basic_thm_proof_part1_4}
		\feps^* = \feps - \frac{1}{\alpha_0}\ffGamma : \left[\C:\feps^* - \alpha_0\,\feps^* \right] \quad \text{with} \quad \ffGamma \equiv \nabla^s\fG \div.
	\end{equation}
	Then, the strain has the form
	\begin{equation}\label{eq:homogenization_basic_thm_proof_part1_5}
		\feps^* = \bar{\feps} + \nabla^s \fu^*
	\end{equation}
	with the displacement fluctuation
	\begin{equation}\label{eq:homogenization_basic_thm_proof_part1_6}
		\fu^* = - \frac{1}{\alpha_0} \fG \div \left[\C:\feps^* - \alpha_0\,\feps^* \right].
	\end{equation}
	The latter equation may be expanded into the form
	\begin{equation}\label{eq:homogenization_basic_thm_proof_part1_7}
		\fu^* = - \frac{1}{\alpha_0} \fG \div \C:\feps^* + \fG \div \,\feps^*.
	\end{equation}
	We observe
	\begin{equation}\label{eq:homogenization_basic_thm_proof_part1_8}
		\div \,\feps^* = \div \,\left[ \bar{\feps} + \nabla^s \fu^* \right] = \div \nabla^s \fu^*
	\end{equation}
	because the macroscopic strain $\bar{\feps}$ is homogeneous, i.e., constant. By definition of Green's operator, we are thus led to the identity
	\begin{equation}\label{eq:homogenization_basic_thm_proof_part1_9}
		\fG \div \,\feps^* = \fu^*.
	\end{equation}
	In view of the expression \eqref{eq:homogenization_basic_thm_proof_part1_7}, we deduce
	\begin{equation}\label{eq:homogenization_basic_thm_proof_part1_10}
		\fu^* = - \frac{1}{\alpha_0} \fG \div \C:\feps^* + \fG \div \,\feps^* = - \frac{1}{\alpha_0} \fG \div \C:\feps^* + \fu^*,
	\end{equation}
	i.e., the validity of the equation
	\begin{equation}\label{eq:homogenization_basic_thm_proof_part1_11}
		\bm{0} = \fG \div \C:\feps^* \equiv \fG \div \C:\left[ \bar{\feps} + \nabla^s \fu^*\right].
	\end{equation}
	As Green's operator is a bounded linear isomorphism, the displacement fluctuation satisfies the equilibrium condition \eqref{eq:homogenization_setup_equilibrium}, as claimed.
	\item As the Eshelby-Green operator \eqref{eq:homogenization_basic_L2Bound_EshelbyGreen} is non-expansive \eqref{eq:homogenization_basic_EshelbyGreen} , the representation formula
	\begin{equation}\label{eq:homogenization_basic_thm_proof_part2_1}
		\mathcal{L}(\feps_1; \bar{\feps},\C,\alpha_0) - \mathcal{L}(\feps_2; \bar{\feps},\C,\alpha_0) = - \frac{1}{\alpha_0}\ffGamma : (\C - \C^0):\left[\feps_1 - \feps_2\right], \quad \feps_1, \feps_2 \in L^2(Y;\Sym{d}),
	\end{equation}
	involving the Lippmann-Schwinger operator \eqref{eq:homogenization_basic_LSO2} implies the estimate
	\begin{equation}\label{eq:homogenization_basic_thm_proof_part2_2}
		\left\| \mathcal{L}(\feps_1; \bar{\feps},\C,\alpha_0) - \mathcal{L}(\feps_2; \bar{\feps},\C,\alpha_0) \right\|_{L^2} \leq \left\| \frac{1}{\alpha_0} (\C - \C^0):\left[\feps_1 - \feps_2\right] \right\|_{L^2}, \quad \feps_1, \feps_2 \in L^2(Y;\Sym{d}).
	\end{equation}
	Due to the assumed lower and upper bounds \eqref{eq:homogenization_setup_microstructures}
	\begin{equation}\label{eq:homogenization_basic_thm_proof_part2_3}
		\alpha_- \, \|\fxi\|^2 \leq \fxi : \C(\fx) : \fxi \leq \alpha_+ \, \|\fxi\|^2, \quad \fxi \in \Sym{d},
	\end{equation}
	valid for a.e. $\fx \in Y$, the inequality
	\begin{equation}\label{eq:homogenization_basic_thm_proof_part2_4}
		\left\| \frac{1}{\alpha_0} (\C(\fx) - \C^0): \fxi\right\| \leq  \max\left( \frac{|\alpha_- - \alpha_0|}{\alpha_0}, \frac{|\alpha_+ - \alpha_0|}{\alpha_0} \right) \|\fxi\|, \quad \fxi \in \Sym{d},
	\end{equation}
	follows from eigenvalue considerations w.r.t. the Frobenius norm for a.e. $\fx \in Y$. For the special case
	\begin{equation}\label{eq:homogenization_basic_thm_proof_part2_5}
		\alpha_0 = \frac{\alpha_- + \alpha_+}{2}, \quad \text{i.e.}, \quad \max\left( \frac{|\alpha_- - \alpha_0|}{\alpha_0}, \frac{|\alpha_+ - \alpha_0|}{\alpha_0} \right) = \frac{\alpha_+ - \alpha_-}{\alpha_+ + \alpha_-},
	\end{equation}
	we are led to the bound
	\begin{equation}\label{eq:homogenization_basic_thm_proof_part2_6}
		\left\| \frac{1}{\alpha_0} (\C(\fx) - \C^0): \fxi\right\| \leq   \frac{\alpha_+ - \alpha_-}{\alpha_+ + \alpha_-} \|\fxi\|, \quad \fxi \in \Sym{d},
	\end{equation}		
	Applied to the previously established inequality \eqref{eq:homogenization_basic_thm_proof_part2_2}, we conclude the desired contraction estimate \eqref{eq:homogenization_basic_contraction_estimate_apx}.
	\item We define the constant
	\begin{equation}\label{eq:homogenization_basic_thm_proof_part3_1}
		C = \frac{1 - \frac{1}{\kappa + 1}}{1 - \frac{2}{\kappa + 1}} > 1.
	\end{equation}
	By Lemma \ref{lem:homogenization_basic_Lp_estimates_EshelbyGreen}, part 3, there is an exponent $\check{p} > 2$, s.t. the estimate
	\begin{equation}\label{eq:homogenization_basic_thm_proof_part3_2}
		\| \ffGamma \|_{L(L^{\check{p}}(Y;\Sym{d}))} \leq C
	\end{equation}
	holds for the Eshelby-Green operator. We will show that the Lippmann-Schwinger operator \eqref{eq:homogenization_basic_LSO2}
	\begin{equation}\label{eq:homogenization_basic_thm_proof_part3_3}
		\mathcal{L}(\feps; \bar{\feps},\C,\alpha_0) = \bar{\feps} - \frac{1}{\alpha_0}\ffGamma : (\C - \C^0):\feps, \quad \feps \in L^{\check{p}}(Y;\Sym{d}),
	\end{equation}
	and the choice \eqref{eq:homogenization_basic_refMat_basic} defines an $L^{\check{p}}$-contraction with contraction constant
	\begin{equation}\label{eq:homogenization_basic_thm_proof_part3_4}
		\check{\gamma} = 1 - \frac{1}{\kappa + 1}.
	\end{equation}
	We estimate, for $\feps_1,\feps_2 \in L^{\check{p}}(Y;\Sym{d})$, 
	\begin{equation}\label{eq:homogenization_basic_thm_proof_part3_5}
		\begin{split}
			\left\| \mathcal{L}(\feps_1; \bar{\feps},\C,\alpha_0) - \mathcal{L}(\feps_2; \bar{\feps},\C,\alpha_0) \right\|_{L^{\check{p}}} 
			&= \left\|  \ffGamma : \frac{\C - \C^0}{\alpha_0}:\left[\feps_1 - \feps_2 \right] \right\|_{L^{\check{p}}}\\
			&\leq C \left\|  \frac{\C - \C^0}{\alpha_0}:\left[\feps_1 - \feps_2 \right] \right\|_{L^{\check{p}}}\\
			&\leq C \gamma \left\| \feps_1 - \feps_2 \right\|_{L^{\check{p}}}\\
			&\leq \check{\gamma} \left\| \feps_1 - \feps_2 \right\|_{L^{\check{p}}}\\
		\end{split}
	\end{equation}
	where we used the $L^{\check{p}}$-bound \eqref{eq:homogenization_basic_thm_proof_part3_2} and the contraction estimate \eqref{eq:homogenization_basic_thm_proof_part2_6}. The last line follows from the definition $C = \check{\gamma} / \gamma$, see eq. \eqref{eq:homogenization_basic_thm_proof_part3_1}. Thus, the operator \eqref{eq:homogenization_basic_thm_proof_part3_3} is an $L^{\check{p}}$-contraction with contraction constant $\tilde{\gamma} \in (0,1)$. Then, we invoke the a-priori estimate
	\eqref{eq:fixed_point_estimates_simpleAPrioriEstimateFixedPoint}
	\begin{equation}\label{eq:homogenization_basic_thm_proof_part3_6}
		\| x^* - x^0 \|_X \leq \frac{ \| F(x^0) - x^0\|_X }{ 1 - \check{\gamma}}
	\end{equation}
	for $X=L^{\check{p}}(Y;\Sym{d})$ and $x^0 = \bar{\feps}$, i.e., we obtain the bound
	\begin{equation}\label{eq:homogenization_basic_thm_proof_part3_7}
			\| \feps^* - \bar{\feps} \|_{L^{\check{p}}} \leq \frac{1}{ 1 - \check{\gamma}} \, \left\|  \ffGamma : \frac{\C - \C^0}{\alpha_0}:\bar{\feps} \right\|_{L^{\check{p}}} \leq \frac{\check{\gamma}}{ 1 - \check{\gamma}} \, \|\bar{\feps}\|.
	\end{equation}
	In particular, we are led to the estimate
	\begin{equation}\label{eq:homogenization_basic_thm_proof_part3_8}
		\| \feps^* \|_{L^{\check{p}}}  \leq \| \bar{\feps} \|_{L^{\check{p}}} + \| \feps^* - \bar{\feps} \|_{L^{\check{p}}}
		 \leq \left( 1 + \frac{\check{\gamma}}{ 1 - \check{\gamma}} \right) \| \bar{\feps} \|_{L^{\check{p}}} \equiv \frac{1}{ 1 - \check{\gamma}}\| \bar{\feps} \|,
	\end{equation}
	which we may also write in the form
	\begin{equation}\label{eq:homogenization_basic_thm_proof_part3_9}
		\| \feps^* \|_{L^{\check{p}}} \leq (1 + \kappa)\,\| \bar{\feps} \|.
	\end{equation}
\end{enumerate}

\section{Contraction estimates for the Lippmann-Schwinger neural operator}
\label{apx:LSNO}

The purpose of this appendix is to provide the arguments for the validity of Prop.~\ref{prop:FNO_LS_FLNO}. More precisely,
under the assumption \ref{ass:mtheta}, we fix a field
\begin{equation}\label{eq:FNO_LS_StiffnessSpace_apx}
	\C \in \mathcal{M}_Y(\alpha_-,\alpha_+)
\end{equation}
of linear elastic stiffness tensors and a macroscopic strain $\bar{\feps} \in \Sym{d}$. Then, the following properties hold:
\begin{enumerate}
	\item The neural operator \eqref{eq:FNO_LS_FLSO} is well-defined and satisfies the estimate
	\begin{equation}\label{eq:FNO_LS_FLNO_estimate_apx}
		\left\| \mathcal{L}_\theta(\feps_1; \bar{\feps},\C,\alpha_0) - \mathcal{L}_\theta(\feps_2; \bar{\feps},\C,\alpha_0) \right\|_{L^2} \leq (\gamma_0 + \delta) \|\feps_1 - \feps_2\|_{L^2}
	\end{equation}
	in case the quantity
	\begin{equation}\label{eq:FNO_LS_FLNO_contraction_constant_apx}
		\gamma_0 = \max \left( \left| \frac{\alpha_+ - \alpha_0}{\alpha_0}\right|, \left| \frac{\alpha_- - \alpha_0}{\alpha_0}\right| \right)
	\end{equation}
	does not exceed unity.
	\item In case the quantity
	\begin{equation}\label{eq:FNO_LS_FLNO_combined_contraction_constant_apx}
		\gamma_\theta = \gamma_0 + \delta
	\end{equation}
	is less than unity, the unique fixed point $\feps_{\theta}^*$ of the contractive neural operator \eqref{eq:FNO_LS_FLSO} satisfies the proximity estimate
	\begin{equation}\label{eq:FNO_LS_FLNO_proximity_estimate_apx}
		\| \feps^* - \feps^*_\theta\|_{L^2} \leq \frac{1}{1 - \gamma_\theta} \left[ \delta\,\| \feps^* \|_{L^2} + C \, \frac{\|\feps^* \|_{L^p}^{\frac{p}{2}}}{M^{\frac{p-2}{2}}} \right]
	\end{equation}
	to the (unique) fixed point $\feps^*$ of the Lippmann-Schwinger operator \eqref{eq:homogenization_basic_LSO2} which resides in an $L^p$ space with an exponent $p>2$ that depends only on the material contrast $\kappa = \alpha_+ / \alpha_-$, see the statement \eqref{eq:homogenization_basic_Lp_bound}.
\end{enumerate}
\begin{enumerate}
	\item \underline{Well-posedness:}
	We want to show that the action of the operator \eqref{eq:FNO_LS_FLSO}
	\begin{equation}\label{eq:FNO_LS_prop_proof_part1_1}
		\mathcal{L}_\theta(\feps; \bar{\feps},\C,\alpha_0) = \bar{\feps} - \ffGamma : \Multiplication_\theta \left( \frac{1}{\alpha_0}(\C - \C^0), \feps \right)
	\end{equation}
	leads to an element of the $L^2$-space for every input $\feps \in L^2(Y;\Sym{d})$. By the triangle inequality and the non-expansivity \eqref{eq:homogenization_basic_L2Bound_EshelbyGreen} of the Eshelby-Green operator \eqref{eq:homogenization_basic_EshelbyGreen}, we observe
	\begin{equation}\label{eq:FNO_LS_prop_proof_part1_2}
		\begin{split}
			\left\| \mathcal{L}_\theta(\feps; \bar{\feps},\C,\alpha_0)\right\| _{L^2}
			& \leq \|\bar{\feps}\| + \left\| \ffGamma : \Multiplication_\theta \left( \frac{1}{\alpha_0}(\C - \C^0), \feps \right) \right\|_{L^2}\\
			& \leq \|\bar{\feps}\| + \left\| \Multiplication_\theta \left( \frac{1}{\alpha_0}(\C - \C^0), \feps \right) \right\|_{L^2}.
		\end{split}
	\end{equation}
	Thus is suffices to show that the field
	\begin{equation}\label{eq:FNO_LS_prop_proof_part1_3}
		\Multiplication_\theta \left( \mathds{T}^0 , \feps \right)
	\end{equation}
	defines an element of the $L^2$-space, where we introduced the tensor field
	\begin{equation}\label{eq:FNO_LS_prop_proof_part1_4}
		\mathds{T}^0 = \frac{1}{\alpha_0}(\C - \C^0).
	\end{equation}
	By assumption \eqref{eq:FNO_LS_FLNO_contraction_constant}, the field \eqref{eq:FNO_LS_prop_proof_part1_4} does not exceed unity in spectral norm for almost every point. Thus, we may use the estimate \eqref{eq:ass_mtheta_smallInputEstimate} to derive the bound
	\begin{equation}\label{eq:FNO_LS_prop_proof_part1_5}
		\begin{split}
		\left\| \Multiplication_\theta \left( \mathds{T}^0(\fx) , \feps(\fx) \right)\right\| 
		&\leq \left\|\Multiplication_\theta \left( \mathds{T}^0(\fx) , \feps(\fx) \right) - \mathds{T}^0(\fx):\feps(\fx)\right\| + \left\|\mathds{T}^0(\fx):\feps(\fx)\right\|\\
		&\leq (1 + \delta)\left\|\feps(\fx)\right\|
		\end{split}
	\end{equation}
	for almost every point $\fx \in Y$, s.t. the estimate
	\begin{equation}\label{eq:FNO_LS_prop_proof_part1_6}
		\| \feps(\fx) \| \leq M
	\end{equation}
	is satisfied. On the other hand, in case the bound
	\begin{equation}\label{eq:FNO_LS_prop_proof_part1_7}
		\| \feps(\fx) \| \geq M
	\end{equation}
	holds, we may use the estimate \eqref{eq:ass_mtheta_largeInputEstimate} in the form
	\begin{equation}\label{eq:FNO_LS_prop_proof_part1_8}
		\left\| \Multiplication_\theta \left( \mathds{T}^0(\fx) , \feps(\fx) \right)\right\| \leq \left\| \Multiplication_\theta \left( \mathds{T}^0(\fx) , \feps(\fx) \right) - \mathds{T}^0(\fx) :\feps(\fx) \right\| + \|\mathds{T}^0(\fx) :\feps(\fx)\| 
		\leq (C +1 )\left\|\feps(\fx)\right\|.
	\end{equation}
	Thus, we obtain the combined estimate
	\begin{equation}\label{eq:FNO_LS_prop_proof_part1_9}
		\left\| \Multiplication_\theta \left( \mathds{T}^0(\fx) , \feps(\fx) \right)\right\| 
		\leq \left( 1 + \max\left( \delta, C\right) \right) \left\|\feps(\fx)\right\|,
	\end{equation}
	which implies the inequality
	\begin{equation}\label{eq:FNO_LS_prop_proof_part1_10}
		\left\| \Multiplication_\theta \left( \mathds{T}^0 , \feps \right)\right\|_{L^2}
		\leq \left( 1 + \max\left( \delta, C\right) \right) \left\|\feps\right\|_{L^2},
	\end{equation}
	that implies the claimed well-posedness.\\		
	\underline{Lipschitz estimate:}
	We consider the tensor
	\begin{equation}\label{eq:FNO_LS_prop_proof_part1_11}
		\mathds{T}^0 = \frac{1}{\alpha_0}(\C - \C^0),
	\end{equation}
	which satisfies the estimate
	\begin{equation}\label{eq:FNO_LS_prop_proof_part1_12}
		\|\mathds{T}^0(\fx)\| \leq \gamma_0, \quad \text{a.e. }\fx \in Y,
	\end{equation}		
	in spectral norm \eqref{eq:ass_mtheta_smallInputEstimate_prereq}. Thus, by assumption \eqref{eq:ass_mtheta_smallInputEstimate}, we obtain the estimate
	\begin{equation}\label{eq:FNO_LS_prop_proof_part1_13}
		\|\Multiplication_\theta(\mathds{T}^0(\fx),\feps_1(\fx)) - \Multiplication_\theta (\mathds{T}^0(\fx),\feps_2(\fx))\| \leq (\gamma_0 + \delta) \, \| \feps_1(\fx) - \feps_2(\fx)\|
	\end{equation}
	for almost every point $\fx \in Y$. As the Eshelby-Green operator is an $L^2$-orthogonal projector, in particular non-expansive \eqref{eq:homogenization_basic_L2Bound_EshelbyGreen}, the representation
	\begin{equation}\label{eq:FNO_LS_prop_proof_part1_14}
		\mathcal{L}_\theta(\feps_1; \bar{\feps},\C,\alpha_0) - \mathcal{L}_\theta(\feps_2; \bar{\feps},\C,\alpha_0) = - \ffGamma : \left[ \Multiplication_\theta \left( \frac{1}{\alpha_0}(\C - \C^0), \feps_1 \right) - \Multiplication_\theta \left( \frac{1}{\alpha_0}(\C - \C^0), \feps_2 \right) \right]
	\end{equation}
	leads to the estimate
	\begin{equation}\label{eq:FNO_LS_prop_proof_part1_15}
		\begin{split}
		\|\mathcal{L}_\theta(\feps_1; \bar{\feps},\C,\alpha_0) - \mathcal{L}_\theta(\feps_2; \bar{\feps},\C,\alpha_0)\|_{L^2} &= \left\| \ffGamma : \left[ \Multiplication_\theta \left( \mathds{T}^0, \feps_1 \right) - \Multiplication_\theta \left( \mathds{T}^0, \feps_2 \right) \right] \right\|_{L^2}\\
		&\leq \left\| \Multiplication_\theta \left( \mathds{T}^0, \feps_1 \right) - \Multiplication_\theta \left( \mathds{T}^0, \feps_2 \right) \right\|_{L^2}\\
		&\leq (\gamma_0 + \delta) \, \left\| \feps_1 - \feps_2 \right\|_{L^2},
		\end{split}
	\end{equation}
	where we made use of the bound \eqref{eq:FNO_LS_prop_proof_part1_13}. Thus, the estimate \eqref{eq:FNO_LS_FLNO_estimate_apx} is established.
\item By the general perturbation estimate on contraction mappings \eqref{eq:fixed_point_estimates_main_estimate}, the estimate
	\begin{equation}\label{eq:FNO_LS_prop_proof_part2_1}
		\| \feps^* - \feps^*_\theta\|_{L^2} \leq \frac{\| \mathcal{L}(\bar{\feps},\C,\alpha_0; \feps^*) - \mathcal{L}_\theta(\feps^*; \bar{\feps},\C,\alpha_0) \|_{L^2}}{1 - \gamma_\theta}
	\end{equation}
	holds. By the representation
	\begin{equation}\label{eq:FNO_LS_prop_proof_part2_2}
		\mathcal{L}(\bar{\feps},\C,\alpha_0; \feps^*) - \mathcal{L}_\theta(\feps^*; \bar{\feps},\C,\alpha_0) = - \ffGamma : \left[ \mathds{T}^0:\feps^* - \Multiplication_\theta \left(\mathds{T}^0, \feps^* \right) \right],
	\end{equation}
	the estimate
	\begin{equation}\label{eq:FNO_LS_prop_proof_part2_3}
		\left\| \mathcal{L}(\bar{\feps},\C,\alpha_0; \feps^*) - \mathcal{L}_\theta(\feps^*; \bar{\feps},\C,\alpha_0) \right\|_{L^2} \leq \left\| \mathds{T}^0:\feps^* - \Multiplication_\theta \left(\mathds{T}^0, \feps^* \right) \right\|_{L^2}
	\end{equation}
	follows. We introduce the subsets
	\begin{equation}\label{eq:FNO_LS_prop_proof_part2_4}
		Y_{\leq M} = \left\{ \fx \in Y \,\middle|\, \|\feps^*(\fx)\| \leq M \right\} \quad \text{and} \quad Y_{>M} = \left\{ \fx \in Y \,\middle|\, \|\feps^*(\fx)\| > M \right\}
	\end{equation}
	of the unit cell. On the one hand, the inequality
	\begin{equation}\label{eq:FNO_LS_prop_proof_part2_5}
		\left\| \mathds{T}^0(\fx):\feps^*(\fx) - \Multiplication_\theta \left(\mathds{T}^0(\fx), \feps^*(\fx) \right) \right\| \leq \delta\,\| \feps^*(\fx) \|, \quad \text{a.e. } \fx \in Y_{\leq M},
	\end{equation}
	follows by assumption \eqref{eq:ass_mtheta_smallInputEstimate}, which leads to the estimate
	\begin{equation}\label{eq:FNO_LS_prop_proof_part2_6}
		\left(\frac{1}{|Y|}\int_{Y_{\leq M}} \left\| \mathds{T}^0(\fx):\feps^*(\fx) - \Multiplication_\theta \left(\mathds{T}^0(\fx), \feps^*(\fx) \right) \right\|^2 \, d\fx\right)^{\frac{1}{2}} \leq \delta\,\| \feps^* \|_{L^2}.
	\end{equation}
	On the other hand, the bound \eqref{eq:ass_mtheta_largeInputEstimate}
	\begin{equation}\label{eq:FNO_LS_prop_proof_part2_7}
		\left\| \mathds{T}^0(\fx):\feps^*(\fx) - \Multiplication_\theta \left(\mathds{T}^0(\fx), \feps^*(\fx) \right) \right\| \leq C\,\| \feps^*(\fx) \|, \quad \text{a.e. } \fx \in Y_{> M},
	\end{equation}
	implies the estimate
	\begin{equation}\label{eq:FNO_LS_prop_proof_part2_8}
		\frac{1}{|Y|}\int_{Y_{> M}} \left\| \mathds{T}^0(\fx):\feps^*(\fx) - \Multiplication_\theta \left(\mathds{T}^0(\fx), \feps^*(\fx) \right) \right\|^2 \, d\fx \leq C^2\, \frac{1}{|Y|}\int_{Y_{> M}} \left\| \feps^*(\fx) \right\|^2 \, d\fx.
	\end{equation}
	The chain of inequalities
	\begin{equation}\label{eq:FNO_LS_prop_proof_part2_9}
		M^{p-2} \frac{1}{|Y|}\int_{Y_{> M}} \left\| \feps^* \right\|^2 \, d\fx \leq \frac{1}{|Y|}\int_{Y_{> M}} \left\| \feps^* \right\|^p \, d\fx \leq \frac{1}{|Y|}\int_{Y} \left\| \feps^* \right\|^p \, d\fx = \|\feps^* \|_{L^p}^p
	\end{equation}
	leads to the estimate
	\begin{equation}\label{eq:FNO_LS_prop_proof_part2_10}
		\left(\frac{1}{|Y|}\int_{Y_{> M}} \left\| \feps^* \right\|^2 \, d\fx\right)^{\frac{1}{2}} \leq \frac{\|\feps^* \|_{L^p}^{\frac{p}{2}}}{M^{\frac{p}{2}-1}}
	\end{equation}
	Combining the bounds \eqref{eq:FNO_LS_prop_proof_part2_6} and \eqref{eq:FNO_LS_prop_proof_part2_10}, we are led to the inequality
	\begin{equation}\label{eq:FNO_LS_prop_proof_part2_11}
		\left\| \mathds{T}^0:\feps^* - \Multiplication_\theta \left(\mathds{T}^0, \feps^* \right) \right\|_{L^2} \leq \delta\,\| \feps^* \|_{L^2} + C \, \frac{\|\feps^* \|_{L^p}^{\frac{p}{2}}}{M^{\frac{p}{2}-1}},
	\end{equation}	
	which we combine with the estimate \eqref{eq:FNO_LS_prop_proof_part2_1} to obtain the claimed result
	\begin{equation}\label{eq:FNO_LS_prop_proof_part2_12}
		(1 - \gamma_\theta) \| \feps^* - \feps^*_\theta\|_{L^2} \leq \delta\,\| \feps^* \|_{L^2} + C \, \frac{\|\feps^* \|_{L^p}^{\frac{p}{2}}}{M^{\frac{p-2}{2}}}.
	\end{equation}
\end{enumerate}

\section{Detailed arguments for the uniform expressivity of FNOs}
\label{apx:main_result_arguments}

The purpose of the appendix is to provide a proof for Thm.~\ref{thm:main}.
To analyze the fidelity of the constructed operator \eqref{eq:FNO_result_FNO_construction_finalExpression}, we notice the equivalence \eqref{eq:FNO_result_FNO_output_as_iterate}
\begin{equation}\label{eq:FNO_result_proof1}
	\feps^{K+1}_\theta = \bar{\feps} + \nabla^s \fu \quad \text{with} \quad \fu = \mathcal{F}_\theta(\bar{\feps}, \C)
\end{equation}
with the iteratively defined sequence
\begin{equation}\label{eq:FNO_result_proof2}
	\feps^{k+1}_\theta = \mathcal{L}_\theta\left(\feps^k_\theta; \bar{\feps},\C, \frac{\alpha_- + \alpha_+}{2}\right), \quad \feps^0_\theta = \bar{\feps},
\end{equation}
where $k=0,1,2,\ldots$.
To estimate the error, we use the splitting
\begin{equation}\label{eq:FNO_result_proof3}
	\tilde{\mathcal{F}}_\theta(\bar{\feps}, \C) - \feps^* = \feps^{K+1}_\theta - \feps^* = \feps^{K+1}_\theta - \feps^*_\theta + \feps^*_\theta- \feps^*,
\end{equation}
where $\feps^* \in L^2(Y;\Sym{d})$ denotes the unique fixed point of the Lippmann-Schwinger operator
\begin{equation}\label{eq:FNO_result_proof4}
	\feps^* = \mathcal{L}\left(\bar{\feps},\C, \frac{\alpha_- + \alpha_+}{2};\feps^*\right)
\end{equation}
and $\feps^*_\theta$ refers to the fixed point of the neural Lippmann-Schwinger operator
\begin{equation}\label{eq:FNO_result_proof5}
	\feps^*_\theta = \mathcal{L}_\theta\left(\feps^*_\theta; \bar{\feps},\C, \frac{\alpha_- + \alpha_+}{2}\right)
\end{equation}
which is unique in case the parameter $\delta > 0$ is small enough. To be precise, we select the parameter $\delta$ in the range
\begin{equation}\label{eq:FNO_result_proof6}
	\delta \in \left( 0, \frac{1}{\kappa + 1}\right] \quad \text{with the material contrast} \quad \kappa = \frac{\alpha_+}{\alpha_-}.
\end{equation}
Then, by the estimate \eqref{eq:FNO_LS_FLNO_estimate}, the operator
\begin{equation}\label{eq:FNO_result_proof7}
	\mathcal{L}_\theta\left(\cdot; \bar{\feps},\C, \frac{\alpha_- + \alpha_+}{2}\right) \in L(L^2(Y;\Sym{d}))
\end{equation}
defines a contraction with (possibly non-optimal) contraction constant
\begin{equation}\label{eq:FNO_result_proof8}
	\gamma_\theta = 1 - \frac{1}{\kappa + 1}.
\end{equation}
Moreover, the proximity estimate \eqref{eq:FNO_LS_FLNO_proximity_estimate}
\begin{equation}\label{eq:FNO_result_proof9}
	\| \feps^* - \feps^*_\theta\|_{L^2} \leq \frac{1}{1 - \gamma_\theta} \left[ \delta\,\| \feps^* \|_{L^2} + C \, \frac{\|\feps^* \|_{L^p}^{\frac{p}{2}}}{M^{\frac{p-2}{2}}} \right]
\end{equation}
is satisfied. Taking into account the bound \eqref{eq:homogenization_basic_Lp_bound}	
\begin{equation}\label{eq:FNO_result_proof10}
	\|\feps^*\|_{L^p} \leq C_p \, \|\bar{\feps}\| \leq C_p \, \eps_0,
\end{equation}
we are led to the bound
\begin{equation}\label{eq:FNO_result_proof11}
	\| \feps^* - \feps^*_\theta\|_{L^2} \leq \frac{1}{1 - \gamma_\theta} \left[ \delta\, C_2 \, \eps_0 + C \, \frac{C_p^{\frac{p}{2}} \, \eps_0^{\frac{p}{2}}}{M^{\frac{p-2}{2}}} \right].
\end{equation}
To estimate the iteration error, we use the contraction property of the FNO \eqref{eq:FNO_LS_FLNO_estimate} and the a-priori estimate \eqref{eq:FNO_LS_FLNO_proximity_estimate}
\begin{equation}\label{eq:FNO_result_proof12}
	\| \feps^K_\theta - \feps^*_\theta\|_{L^2} \leq \gamma_\theta^{K+1} \| \feps^0_\theta - \feps^*_\theta\|_{L^2} \leq \frac{\gamma_\theta^{K+1}}{1 - \gamma_{\theta}} \left\| \mathcal{L}_\theta\left(\bar{\feps}; \bar{\feps},\C, \frac{\alpha_- + \alpha_+}{2} \right) - \bar{\feps} \right\|_{L^2}.
\end{equation}
To proceed, we take a look at the definition \eqref{eq:FNO_LS_FLSO},
\begin{equation}\label{eq:FNO_result_proof13}
	\mathcal{L}_\theta\left(\bar{\feps}; \bar{\feps},\C, \frac{\alpha_- + \alpha_+}{2}\right) - \bar{\feps} = - \ffGamma : \Multiplication_\theta \left( \frac{1}{\alpha_0}(\C - \C^0), \bar{\feps} \right),
\end{equation}
and utilize the non-expansivity \eqref{eq:homogenization_basic_L2Bound_EshelbyGreen} of the Eshelby-Green operator
\begin{equation}\label{eq:FNO_result_proof14}
	\left\| \mathcal{L}_\theta\left( \bar{\feps}; \bar{\feps},\C, \frac{\alpha_- + \alpha_+}{2} \right) - \bar{\feps} \right\|_{L^2} \leq \left\| \Multiplication_\theta \left( \frac{1}{\alpha_0}(\C - \C^0), \bar{\feps} \right) \right\|_{L^2} \leq \gamma_\theta \|\bar{\feps}\| \leq \gamma_\theta \eps_0
\end{equation}
to deduce the inequality
\begin{equation}\label{eq:FNO_result_proof15}
	\| \feps^K_\theta - \feps^*_\theta\|_{L^2} \leq \frac{\gamma_\theta^{K+2}}{1 - \gamma_{\theta}} \eps_0.
\end{equation}
where we implicitly assumed the validity of the estimate
\begin{equation}\label{eq:FNO_result_proof16}
	M \geq \eps_0.
\end{equation}
Combining the decomposition \eqref{eq:FNO_result_proof3} with the estimates \eqref{eq:FNO_result_proof11} and \eqref{eq:FNO_result_proof15}, we are thus led to the bound
\begin{equation}\label{eq:FNO_result_proof17}
	\left\| \tilde{\mathcal{F}}_\theta(\bar{\feps}, \C) - \feps^* \right\|_{L^2} 
	\leq \gamma_\theta^{K+2} \, \frac{\eps_0}{1 - \gamma_{\theta}} + \delta\, \frac{C_2 \, \eps_0}{1 - \gamma_\theta} + \frac{1}{M^{\frac{p-2}{2}}}  \frac{C \, C_p^{\frac{p}{2}} \, \eps_0^{\frac{p}{2}}}{1 - \gamma_\theta}.
\end{equation}	
Then, we select the parameters \eqref{eq:FNO_result_parameters} s.t. the inequalities
\begin{align}
	\gamma_\theta^{K+2} \, \frac{\eps_0}{1 - \gamma_{\theta}} &\leq \frac{\delta_{\texttt{target}}}{3}, \label{eq:FNO_result_proof18}\\
	\delta\, \frac{C_2 \, \eps_0}{1 - \gamma_\theta} &\leq \frac{\delta_{\texttt{target}}}{3}, \label{eq:FNO_result_proof19}\\
	\frac{1}{M^{\frac{p-2}{2}}}  \frac{C \, C_p^{\frac{p}{2}} \, \eps_0^{\frac{p}{2}}}{1 - \gamma_\theta} &\leq \frac{\delta_{\texttt{target}}}{3}, \label{eq:FNO_result_proof20}
\end{align}
are satisfied, and the desired estimate \eqref{eq:FNO_result_contributions_result} holds. More precisely, we set
\begin{align}
	K  \,  &= \left\lceil \frac{\ln \left( \frac{\delta_{\texttt{target}}}{3} \frac{1 - \gamma_{\theta}}{\eps_0} \right)}{\ln \gamma_\theta} \right\rceil, \label{eq:FNO_result_proof21}\\
	\delta\,  &= \min \left( \frac{\delta_{\texttt{target}}}{3} \frac{1 - \gamma_\theta}{C_2 \, \eps_0}, \frac{1}{\kappa + 1} \right), \label{eq:FNO_result_proof22}\\
	 M & = \max \left( \sqrt[\frac{2}{p-2}]{\frac{C \, C_p^{\frac{p}{2}} \, \eps_0^{\frac{p}{2}}}{1 - \gamma_\theta} \frac{3}{\delta_{\texttt{target}}}}, \eps_0 \right), \label{eq:FNO_result_proof23}
\end{align}
finishing the argument.

\section{Estimates for the constructed ReLU network approximating the double contraction}
\label{apx:multiplication}

The goal of this Appendix is to provide the mathematical arguments for the validity of Lemma \ref{lem:activations_assumptions}. For this purpose, we state and prove the following estimate which is the main technical tool used in this Appendix.
\begin{lem}\label{lem:activations_scalar_multiplication}
	For all numbers $x,y,c \in [-M,M]$, the following estimate holds:
	\begin{equation}\label{eq:activations_scalar_multiplicationNN_estimate}
		\left|\multiplication_\theta(c,x) - \multiplication_\theta(c,y) - c(x-y) \right| \leq 2M\,\delta |x-y|.
	\end{equation}
\end{lem}
\begin{proof}
	The exact multiplication of two real numbers is represented by the operator
	\begin{equation}\label{eq:activations_scalar_multiplication_proof1}
		\multiplication: \R \times \R \rightarrow \R, \quad m(a,b) = 2M^2 \left[ q\left( \frac{a+b}{2M} \right) - q\left( \frac{a}{2M}\right) - q\left( \frac{b}{2M} \right) \right].
	\end{equation}
	Both functions $q$ and $q_\theta$, defined in eq. \eqref{eq:activations_square} and \eqref{eq:activations_square_approximation}, are absolutely continuous, more precisely admit the representations
	\begin{equation}\label{eq:activations_scalar_multiplication_proof2}
		q_\theta(x) = \int_0^x q_\theta'(\xi)\, d\xi \quad \text{and} \quad q(x) = \int_0^x q'(\xi)\, d\xi
	\end{equation}
	with piece-wise constant functions $q', q_\theta':[-1,1] \rightarrow \R$. 
	To obtain the estimate \eqref{eq:activations_scalar_multiplicationNN_estimate}, we re-write the terms as follows
	\begin{equation}\label{eq:activations_scalar_multiplication_proof3}
		\begin{split}
			\multiplication_\theta(c,x) - \multiplication_\theta(c,y) - \multiplication(c,x) - \multiplication(c,y)
			&= 2M^2 \left[ 
			q_\theta\left( \frac{x+c}{2M} \right) - q_\theta\left( \frac{x}{2M}\right) - q_\theta\left( \frac{c}{2M} \right)	\right.\\
			&- q_\theta\left( \frac{y+c}{2M} \right) + q_\theta\left( \frac{y}{2M}\right) + q_\theta\left( \frac{c}{2M} \right)	\\
			&-q\left( \frac{x+c}{2M} \right) + q\left( \frac{x}{2M}\right) + q\left( \frac{c}{2M} \right)			\\
			&+\left.q\left( \frac{y+c}{2M} \right) - q\left( \frac{y}{2M}\right) - q\left( \frac{c}{2M} \right) \right]\\
			&= 2M^2 \left[ 
			q_\theta\left( \frac{x+c}{2M} \right) - q_\theta\left( \frac{y+c}{2M} \right) \right.\\
			& + q_\theta\left( \frac{y}{2M}\right) - q_\theta\left( \frac{x}{2M}\right) \\
			&-q\left( \frac{x+c}{2M} \right)  +q\left( \frac{y+c}{2M} \right)	\\
			&-\left.q\left( \frac{y}{2M}\right) + q\left( \frac{x}{2M}\right) \right].	
		\end{split}
	\end{equation}
	By the absolutely continuous representations \eqref{eq:activations_scalar_multiplication_proof2}, we notice the identities
	\begin{align}
			q\left( \frac{x+c}{2M} \right) - q\left( \frac{y+c}{2M} \right) &= \int^{\frac{x+c}{2M}}_{\frac{y+c}{2M}}q_\theta'(\xi)\, d\xi, \label{eq:activations_scalar_multiplication_proof4}\\
			q\left( \frac{x}{2M}\right) - q\left( \frac{y}{2M}\right) &= \int^{\frac{x}{2M}}_{\frac{y}{2M}}q_\theta'(\xi)\, d\xi, \label{eq:activations_scalar_multiplication_proof5}\\
			q\left( \frac{x+c}{2M} \right)  -q\left( \frac{y+c}{2M} \right)	&= \int^{\frac{x+c}{2M}}_{\frac{y+c}{2M}}q'(\xi)\, d\xi, \label{eq:activations_scalar_multiplication_proof6}\\
			q\left( \frac{x}{2M}\right) - q\left( \frac{y}{2M}\right) &= \int^{\frac{x}{2M}}_{\frac{y}{2M}}q'(\xi)\, d\xi. \label{eq:activations_scalar_multiplication_proof7}
	\end{align}
	Inserted into the expression \eqref{eq:activations_scalar_multiplication_proof3}, we thus observe
	\begin{equation}\label{eq:activations_scalar_multiplication_proof8}
		\multiplication_\theta(c,x) - \multiplication_\theta(c,y) - \multiplication(c,x) - \multiplication(c,y)
			= 2M^2 \left[ \int^{\frac{x+c}{2M}}_{\frac{y+c}{2M}}q_\theta'(\xi) - q'(\xi)\, d\xi - \int^{\frac{x}{2M}}_{\frac{y}{2M}}q_\theta'(\xi) - q'(\xi)\, d\xi \right],
	\end{equation}
	which by the bound \eqref{eq:activations_square_approximation_estimate} and the triangle inequality for integrals implies the estimate
	\begin{equation}\label{eq:activations_scalar_multiplication_proof9}
		\left|\multiplication_\theta(c,x) - \multiplication_\theta(c,y) - \multiplication(c,x) - \multiplication(c,y) \right|
			\leq 2M^2 \frac{2\delta|x-y|}{2M} = 2 \, M \, \delta |x-y|,
	\end{equation}
	as claimed.
\end{proof}
With the obtained estimate \eqref{eq:activations_scalar_multiplicationNN_estimate}, we are in the position to provide the arguments for the validity of Lemma \ref{lem:activations_assumptions}. More precisely, The construction \eqref{eq:activations_double_contraction_construction}--\eqref{eq:activations_double_contraction_construction2} satisfies the following properties:
\begin{enumerate}[(i)]
	\item 	The Lipschitz estimate
		\begin{equation}\label{eq:activations_assumptions_lem_Lipschitz_apx}
			\left\| \Multiplication_\theta(\mathds{T},\feps_1) - \Multiplication_\theta(\mathds{T},\feps_2) \right\| \leq \left( \|\mathds{T}\| + M\delta d(d+1)\, \max_{i,j}|T_{ij}| \right) \| \feps_1 - \feps_2\|
		\end{equation}
		holds for all strains $\feps_1, \feps_2 \in \Sym{d}$ and all tensors $\mathds{T} \in L(\Sym{d})$ obeying the constraint
		\begin{equation}\label{eq:activations_assumptions_lem_boundedTensors_apx}
			\| \mathds{T} \| \leq M.
		\end{equation}
	\item The estimate
		\begin{equation}\label{eq:activations_assumptions_lem_multiplication_proximity_small_apx}
			\left\| \Multiplication_\theta(\mathds{T}, \feps) - \mathds{T}:\feps \right\| \leq M\delta d(d+1)\, \max_{i,j}|T_{ij}| \|\feps\|
		\end{equation}
		holds for all tensors $\feps \in \Sym{d}$ and all tensors $\mathds{T} \in L(\Sym{d})$ under the constraints
		\begin{equation}\label{eq:activations_assumptions_lem_boundedStrainsTensors_apx}
			 \| \mathds{T} \| \leq M \quad \text{and} \quad \|\feps\| \leq M.
		\end{equation}
	\item 	The bound 
		\begin{equation}\label{eq:activations_assumptions_lem_multiplication_proximity_large_apx}
			\left\| \Multiplication_\theta(\mathds{T},\feps) - \mathds{T}:\feps \right\| \leq (1 + \delta M d (d+1)) \|\mathds{T}\| \, \|\feps\|
		\end{equation}
		holds for all tensors $\mathds{T} \in L(\Sym{d})$ and $\feps \in \Sym{d}$ satisfying the constraints
			\begin{equation}\label{eq:activations_assumptions_lem_boundedTensors_unbounded strains_apx}
				\| \mathds{T} \| \leq M \quad \text{and} \quad \| \feps\| \geq M.
			\end{equation}
\end{enumerate}
We provide the arguments in order:
\begin{enumerate}[(i)]
	\item 	We start with the estimate
	\begin{equation}\label{eq:activations_assumptions_lem_proof_part1_1}
		\left\| \Multiplication_\theta(\mathds{T},\feps_1) - \Multiplication_\theta(\mathds{T},\feps_2) \right\| \leq
		\left\| \Multiplication_\theta(\mathds{T},\feps_1) - \Multiplication_\theta(\mathds{T},\feps_2) - \mathds{T}:r_\theta(\feps_1) + \mathds{T}:r_\theta(\feps_2) \right\|
		+ \| \mathds{T}:(r_\theta(\feps_1) - r_\theta(\feps_2))\|,
	\end{equation}
	where the restriction \eqref{eq:activations_ridge_function} is applied component-wise. The second term may be estimated via
	\begin{equation}\label{eq:activations_assumptions_lem_proof_part1_2}
		\| \mathds{T}:(r_\theta(\feps_1) - r_\theta(\feps_2))\| \leq \| \mathds{T} \| \|(r_\theta(\feps_1) - r_\theta(\feps_2))\| \leq \| \mathds{T} \| \|\feps_1 - \feps_2\|,
	\end{equation}
	where we used the non-expansivity of the ridge function \eqref{eq:activations_ridge_function}, responsible for the validity of the estimate
	\begin{equation}\label{eq:activations_assumptions_lem_proof_part1_3}
		\| r_\theta(\feps_1) - r_\theta(\feps_2)\| \leq \|\feps_1 - \feps_2\|.
	\end{equation}
	The first term on the right-hand side requires a bit more work. Let us define the quantities
	\begin{equation}\label{eq:activations_assumptions_lem_proof_part1_4}
		c_{ij} = \multiplication_\theta(T_{ij}, r_\theta(\eps_{1,j})) - \multiplication_\theta(T_{ij}, r_\theta(\eps_{2,j})) - T_{ij}r_\theta(\eps_{1,j}) + T_{ij}r_\theta(\eps_{2,j})
	\end{equation}
	for $i,j=1,2,\ldots,K = d (d+1) / 2$, which satisfy, by ineq. \eqref{eq:activations_scalar_multiplicationNN_estimate}, the estimates
	\begin{equation}\label{eq:activations_assumptions_lem_proof_part1_5}
		|c_{ij}| \leq 2M\delta |T_{ij}| |r_\theta(\eps_{1,j}) - r_\theta(\eps_{2,j})|.
	\end{equation}
	Here, we used that both the components $|T_{ij}|$ and the components of the strains $\underline{\eps}_1$ as well as $\underline{\eps}_2$ do not exceed the threshold $M$. Then, we may estimate
	\begin{equation}\label{eq:activations_assumptions_lem_proof_part1_6}
		\begin{split}
			& \phantom{=}\, \left\| \Multiplication_\theta(\mathds{T},\feps_1) - \Multiplication_\theta(\mathds{T},\feps_2) - \mathds{T}:r_\theta(\feps_1) + \mathds{T}:r_\theta(\feps_2) \right\|\\
			& {=} \, \left( \sum_{i=1}^K \left( \sum_{j=1}^K c_{ij} \right)^2\right)^{\frac{1}{2}} \leq \left( \sum_{i=1}^K \left( \sum_{j=1}^K |c_{ij}| \right)^2\right)^{\frac{1}{2}}\\
			& {\leq} \, \left( \sum_{i=1}^K \left( \sum_{j=1}^K 2M\delta |T_{ij}| |r_\theta(\eps_{1,j}) - r_\theta(\eps_{2,j})| \right)^2\right)^{\frac{1}{2}}\\
			& {\leq} \, 2M\delta \, \max_{i,j}|T_{ij}| \left( \sum_{i=1}^K \left( \sum_{j=1}^K |r_\theta(\eps_{1,j}) - r_\theta(\eps_{2,j})| \right)^2\right)^{\frac{1}{2}}\\
			& {\leq} \, 2M\delta \, \max_{i,j}|T_{ij}|K \left( |r_\theta(\eps_{1,j}) - r_\theta(\eps_{2,j})|^2\right)^{\frac{1}{2}},
		\end{split}
	\end{equation}
	where we used the Cauchy-Schwarz inequality in $K$ dimensions for the last step. Thus, we obtain the bound
	\begin{equation}\label{eq:activations_assumptions_lem_proof_part1_7}
		\left\| \Multiplication_\theta(\mathds{T},\feps_1) - \Multiplication_\theta(\mathds{T},\feps_2) - \mathds{T}:r_\theta(\feps_1) + \mathds{T}:r_\theta(\feps_2) \right\| \leq 2M\delta \, \max_{i,j}|T_{ij}|K \| r_\theta(\feps_1) - r_\theta(\feps_2) \|. 
	\end{equation}
	Utilizing the estimate \eqref{eq:activations_assumptions_lem_proof_part1_3} and the previously established result \eqref{eq:activations_assumptions_lem_proof_part1_2} proves the claimed inequality \eqref{eq:activations_assumptions_lem_Lipschitz_apx}.
	\item We take a look at the components
	\begin{equation}\label{eq:activations_assumptions_lem_proof_part2_1}
		c_{ij} = \multiplication_\theta(T_{ij},\eps_j) - T_{ij}\eps_j, \quad i,j=1,2,\ldots, K = \frac{d(d+1)}{2}
	\end{equation}
	Invoking the inequality \eqref{eq:activations_scalar_multiplicationNN_estimate} for $c={T_{ij}}$, $x = \eps_j$ and $y = 0$, we obtain the estimate
	\begin{equation}\label{eq:activations_assumptions_lem_proof_part2_2}
		|c_{ij}| \leq 2M\delta |T_{ij}| |\eps_j|.
	\end{equation}
	Here, we used that the numbers $c$, $\eps_j$ and $y$ do not exceed the threshold $M$, and we also relied upon the fact
	\begin{equation}\label{eq:activations_assumptions_lem_proof_part2_3}
		\multiplication_\theta(x,0) = 0, \quad \text{valid for all} \quad x \in \R,
	\end{equation}
	implied by the construction \eqref{eq:activations_square_approximation_origin}.
	Arguing similarly to the bounds \eqref{eq:activations_assumptions_lem_proof_part1_6} and \eqref{eq:activations_assumptions_lem_proof_part1_7}, we arrive at the desired estimate \eqref{eq:activations_assumptions_lem_multiplication_proximity_small_apx}.
	\item The key idea is to use the triangle inequality
	\begin{equation}\label{eq:activations_assumptions_lem_proof_part3_1}
		\left\| \Multiplication_\theta(\mathds{T},\feps) - \mathds{T}:\feps \right\| \leq 
		\left\| \Multiplication_\theta(\mathds{T},\feps) - \mathds{T}:r_{\theta}(\feps) \right\| + \left\| \mathds{T}:r_{\theta}(\feps) - \mathds{T}:\feps \right\|.
	\end{equation}
	For the first term, similar arguments as for establishing the Lipschitz estimate \eqref{eq:activations_assumptions_lem_Lipschitz} permit us to conclude the bound
	\begin{equation}\label{eq:activations_assumptions_lem_proof_part3_2}
		\left\| \Multiplication_\theta(\mathds{T},\feps) - \mathds{T}:r_{\theta}(\feps) \right\| \leq M\delta d(d+1)\, \max_{i,j}|T_{ij}| \| r_\theta(\feps) \|.
	\end{equation}
	By the elementary estimates
	\begin{equation}\label{eq:activations_assumptions_lem_proof_part3_3}
		\max_{i,j}|T_{ij}| \leq \|\mathds{T}\|
	\end{equation}
	and
	\begin{equation}\label{eq:activations_assumptions_lem_proof_part3_4}
		\| r_\theta(\feps) \| \leq \| \feps \|,
	\end{equation}
	we are led to the inequality
	\begin{equation}\label{eq:activations_assumptions_lem_proof_part3_5}
		\left\| \Multiplication_\theta(\mathds{T},\feps) - \mathds{T}:r_{\theta}(\feps) \right\| \leq M\delta d(d+1)\, \|\mathds{T}\| \| \feps \|.
	\end{equation}
	This insight, together with the simple bound
	\begin{equation}\label{eq:activations_assumptions_lem_proof_part3_6}
		\left\| \mathds{T}:r_{\theta}(\feps) - \mathds{T}:\feps \right\| = \left\| \mathds{T}:\left[r_{\theta}(\feps) - \feps\right] \right\| \leq  \| \mathds{T}\| \left\|r_{\theta}(\feps) - \feps \right\| \leq \| \mathds{T}\| \left\| \feps \right\|
	\end{equation}
	imply the claim.
\end{enumerate}	

\bibliographystyle{ieeetr}
{\footnotesize \bibliography{literature}}

\end{document}